\documentclass[12pt]{article}
\usepackage[utf8]{inputenc}

\usepackage{graphicx}

\usepackage{amssymb,amsmath,amsthm,amscd}
\usepackage{url}
\usepackage{xspace}
\usepackage{enumerate}

\usepackage[dvipsnames,usenames]{xcolor}
\newif\ifcolored
\coloredtrue
\ifcolored
\colorlet{Ngrn}{green!60!black}
\colorlet{Lred}{red!60!black}
\colorlet{gblu}{black!70!blue!35}
\colorlet{ggrn}{black!70!green!35}
\colorlet{gred}{black!70!red!35}
\colorlet{gorg}{black!70!orange!35}
\def\cN{\color{Ngrn}N}
\def\cL{\color{Lred}L}
\else
\colorlet{Ngrn}{black}
\colorlet{Lred}{black}
\colorlet{gblu}{black!40}
\colorlet{ggrn}{black!40}
\colorlet{gred}{black!40}
\colorlet{gorg}{black!40}
\def\cN{N}
\def\cL{L}
\fi
\usepackage{tikz}
\usetikzlibrary{fit}
\usetikzlibrary{shapes}
\usetikzlibrary{backgrounds}
\usetikzlibrary{patterns}
\usetikzlibrary{calc}
\usetikzlibrary{graphs}
\usetikzlibrary{matrix}
\usetikzlibrary{arrows}
\usetikzlibrary{decorations.pathmorphing}
\tikzset{
  every picture/.style={>=stealth'},
  zz/.style={decorate,decoration={zigzag,amplitude=0.7pt}},
  bps/.style={decorate,decoration={bumps,amplitude=1.4pt,
    segment length=9pt}},
  sw/.style={decorate,decoration={saw,amplitude=1.4pt,
    segment length=11pt}},
  brace/.style={decorate,decoration=brace},
  semithick,
  nolabel/.style={label/.code={}},
  graphs/declare={gap}{__{}[draw=none,fill=none] },
  graphs/declare={gp2}{__{}[draw=none,fill=none]  -!-  __{}
    [draw=none,fill=none] },
  graphs/declare={gp3}{__{}[draw=none,fill=none]  -!-  __{}
    [draw=none,fill=none] -!-  __{}[draw=none,fill=none] }
}

\newcommand{\alsolabel}[2][]{
  \foreach \pos/\where/\lab in {#2} {
    \scoped[label distance=1.2mm,#1,]%
    \node[rectangle,fill=none,draw=none,%
      label={[rectangle]\where:{$%
      \pgfkeysvalueof{/mnn/label prefix}%
      \lab
      \pgfkeysvalueof{/mnn/label suffix}%
      $}},at={(\pos)}] (\lab) {};%
  };%
}
\pgfkeys{
  /mnn/label prefix/.initial={},
  /mnn/label suffix/.initial={}
}

\makeatletter
\newcommand{\mathnamenodebasic}[2][]{%
  \node[circle,fill,inner sep=2pt,#1,
  ] (#2) {};%
  \@ifundefined{mnn@savelabel}{%
    \xdef\mnn@savelabel{#2/#1}%
  }{%
    \xdef\mnn@savelabel{\mnn@savelabel,#2/#1}%
  }%
}
\newcommand{\matrixgraphfinal}[1][]{%
  \foreach \lab/\opts in \mnn@savelabel {
    \scoped[label distance=1.2mm,#1,\opts]%
    \node[\opts,rectangle,fill=none,draw=none,%
      label={[rectangle]{$%
      \pgfkeysvalueof{/mnn/label prefix}%
      \lab
      \pgfkeysvalueof{/mnn/label suffix}%
      $}},at={(\lab)}] {};%
  };%
  \global\let\mnn@savelabel\relax
}
\newcommand{\mnn@exec}[1][]{\mathnamenodebasic[#1]{\mnn@name}}
\def\mnn@gathername#1{%
  \edef\mnn@name{\mnn@name#1}\mnn@work%
}

\def\mnn@work{%
  \@ifnextchar\pgfmatrixnextcell{%
    \mnn@exec
  }{%
  \@ifnextchar\pgfmatrixendrow{%
      \mnn@exec
  }{%
      \@ifnextchar[{%
      \mnn@exec
      }{%
      \mnn@gathername%
      }%
    }%
  }%
}
\newcommand{\mathnamenode}{\def\mnn@name{}\mnn@work}
\makeatother

\makeatletter
\newcommand{\matrixgraph}{
  \begingroup
  \catcode`\&=13%
  \matrixgraph@work
}
\newcommand{\matrixgraph@work}[3][]{
  \matrix[fill=none,draw=none,rectangle,inner sep=0pt,%
    matrix anchor=north west,#1,%
    execute at begin cell=\mathnamenode] {#2};%
  \endgroup%
  \graph[use existing nodes,%
    default edge operator=complete bipartite]{#3};%
  \matrixgraphfinal[#1]%
}
\makeatother
\newcommand{\legendrow}[3][]{
  \draw[line width=1.5pt,#1,#2] (0,0)--(0.6,0); \&
  \node[rectangle,black,draw=none,fill=none,anchor=west,#1] {#3};\\
}
\newcommand{\legend}[2][]{
  \matrix%
  [rectangle,draw=none,fill=none,row sep=1.5pt,column sep=1pt,
  ampersand replacement=\&,#1]%
  (legend){#2};
}

\date{}

\theoremstyle{plain}
\newtheorem{theorem}{Theorem}[section]
\newtheorem{lemma}[theorem]{Lemma}
\newtheorem{proposition}[theorem]{Proposition}
\newtheorem{corollary}[theorem]{Corollary}
\theoremstyle{definition} 
\newtheorem{example}[theorem]{Example}

\newtheorem{remark}[theorem]{Remark}

\newcommand{\Case}[2]{\smallskip\par{\it Case #1:\/ #2}}

\newcounter{claim}
\renewcommand{\theclaim}{\Alph{claim}}
\newenvironment{claim}{\refstepcounter{claim}%
\par\medskip\par\noindent{\it Claim~\theclaim.~}~\rm}%
{\par\smallskip\par}
\newenvironment{subproof}{\par\noindent{\sl Proof of Claim~\theclaim.~}}%
{$\,\triangleleft$\par\medskip\par}

\newcommand{\refeq}[1]{(\ref{eq:#1})}

\newcommand{\of}[1]{\left( #1 \right)}
\newcommand{\Set}[1]{\left\{\, #1 \,\right\}}
\newcommand{\setdef}[2]{\left\{ \hspace{0.5mm} #1 : \hspace{0.5mm} #2 \right\}}
\newcommand{\msetdef}[2]{\left\{\!\!\left\{ \hspace{0.5mm} #1 : \hspace{0.5mm} #2 \right\}\!\!\right\}}
\newcommand{\mset}[1]{\left\{\!\!\left\{ \hspace{0.5mm} #1 \hspace{0.5mm} \right\}\!\!\right\}}

\newcommand{\function}[2]{:#1 \rightarrow #2}

\newcommand{\complex}{\mathbb{C}}
\newcommand{\integers}{\mathbb{Z}}
\newcommand{\WL}[1]{\ensuremath{#1\text{-}\mathrm{WL}}\xspace}
\newcommand{\kWL}{\WL k}
\newcommand{\wl}{\WL2}

\newcommand{\eqq}{\equiv_{2\text{-}\mathrm{WL}}}

\newcommand{\ciso}{\cong_{\mathit{comb}}}
\newcommand{\aiso}{\cong_{\mathit{alg}}}

\newcommand{\calA}{{\mathcal A}}

\newcommand{\calP}{{\mathcal P}}
\newcommand{\calQ}{{\mathcal Q}}
\newcommand{\calR}{{\mathcal R}}
\newcommand{\calX}{{\mathcal X}}

\newcommand{\ccc}{\ensuremath{{\mathcal C}}\xspace}
\newcommand{\ccd}{\ensuremath{{\mathcal D}}\xspace}
\newcommand{\cct}{\ensuremath{{\mathcal T}}\xspace}
\newcommand{\ccg}[1]{\ensuremath{{\mathcal C}(#1)}\xspace}
\newcommand{\fg}[1]{F_{#1}}
\newcommand{\dcc}[1]{D_{#1}}
\newcommand{\saa}[1]{\mathbb{A}(#1)}
\newcommand{\saai}[1]{\mathbb{A}^*(#1)}
\newcommand{\sca}[1]{\mathbb{C}(#1)}
\newcommand{\scac}[1]{\mathbb{C}_0(#1)}
\newcommand{\fa}{\mathrm{Fano}}
\newcommand{\mk}{\mathrm{MK}}

\newcommand{\colclass}[2][0]{
\makeatletter
\xdef\clcl@verts{}
\foreach \v in {#2}{\xdef\clcl@verts{\clcl@verts(\v)}}
\begin{scope}[on background layer]
\node[fill=black!20,rotate fit=#1,fit/.expand once={\clcl@verts},inner sep=4pt,rectangle,rounded corners=8pt,draw=none] {};
\end{scope}
\makeatother
}

\newcommand{\colclasstrpz}[4]{
\begin{scope}[on background layer]
\fill[fill=black!20,rounded corners] 
([shift=(160:6mm)] #1.north west) --
([shift=(20:6mm)]  #2.north east) --
([shift=(340:6mm)] #3.south east) --
([shift=(200:6mm)] #4.south west) -- cycle;
\end{scope}
}

\title{Identifiability of Graphs with Small Color Classes\\
by the Weisfeiler-Leman Algorithm}
\author{Frank Fuhlbrück, Johannes Köbler, and Oleg Verbitsky%
\thanks{Supported by DFG grant KO 1053/8--1. On leave from the IAPMM, Lviv, Ukraine.}\\[4mm]
\normalsize
Institut für Informatik, Humboldt-Universität zu Berlin}

\date{}

\begin{document}

\maketitle

\begin{abstract}
As it is well known, the isomorphism problem for vertex-colored graphs 
with color multiplicity at most 3 is solvable by the classical 2-dimensional
Weisfeiler-Leman algorithm (\wl). On the other hand, the prominent Cai-Fürer-Immerman construction
shows that even the multidimensional version of the algorithm
does not suffice for graphs with color multiplicity 4.
We give an efficient decision procedure that, given a graph $G$ of color multiplicity 4,
recognizes whether or not $G$ is identifiable by \wl, that is, whether or not \wl
distinguishes $G$ from any non-isomorphic graph.
In fact, we solve the much more
general problem of recognizing whether or not a given coherent configuration
of maximum fiber size 4 is separable. This extends our recognition
algorithm to graphs of color multiplicity 4 with directed and colored edges.

Our decision procedure is based on an explicit description of the class of
graphs with color multiplicity 4 that are \emph{not} identifiable by \wl.
The Cai-Fürer-Immerman graphs of color multiplicity 4 distinctly appear here as a natural
subclass, which demonstrates that the Cai-Fürer-Immerman construction is not ad hoc.
Our classification reveals also other types of graphs that are hard for \wl.
One of them arises from patterns known as \emph{$(n_3)$-configurations}
in incidence geometry.
\end{abstract}

\tableofcontents

\section{Introduction}\label{s:intro}

Over 50 years ago Weisfeiler and Leman \cite{WLe68} described a natural
combinatorial procedure that since then constantly plays a significant role
in the research on the graph isomorphism problem. The procedure is now
most often referred to as the 2-dimensional Weisfeiler-Leman algorithm (\WL2).
It generalizes and improves the classical color refinement method  (\WL1)
and has an even more powerful $k$-dimensional version (\WL k) for any $k>2$.
The original 2-dimensional version and the logarithmic-dimensional enhancement
are important components in Babai's quasipolynomial-time isomorphism algorithm~\cite{Babai16}.

Even on its own, \WL2 is a quite powerful tool in isomorphism testing.
For instance, it solves the isomorphism problem for several important graph classes, 
in particular, for interval graphs \cite{EvdokimovPT00}.
Also, it is successful for almost all regular graphs of a fixed degree \cite{Bollobas82}.
On the other hand, not every pair of non-isomorphic graphs is distinguishable by \WL2.
For example, it cannot detect any difference between two non-isomorphic
strongly regular graphs with the same parameters.

We call a graph $G$ \emph{amenable to \WL k} if the algorithm distinguishes $G$
from any non-isomorphic graph. An efficient characterization of the class of graphs
amenable to \WL1 is obtained by Arvind et al.\ in \cite{ArvindKRV17},
where it is given also for vertex-colored graphs. Independently, Kiefer et al.~\cite{KSS15}
give an efficient criterion of amenability to \WL 1 in a more general framework
including also directed graphs with colored edges.
Similar results for \WL2 are currently out of reach, even for undirected uncolored graphs.
A stumbling block here is the lack of understanding which strongly regular graphs are 
uniquely determined by their parameters. Note that a strongly regular graph is 
determined by its parameters up to isomorphism if and only if it is amenable to~\WL2.

A general strategy to approach a hard problem is to examine its
complexity in the parameterized setting.
We consider vertex-colored graphs with
the \emph{color multiplicity}, that is, the maximum number of equally colored
vertices, as parameter.
If this parameter is bounded, the graph isomorphism problem is known to be efficiently solvable.
More specifically, it is solvable
in time polynomial in the number of vertices and quasipolynomial in the color multiplicity
\cite[Corollary 4]{Babai16}, and it is solvable in polylogarithmic parallel time \cite{Luks86}.
Graph Isomorphism is known to be in the $\mathrm{Mod}_k\mathrm{L}$ hierarchy for any fixed 
color multiplicity \cite{ArvindKV05}, and even in the class $\oplus\mathrm{L}=\mathrm{Mod}_2\mathrm{L}$ 
for color multiplicity 4 and 5; see~\cite{ArvindK06}.
Recall that $\mathrm{Mod}_k\mathrm{L}$ is the class of decision problems solvable
non-deterministically in logspace in the sense that the answer is ``no'' if and only if
the number of accepting paths is divisible by~$k$.

Every graph of color multiplicity at most 3 is amenable to \WL2
(Immerman and Lander \cite{ImmermanL90}).  Starting from the color
multiplicity 4, the amenability concept is non-trivial: The prominent
Cai-Fürer-Immerman construction \cite{CaiFI92} shows that for any $k$,
there exist graphs with color multiplicity 4 that are not amenable to
\WL k.

We design an efficient decision procedure that, given a graph $G$ with color multiplicity 4,
recognizes whether or not $G$ is amenable to \WL2.
Note that an a priori upper complexity bound for this decision problem is $\mathrm{coNP}$,
as a consequence of the aforementioned fact that Graph Isomorphism for graphs of bounded
color multiplicity is in~$\mathrm{P}$.
From now on, amenability is meant with respect to \WL2, unless stated otherwise.

We actually solve a much more general problem.
\WL2 transforms an input graph $G$, possibly with colored vertices and directed and
colored edges, into a
\emph{coherent configuration} $\ccg G$, which is called the \emph{coherent closure} of $G$.
The concept of a coherent configuration has been discovered independently in statistics
\cite{BoseM59} and algebra \cite{Higman64} and, playing an important role
in diverse areas, has been developed to the subject of a rich theory; see
a recent monograph \cite{Ponomarenko-book}, that we will use in this paper as
a reference book. A coherent configuration \ccc is called \emph{separable}
if the isomorphism type of \ccc is determined
by its regularity parameters in a certain strong sense; see the definition in Section \ref{s:basics}.
The separability of the coherent closure \ccg G
implies the amenability of the graph $G$. This was the approach undertaken
in \cite{EvdokimovPT00}, where it was shown that
the coherent closure of any  interval graph is separable.
Somewhat less obviously, the converse relation between amenability of $G$ and
separability of \ccg G is also true: For every graph $G$,
\begin{equation}
  \label{eq:sepamen}
G\text{ is amenable if and only if } \ccg G\text{ is separable;} 
\end{equation}
see Theorem \ref{thm:reduction} in Section \ref{s:basics}.
Equivalence \refeq{sepamen} reduces the amenability problem for graphs
to the separability problem for coherent configurations. This reduction
works as well for directed graphs with colored vertices and colored edges,
that is, essentially for general binary relational structures.
If $G$ has color multiplicity $b$, then the maximum \emph{fiber} size of \ccg G
is also bounded by $b$. While all coherent configurations with fibers of size at most 3
are known to be separable~\cite{Ponomarenko-book}, the separability property for
coherent configurations with fibers of size 4 is non-trivial, and our first result
is this.

\begin{theorem}\label{thm:sep4}
The problem of deciding whether a given coherent configuration with maximum fiber size 4
is separable is solvable in $\oplus\mathrm{L}$.
\end{theorem}

Since $\oplus\mathrm{L}\subseteq\mathrm{NC}^2$
(which follows from the inclusion $\#L\subseteq\mathrm{NC}^2$ in \cite{AlvarezJ93}),
Theorem \ref{thm:sep4} implies that the separability problem is solvable
in parallel polylogarithmic time. Using the reduction \refeq{sepamen},
we obtain our result for graphs.

\begin{theorem}\label{thm:amen4}
The problem of deciding whether a given vertex-colored graph of color multiplicity 4 
is amenable to \WL2 is solvable in $\mathrm{P}$.
This holds true also for vertex- and edge-colored directed graphs.
\end{theorem}

More precisely, the proof of Theorem \ref{thm:amen4} yields
an algorithm deciding amenability of graphs of color multiplicity at most 4 
with running time $O(n^{2+\omega})$, where $\omega<2.373$ is the exponent of fast matrix multiplication \cite{Gall14}.
Using randomization, the running time can be improved to~$O(n^{4}\log^2 n)$. 

Our results have the following consequences.

\paragraph{\textit{Highlighting the inherent structure of the Cai-Fürer-Immerman graphs.}}
The essence of our proof of Theorem \ref{thm:amen4} is an explicit description of the class of
graphs with color multiplicity 4 that are \emph{not} amenable to \wl.
The Cai-Fürer-Immerman graphs of color multiplicity 4 distinctly appear here as a natural
subclass, which demonstrates that the Cai-Fürer-Immerman construction is not ad hoc.
In a sense, the famous CFI gadget \cite[Fig.~3]{CaiFI92} (or \cite[Fig.~13.24]{dcbook})
appears in our analysis inevitably ``by itself''.
More precisely, this concerns a simplified version of the CFI gadget, where 
each vertex in a cubic pattern graph is replaced with a quadruple of new vertices
and two quadruples are connected by edges directly, and not via two extra pairs of
auxiliary vertices as in the original version; cf.\ Figure~\ref{fig:3Intsp3Match}. 
The simplified gadget appears in
an algebraic analog of the CFI result by Evdokimov and Ponomarenko \cite{EvdokimovP99};
see also Fürer's survey paper \cite{Fuerer17}. This gadget comes out also
in the \emph{shrunken} multipede graphs~\cite{NeuenS17} (we discuss the multipede graphs below).
A transformation of the original CFI gadget
into the simplified one is easy to retrace using the framework of coherent configuration;
see Section \ref{s:excl2points} where it is shown that the auxiliary vertex pairs can be cut down in $\ccg G$
without affecting the separability property.

\paragraph{\textit{Relevance to  multipede graphs.}}
While the CFI graphs have many automorphisms, Gurevich and Shelah \cite{GurevichS96}
came up with a construction of (non-binary) \emph{multipede} structures
that are rigid and yet not identifiable by \WL k. Neuen and Schweitzer \cite{NeuenS17,NeuenS18}
combined both approaches to construct \emph{multipede graphs} and to give sufficient
conditions ensuring that these graphs are not amenable to \WL k
(see also a recent related paper \cite{DawarK19}). The multipede graphs
are vertex-colored and the results of \cite{NeuenS17,NeuenS18} make perfect sense
if the color multiplicity is bounded by~4. 
We observe a close connection between such multipede graphs and the class of \emph{irredundant}
coherent configurations playing a key role in the proof of Theorem \ref{thm:sep4}.
An irredundant coherent configuration
typically admits a natural representation by a multipede graph and vice versa;
see Remark \ref{rem:multipede}.
Though non-amenability to \WL k for higher dimensions implies non-amenability to \WL2,
the results obtained in \cite{DawarK19,NeuenS17,NeuenS18} and in our paper are
incomparable as we provide both sufficient and \emph{necessary} conditions 
for \WL2-non-amenability.

\paragraph{\textit{More graphs hard for 2-WL.}}
Our analysis reveals new types of non-amenable graphs. 
A particularly elegant construction is based on 
the well-studied $(n_3)$-configurations of lines and points \cite{Gruenbaum,PisanskiS}. 
For example, the 7-point Fano plane and the 9-point Pappus configuration give rise to 
non-amenable graphs of color multiplicity 4 with, respectively, 28 and 36 vertices.

\paragraph{\textit{Classification of small graphs.}}
Our amenability criteria are easy to apply in many cases. In particular, they
imply that all graphs of color multiplicity 4 with no more than 15 vertices are
amenable. Among graphs of color multiplicity 4 with 16 vertices there are 436 non-amenable
graphs. They are split into 218 pairs of \wl-indistinguishable non-isomorphic graphs,
where a typical instance is the pair consisting of vertex-colored Shrikhande and
$4\times4$ rook's graphs, which are known as the smallest pair of strongly regular graphs with 
the same parameters.

\paragraph{\textit{Small coherent configurations.}}
The corresponding fact about coherent configurations is that all of them 
with 15 or fewer vertices and fiber size 4 are separable. 
This result can be obtained from our cut-down lemmas in Sections \ref{s:excl-matching}--\ref{s:excl-C8}
and the known fact that all quasiregular coherent configurations with at most 3 fibers
are separable \cite{HirasakaKP18}, but we provide a self-contained proof; see Theorem~\ref{thm:cc16}.
Moreover, on 16 vertices there is a unique, up to isomorphism, non-separable coherent configuration
with fiber size 4. Note that 
all coherent configuration with at most 15 points have been enumerated
\cite{KlinZ17}, but their separability analysis seems to be still missing.\footnote{%
Coherent configurations on 16 vertices were studied in~\cite{KlinMR06}.}

\paragraph{\textit{Definability of a graph in 3-variable logic.}}
A graph $G$ is \emph{definable} in a logic $\mathcal L$ if $\mathcal L$ contains a sentence
$\Phi$ that is true on $G$ and false on any graph $H$ non-isomorphic to $G$.
It is well known \cite{CaiFI92} that $G$ is definable in $(k+1)$-variable first-order logic
with counting quantifiers if and only if $G$ is amenable to \WL k.
The aforementioned result by Immerman and Lander \cite{ImmermanL90} actually says that
every graph of color multiplicity at most 3 is definable in 3-variable logic, even
without counting quantifiers. Our Theorem \ref{thm:amen4} can be recast as follows:
It can be decided in polynomial time whether a given graph of color multiplicity 4
is definable in the counting 3-variable logic.

\subsection*{Structure of the paper}

Our formal framework is presented in Section \ref{s:basics},
where we give basic facts about coherent configurations and use them
to prove equivalence \refeq{sepamen}. 

Formally, a coherent configuration \ccc on a point set $V$ 
is a partition of the Cartesian square $V^2$. Elements of the partition are called \emph{basis relations} 
of \ccc. The reflexive basis relations determine a partition of $V$ into \emph{fibers}
$X_1,\ldots,X_s$. A \emph{cell} $\ccc[X_i]$ of \ccc (or a \emph{homogeneous component} of \ccc)
is formed by the basis relations that are defined on $X_i$. An \emph{interspace} $\ccc[X_i,X_j]$
is formed by the basis relations between $X_i$ and $X_j$. A coherent configuration has the property
that every basis relation belongs either to a cell or to an interspace.

In Section \ref{s:prel} we explore the local structure of a coherent configuration \ccc
under the condition that $|X_i|\le4$ for every fiber of \ccc. We call an interspace
$\ccc[X_i,X_j]$ \emph{uniform} if it contains a single basis relation $X_i\times X_j$.
We observe that a non-uniform interspace $\ccc[X_i,X_j]$ between 4-point fibers $X_i$ and $X_j$
contains either a matching relation between $X_i$ and $X_j$ or a relation whose underlying
undirected graph is an 8-cycle or the union of two 4-cycles. In the last case
we say that $\ccc[X_i,X_j]$ is an interspace of \emph{type $2K_{2,2}$}.
It is known \cite{Ponomarenko-book} (see also Subsection \ref{ss:direct}) that
it suffices to solve the separability problem for coherent configurations that are
indecomposable in a direct sum of smaller configurations. As a consequence, we can
assume in our analysis that every fiber $X_i$ consists of either 4 or 2 points;
see Sections \ref{ss:f-i}--\ref{ss:direct} for details.

In Sections \ref{s:excl-matching}--\ref{s:excl-C8} we prove three cut-down lemmas:
\begin{itemize}
\item 
If an interspace $\ccc[X_i,X_j]$ contains a matching relation, then
removal of the fiber $X_i$ from \ccc does not affect the separability property.
\item 
Furthermore, all fibers of size 2 can be removed without affecting the separability property.
\item 
Finally, all pairs of fibers $X_i,X_j$ such that $\ccc[X_i,X_j]$ contains an 8-cycle
can be removed without affecting the separability property.
\end{itemize}
The first cut-down lemma, allowing elimination of matching basis relations,
is proved in the general case, with no assumptions on the coherent configuration \ccc.
One direction, namely the non-separability of the reduced version of a non-separable
configuration \ccc, was known earlier due to Evdokimov and Ponomarenko~\cite{EvdokimovP02}.

The cut-down lemmas reduce our task to consideration of indecomposable coherent configurations \ccc
with all fibers of size 4 and all non-uniform interspaces of type $2K_{2,2}$.
We call such coherent configurations \emph{irredundant}. This class is close to the
\emph{reduced Klein configurations} studied in \cite[Section 4.1.2]{Ponomarenko-book}.
The \emph{fiber graph} $\fg\ccc$ has the fibers of \ccc as vertices, and two fibers
$X_i$ and $X_j$ are adjacent in $\fg\ccc$ if the interspace $\ccc[X_i,X_j]$ is non-uniform.
Like the reduced Klein configurations, the structure of an irredundant configuration \ccc
determines a clique partition $\dcc\ccc$ of $\fg\ccc$ such that the cliques and the fibers form 
a line-point incidence structure known as \emph{partial linear spaces} (see \cite{Bruyn16,Metsch91}),
where every point (fiber) is incident to at most 3 lines (cliques in~$\dcc\ccc$).

The case when all cliques in $\dcc\ccc$ have size 2 corresponds to the 
Cai-Fürer-Immerman construction. Though coherent configurations of this kind
are well studied (Evdokimov and Ponomarenko \cite{EvdokimovP99}, see also \cite[Section 4.1.3]{Ponomarenko-book}),
we consider them in Section \ref{ss:CFI} for expository purposes as the simplest case
of irredundant configurations. Another instructive particular case, when all cliques in $\dcc\ccc$ have size 3,
is called \emph{3-harmonious} and considered in Section \ref{ss:3-reg}. The underlying
partial linear spaces of such coherent configurations are the well-studied geometric
$(n_3)$-configurations~\cite{Gruenbaum,PisanskiS}.

After this case study, we consider the general irredundant configurations in Section \ref{ss:gen-case}.
Note that the standard isomorphism concept
of coherent configurations is called \emph{combinatorial isomorphism}, while
the equivalence with respect to the regularity parameters is captured
by the concept of \emph{algebraic isomorphism}.
Deciding separability of an irredundant configuration \ccc, we actually have to check whether every
algebraic isomorphism from \ccc to another coherent configuration $\ccc'$ is induced by
a combinatorial isomorphism. The first observation (made in Section \ref{s:2K22}), which makes our analysis easier,
is that we can suppose that $\ccc'=\ccc$, that is, we can focus on algebraic automorphisms
of \ccc. Moreover, it is enough to check only those automorphisms which fix each cell of \ccc.
All such algebraic automorphisms form a group, which we denote by $\saa\ccc$.
Given $f\in\saa\ccc$, we can efficiently decide whether $f$ is induced by a combinatorial
automorphism by considering suitable colored versions of \ccc and its image $\ccc^f$
and applying the algorithm of \cite{ArvindK06} for testing isomorphism of vertex-colored graphs
of color multiplicity 4. The main difficulty is that the group $\saa\ccc$ can be of
exponentially large order. Luckily, it is enough to consider any set of generators of $\saa\ccc$
of polynomial size. We give an explicit description of an appropriate generating set 
based on the system of cliques~$\dcc\ccc$.

We summarize our decision procedure in Section \ref{s:putting}. Theorem \ref{thm:sep4}
is proved by showing that the separability problem for irredundant configurations
reduces in logarithmic space to the isomorphism problem for graphs of color multiplicity 4.
Theorem \ref{thm:amen4} is proved in Section \ref{s:back} based on Equivalence \refeq{sepamen}
and the result of \cite{ArvindK06} that isomorphism testing for graphs of color multiplicity 4
is not harder than computing the rank of a matrix over the 2-element field.
Finally, we apply our amenability criteria to graphs of color multiplicity 4
with at most 16 vertices.  

We conclude with a brief discussion of further questions in Section~\ref{s:concl}.

\section{Basic definitions and facts}\label{s:basics}

We begin with a formal definition of an undirected \emph{vertex-colored} graph
and then introduce a more general notion of a \emph{colored graph},
whose \emph{edges} are directed and colored (Subsection \ref{ss:graphs}). 
The subsequent notion of a \emph{rainbow} (Subsection \ref{ss:ccs})
is identical at first sight but uses a different isomorphism concept.
Informally speaking, rainbows are colored graphs considered up to
renaming the colors. \emph{Coherent configurations} are rainbows with
certain regularity properties. The Weisfeiler-Leman algorithm (Subsection \ref{ss:wl}) 
converts an input graph into a coherent configuration and, moreover, furnishes
this configuration with a canonical coloring. We introduce the amenability
and separability concepts and reduce the former to the latter in Subsection~\ref{ss:reduction}.

\subsection{Colored graphs}\label{ss:graphs}

By a \emph{vertex-colored graph} $G$ we mean an undirected graph without
multiple edges and loops that is endowed with a coloring
of the vertex set $c_G\function{V(G)}C$, where $C$ is a set whose
elements are called \emph{colors}. Vertex-colored graphs $G$ and $H$
are isomorphic if there is a graph isomorphism $\phi\function{V(G)}{V(H)}$
that preserves colors, i.e., 
$$
c_H(\phi(v))=c_G(v)\text{ for all }v\in V(G).
$$

In a more general setting, we consider directed graphs with colored
edges (arrows). The loops are allowed, but the color of a loop $vv$
must differ from the color of any arrow $uw$ with $u\ne w$.
This generalizes the concept of a vertex-colored graph because
the colors of loops can be seen as a vertex coloring, and
a symmetric adjacency relation can be simulated by requiring
that $uv$ is an arrow if any only if $vu$ is an arrow.

In fact, we do not need an adjacency relation (symmetric or not) at all
because all non-arrows can be assigned a special color.
Formally, we define a \emph{colored (directed) graph} $G$ as 
a function $c_G\function{V(G)^2}C$ such that
\begin{equation}
  \label{eq:loops}
 c_G(vv)\ne c_G(uw)\text{ whenever }u\ne w. 
\end{equation}
Two colored graphs $G$ and $H$ are isomorphic if there is a bijection 
$\phi\function{V(G)}{V(H)}$ such that
$$
c_H(\phi(u)\phi(v))=c_G(uv)\text{ for all }u,v\in V(G).
$$

In the context of the isomorphism problem, we can always assume that
\begin{equation}
  \label{eq:trans}
c_G(uv)=c_G(u'v')\text{ if and only if }c_G(vu)=c_G(v'u'),
\end{equation}
that is, if arrows have the same color, then the inverse arrows
must also be equally colored.
This condition can be ensured by modifying the coloring as follows. 
Suppose that an arrow $uv$ is colored \textsl{red} in $G$, and the inverse
arrow $vu$ is colored \textsl{blue}. Then $uv$ is recolored a new color 
\textsl{redblue}, and $vu$ is recolored a new color \textsl{bluered}.
The new colored graph $\hat G$ satisfies the condition \refeq{trans}.
Note that $\hat G\cong\hat H$ exactly when $G\cong H$.
This motivates imposing the condition \refeq{trans} on any colored graph.

For each color $c\in C$, the set $\setdef{uv}{c_G(uv)=c}$ is called
a \emph{color class} of $G$. A color class consisting of loops
is referred to as \emph{vertex color class}. We define the \emph{color multiplicity}
of $G$ as the maximum cardinality of a vertex color class of~$G$.

\subsection{Coherent configurations}\label{ss:ccs}

Let $V$ be a set, whose elements are called \emph{points}.
Let $\ccc=\{R_1,\ldots,R_s\}$ be a partition of the Cartesian square $V^2$,
that is, $\bigcup_{i=1}^sR_i=V^2$ and any two $R_i$ and $R_j$ are
disjoint. An element $R_i$ of \ccc
will be referred to as a \emph{basis relation}.
\ccc is called a \emph{rainbow} if it has the following two properties:
\begin{enumerate}[(A)]
\item\label{item:A}
If a basis relation $R\in\ccc$ contains a loop $vv$, then
all pairs in $R$ are loops;
\item \label{item:B}
For every $R\in\ccc$, the \emph{transpose relation} $R^*=\setdef{uv}{vu\in R}$
is also in~\ccc.
\end{enumerate}

Though formally a rainbow is a pair $(V,\ccc)$, we simplify the notation
by using the same character \ccc for the rainbow and its set of basis relations.
This will cause no ambiguity as the point set $V=V(\ccc)$ is uniquely determined
as the set of all elements occurring in the relations from~\ccc.

A set of points $X\subseteq V$ is called a \emph{fiber} of \ccc
if the set of loops $\setdef{xx}{x\in X}$ is a basis relation of~\ccc.

Two rainbows (or, more generally, two partitions) \ccc and \ccd are \emph{isomorphic}
if there is a bijection $\phi\function{V(\ccc)}{V(\ccd)}$, an \emph{isomorphism} from
\ccc to \ccd, such that $\phi(R)\in\ccd$ for every $R\in\ccc$.
Here $\phi(R)=\setdef{\phi(u)\phi(v)}{uv\in R}$. We can sometimes write the same
as $R^\phi=\setdef{u^\phi v^\phi}{uv\in R}$.

Note that Conditions (\ref{item:A}) and (\ref{item:B}) are analogs of
Conditions \refeq{loops} and \refeq{trans}. By this reason,
a colored graph will also be called a \emph{colored rainbow}.
Let $G$ be a colored graph and \ccc be a rainbow.
If $V(G)=V(\ccc)$ and the color classes of $G$ are exactly the basis relations
of \ccc, then we say that $G$ is a \emph{colored version} of \ccc.
Thus, rainbows \ccc and \ccd are isomorphic if and only if
they have colored versions that are isomorphic (as colored graphs).

A rainbow \ccc is called a \emph{coherent configuration} if,
\begin{enumerate}[(A)]\setcounter{enumi}{2}
\item\label{item:C}
for every triple $R,S,T\in\ccc$, the number $p(uv)=|\setdef{w}{uw\in R,\,wv\in S}|$
is the same for all $uv\in T$.
\end{enumerate}
For a coherent configuration \ccc,
the number $p(uv)$ in (\ref{item:C}) does not depend on the choice of $uv$ in $T$
and is denoted by $p_{RS}^T$. The entries of this 3-dimensional matrix are
called \emph{intersection numbers} of~\ccc.

Coherent configurations \ccc and \ccd are \emph{combinatorially isomorphic}
if they are isomorphic as rainbows. We write $\ccc\ciso\ccd$ for this relationship. 
Correspondingly, any isomorphism from \ccc to \ccd is called \emph{combinatorial}.
Coherent configurations \ccc and \ccd are \emph{algebraically isomorphic} if 
their 3-dimensional matrices of intersection numbers, $p_{RS}^T$ and $p_{R'S'}^{T'}$, 
are isomorphic, that is, there is a bijection $f\function\ccc\ccd$ such that
$$
p_{RS}^T=p_{f(R)f(S)}^{f(T)}.
$$
In this case we write $\ccc\aiso\ccd$.
Such a bijection $f$ is called an \emph{algebraic isomorphism} from \ccc to \ccd.
Note that combinatorially isomorphic coherent configurations are also
algebraically isomorphic. Indeed, any combinatorial isomorphism $\phi$ from \ccc
to \ccd gives rise to the algebraic isomorphism $f$ defined by $f(R)=R^\phi$.

Let $A_R$ denote the adjacency matrix of a relation $R\subseteq V^2$, i.e.,
$A_R[u,v]$ is equal to 1 if $uv\in R$ and to 0 otherwise.
Define $\calA_\ccc$ to be the linear span of the set of 0-1-matrices
$\setdef{A_R}{R\in\ccc}$ over $\complex$. Condition (\ref{item:C}) implies
that $\calA_\ccc$ is closed under matrix multiplication and, hence,
forms a matrix algebra over $\complex$. This algebra is called the
\emph{adjacency algebra} of the coherent configuration \ccc.
It turns out \cite[Proposition 2.3.17]{Ponomarenko-book} that
coherent configurations \ccc and \ccd are algebraically isomorphic
if and only if $\calA_\ccc$ and $\calA_\ccd$ are isomorphic algebras with distinguished bases.
Another important characterization of algebraic isomorphism will be
given in Subsection~\ref{ss:wl}.

Given a family of sets $\calP$, we use $\calP^\cup$ to denote the closure
of $\calP$ under unions.
Given two partitions $\calP$ and $\calQ$ of the same set, we write
$\calP\preccurlyeq\calQ$ if every set in $\calQ$ belongs to $\calP^\cup$
or, equivalently, every set in $\calP$ is a subset of some set in $\calQ$.
In this case we say that $\calP$ is \emph{finer} than $\calQ$
and $\calQ$ is \emph{coarser} than~$\calP$.

\begin{proposition}[{see \cite[Section 2.6.1]{Ponomarenko-book}}]\label{prop:closure}
  Let $\calP$ be a partition of the Cartesian square $V^2$. Then there
is a unique coherent configuration $\ccc=\ccc(\calP)$ such that
\begin{itemize}
\item 
$\ccc\preccurlyeq\calP$, and
\item 
if $\ccc'$ is a coherent configuration such that $\ccc'\preccurlyeq\calP$,
then $\ccc'\preccurlyeq\ccc$.
\end{itemize}
\end{proposition}

\noindent
The coherent configuration $\ccc(\calP)$ is called the \emph{coherent closure} of $\calP$.
In other words, the coherent closure of $\calP$ is the coarsest of those coherent configurations
refining $\calP$.
Given a colored (directed) graph $G$ (in particular, a vertex-colored undirected graph),
let $\calR_G$ denote its uncolored version, that is, the rainbow whose basis relations
are exactly the color classes of $G$. The coherent configuration $\ccc(\calR_G)$ is called 
the \emph{coherent closure of the graph $G$} and denoted by~$\ccc(G)$.

The following notational convention will be intensively used till the end of this section.
Suppose that $\calP$ and $\calQ$ are partitions. Any map $f\function\calP\calQ$ extends
to a map from $\calP^\cup$ to $\calQ^\cup$ in a natural way. Specifically, if
$X=X_1\cup\ldots\cup X_s$ where $X_i\in\calP$, then $X^f=X_1^f\cup\ldots\cup X_s^f$.
Usage of the superscript $f$ can be extended as usually: If 
$\calX=\{X_1,\ldots,X_q\}$ where $X_i\in\calP^\cup$, then $\calX^f=\{X_1^f,\ldots,X_q^f\}$.
Note that, if $f\function\calP\calQ$ is a bijection and $\calP\preccurlyeq\calR$, then
$\calQ\preccurlyeq\calR^f$.

\begin{lemma}[{see \cite[Corollary 2.3.21]{Ponomarenko-book}}]\label{lem:fusion}
If $f$ is an algebraic isomorphism from a coherent configuration $\calP$
to a coherent configuration $\calQ$, and $\calR$ is a coherent configuration
such that $\calP\preccurlyeq\calR$, then $\calR^f$ is also a coherent configuration.
\end{lemma}

\subsection{The Weisfeiler-Leman algorithm}\label{ss:wl}

The following algorithm, that was described by Weisfeiler and Leman in \cite{WLe68},
is now known as the \emph{2-dimensional Weisfeiler-Leman algorithm} (\WL2 for short).
Given a colored graph $G$ as input, the algorithm iteratively computes colorings $c^i_G$
of the Cartesian square $V^2$ for $V=V(G)$.
Initially, $c^0_G=c_G$ and then,
\begin{equation}
  \label{eq:refine-2}
  c^{i+1}_G(uv)=c^{i}_G(uv)\mid\mset{c^{i}_G(uw)\mid c^{i}_G(wv)}_{w\in V},
\end{equation}
where $\mset{}$ denotes the multiset
and $\mid$ denotes the string concatenation (an appropriate encoding is assumed).
Denote the partition of $V^2$ into the color classes of $c^i_G$ by $\calR^i_G$.
Note that $\calR^{i+1}_G\preccurlyeq\calR^i_G$.
Let $t=t_G$ be the minimum number such that $\calR^t_G=\calR^{t-1}_G$.
The algorithm terminates after the $t$-th color refinement round.
As easily seen, $\calR^t_G$ is a coherent configuration.\footnote{%
Moreover, a partition $\calP$ is a coherent configuration if and only if \WL2
does not make any color refinement when applied to a colored version of~$\calP$.}

\begin{proposition}[{see \cite[Section 2.6.1]{Ponomarenko-book}}]\label{prop:wl}
  $\calR^t_G=\ccc(G)$.
\end{proposition}

An easy induction on $i$ shows that, if $\phi$ is an isomorphism from $G$ to $H$, then
\begin{equation}
  \label{eq:phi}
c^i_G(uv)=c^i_H(u^\phi v^\phi).  
\end{equation}
Thus, the coloring produced by \WL2 is canonical and can be used for isomorphism testing.
We say that colored graphs $G$ and $H$ are \emph{\WL2-equivalent}
and write $G\eqq H$ if
\begin{equation}
  \label{eq:msets}
\msetdef{c^t_G(uv)}{uv\in V(G)^2}=\msetdef{c^t_H(uv)}{uv\in V(H)^2}  
\end{equation}
for $t=t_G$ (equivalently, for $t=t_H$, or for all~$t$).

Suppose that $G\eqq H$. Equality \refeq{msets} implies that there is
a one-to-one map $f\function{\ccg G}{\ccg H}$ preserving the \WL2colors.
Note that $f$ is an algebraic isomorphism from $\ccg G$ to $\ccg H$.
We, therefore, have the following diagram:
$$
\begin{array}{ccc}
G\cong H & \implies & G\eqq H \\[1mm]
\Downarrow && \Downarrow \\[1mm]
\ccg G\ciso\ccg H & \implies & \ccg G\aiso\ccg H
\end{array}
$$
In the other direction, let $f$ be an algebraic isomorphism from \ccc
to \ccd. If $\tilde\ccc$ is a colored version of \ccc, and $\tilde\ccd$
is the colored version of \ccd where each color class $f(C)$ inherits
the color of $C$, then $\tilde\ccc\eqq\tilde\ccd$. Thus,
the Weisfeiler-Leman algorithm provides yet another interpretation for
algebraic isomorphism of coherent configurations: $\ccc\aiso\ccd$ if
and only if $\ccc$ and $\ccd$ have colored versions $\tilde\ccc$ and
$\tilde\ccd$ respectively such that $\tilde\ccc\eqq\tilde\ccd$.

The fact that the \wl-equivalence of graphs $G$ and $H$ determines
an algebraic isomorphism of their coherent closures $\ccg G$ and $\ccg H$
has the converse, which we state now.

\begin{lemma}\label{lem:aiso-eqq}
Let $\calP$ and $\calQ$ be rainbows and $f\function{\ccg\calP}{\ccg\calQ}$
be an algebraic isomorphism such that $\calQ=\calP^f$.
Let $G$ be a colored version of $\calP$. Let $H$ be the colored version of $\calQ$
that inherits the colors from $G$ via $f$, that is, every color class $C$ of $G$
has the same color as the color class $C^f$ of~$H$.
\begin{enumerate}[\bf 1.]
\item 
If \wl is run on $G$ (resp.\ $H$), then it outputs a colored version of $\ccg\calP$ 
(resp.~$\ccg\calQ$).
\item 
The map $f$ preserves the \wl coloring, that is,
\begin{equation}
  \label{eq:ci}
c^i_G(Z)=c^i_H(Z^f)  
\end{equation}
for all $i\ge0$ and $Z\in\ccg\calP$,
where $c^i_G(Z):=c^i_G(uv)$ for an arrow $uv$ in $Z$ and $c^i_H(Z^f)$ is defined similarly.
In particular, $G\eqq H$.
\item 
If, moreover, $G\cong H$, then $f$ is induced by a combinatorial isomorphism
from $\ccg\calP$ to $\ccg\calQ$.
\end{enumerate}
\end{lemma}

\begin{proof}
\textit{1.}  
This part follows directly from Proposition~\ref{prop:wl}.

\textit{2.}  
We use the induction on $i$. For $i=0$, Equality \refeq{ci} follows from the
fact that the coloring of $H$ is defined according to the map $f$.
Assume that Equality \refeq{ci} is true for some value of $i$ for all $Z\in\ccg\calP$ 
and prove that then $c^{i+1}_G(Z)=c^{i+1}_H(Z^f)$ for all $Z\in\ccg\calP$.
Choose arbitrarily an arrow $uv$ in $Z$ and an arrow $u'v'$ in $Z^f$. 
It suffices to prove that
\begin{equation}
  \label{eq:msets2}
\mset{c^{i}_G(uw)\mid c^{i}_G(wv)}_{w\in V(G)}=
\mset{c^{i}_H(u'w')\mid c^{i}_H(w'v')}_{w'\in V(H)}.
\end{equation}
Each pair $X,Y\in\ccg\calP$ contributes $p^Z_{XY}$ elements $c^{i}_G(X)\mid c^{i}_G(Y)$
into the left-hand side of \refeq{msets2}. Similarly, for each $X,Y\in\ccg\calP$, 
the right-hand side of \refeq{msets2} contains $p^{f(Z)}_{f(X)f(Y)}$ elements 
$c^{i}_H(X^f)\mid c^{i}_H(Y^f)$. Since $f$ is an algebraic isomorphism,
$$
p^{f(Z)}_{f(X)f(Y)}=p^Z_{XY}.
$$
By the induction assumption,
$$
c^{i}_H(X^f)\mid c^{i}_H(Y^f)=c^{i}_G(X)\mid c^{i}_G(Y).
$$
Equality \refeq{msets2} follows.

\textit{3.}  
Let $\phi$ be an isomorphism from $G$ to $H$.
By \refeq{phi}, $\phi$ preserves the \wl coloring and, therefore,
every $X$ in $\ccg\calP$ has the same \wl color as $X^\phi$ in $\ccg\calQ$.
By Part 2, also $X^f$ has this color. This implies that $X^f=X^\phi$.
It remains to note that $\phi$ is a combinatorial isomorphism from
$\ccg\calP$ to $\ccg\calQ$ (which readily follows from the fact that
$\phi$ preserves the \wl coloring).
\end{proof}

\subsection{Amenability to \WL2 and separability of the coherent closure}\label{ss:reduction}

We call a colored graph $G$ \emph{amenable (to \wl)}
if \wl distinguishes $G$ from any non-isomorphic graph $H$,
that is, $G\eqq H$ implies $G\cong H$.

A coherent configuration \ccc is \emph{separable} if 
every algebraic isomorphism from \ccc to any coherent configuration \ccd
is induced by a combinatorial isomorphism from \ccc to~\ccd.

\begin{theorem}\label{thm:reduction}
  A colored graph $G$ is amenable if and only if its
coherent closure $\ccg G$ is separable.
\end{theorem}

\begin{proof}
  ($\Longleftarrow$)
Suppose that $G\eqq H$. Consider the map $f\function{\ccg G}{\ccg H}$
taking each $X\in\ccg G$ to the basis relation in $\ccg H$ with the same \wl color.
This map is an algebraic isomorphism from $\ccg G$ to $\ccg H$.
Since $\ccg G$ is separable, $f$ is induced by a combinatorial isomorphism $\phi$
from $\ccg G$ to $\ccg H$. Since $\phi$ preserves the \wl coloring $c^t_G$,
it preserves also the initial coloring $c^0_G$, which means that $\phi$ is
an isomorphism from $G$ to~$H$.

  ($\implies$)
Given an algebraic isomorphism $f\function{\ccg G}\ccd$, we have to show that $f$
is induced by a combinatorial isomorphism from $\ccg G$ to $\ccd$.
Let $\calR$ denote the rainbow which is the uncolored version of $G$.

\begin{claim}\label{cl:a}
  $\ccd=\ccg{\calR^f}$.
\end{claim}

\begin{subproof}
Since $\ccg{\calR}\preccurlyeq\calR$, we have $\ccd\preccurlyeq\calR^f$. By Proposition \ref{prop:closure},
this implies that $\ccd\preccurlyeq\ccg{\calR^f}$.  
Applying the inverse map $f^{-1}\function{\ccd^\cup}{{\ccg G}^\cup}$, we have
\begin{equation}
  \label{eq:CRf}
\ccg{\calR}\preccurlyeq(\ccg{\calR^f})^{f^{-1}}.  
\end{equation}
Since $f^{-1}$ is an algebraic isomorphism from $\ccd$ to $\ccg{\calR}$,
Lemma \ref{lem:fusion} implies that $(\ccg{\calR^f})^{f^{-1}}$ is a coherent configuration.
Since $f^{-1}$ takes $\calR^f$ back to $\calR$, we have
$$
(\ccg{\calR^f})^{f^{-1}}\preccurlyeq\calR
$$
and, by Proposition \ref{prop:closure},
$$
(\ccg{\calR^f})^{f^{-1}}\preccurlyeq\ccg{\calR}.
$$
Along with \refeq{CRf}, this shows that $(\ccg{\calR^f})^{f^{-1}}=\ccg{\calR}$
or, equivalently, $\ccg{\calR^f}=\ccg{\calR}^f=\ccd$.
\end{subproof}

Let $G^f$ denote the colored version of $\calR^f$ that inherits the colors of $G$
according to the bijection $f\function{\calR}{\calR^f}$. Using Claim \ref{cl:a}
and Part 2 of Lemma \ref{lem:aiso-eqq}, we conclude that $G\eqq G^f$.
Since $G$ is amenable, $G\cong G^f$. By Part 3 of Lemma \ref{lem:aiso-eqq},
the algebraic isomorphism $f$ is induced by a combinatorial isomorphism
from $\ccg G$ to~$\ccd$.
\end{proof}

\section{Preliminaries on the structure of coherent configurations}\label{s:prel}

\subsection{Fibers and interspaces}\label{ss:f-i}

Let \ccc be a coherent configuration on the point set $V=V(\ccc)$.
Recall that a set of points $X\subseteq V$ is a fiber of \ccc
if it underlies a reflexive basis relation of \ccc.
Denote the set of all fibers of \ccc by $F(\ccc)$.
By Property (\ref{item:A}) in Section \ref{ss:ccs},
$F(\ccc)$ is a partition of $V$. Property (\ref{item:C}) implies that
for every basis relation $R$ of \ccc there are, not necessarily distinct, 
fibers $X$ and $Y$ such that $R\subseteq X\times Y$. Thus, if $X,Y\in F(\ccc)$,
then the Cartesian product $X\times Y$ is split
into basis relations of \ccc. We denote this partition by $\ccc[X,Y]$.
If $X=Y$, we simplify our notation to $\ccc[X]=\ccc[X,X]$.
Note that $\ccc[X]$ is a coherent configuration on $X$,
with $X$ being its single fiber. We will call $\ccc[X]$ a \emph{cell} of \ccc. 
In general, coherent configurations with a single fiber are called \emph{association schemes}.

All possible association schemes on 2, 3, and 4 points are
depicted in Figure \ref{fig:cells}. Undirected edges are used to show
basis relations that are equal to their transposes. 
Directed edges are used to show basis relations that are not
equal to their transposes, and those are then not shown
as they are reconstructable by reversing the arrows.
Loops are not shown at all. Though we use different patterns for different
basis relations, remember that $\ccc[X]$ is just an (uncolored) partition of~$X$.
We give the 4-point cells names $K_4$, $C_4$, $\vec C_4$, and $F_4$
according to the names of the graphs appearing as underlying shapes
in the cell representations. Here, $\vec C_4$ stands for the
directed 4-cycle, and $F_4$ stands for the factorization of the complete graph $K_4$
into three matchings $2K_2$. We use notation $\ccc[X]\simeq C_4$ etc.\
to indicate which type the cell $\ccc[X]$ has.

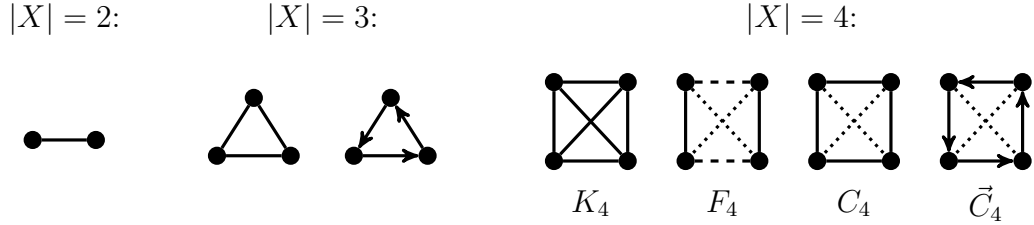
\begin{figure}
  \centering
\begin{tikzpicture}[every node/.style={circle,draw,inner sep=2pt,fill=black},very thick,scale=.7]

  \begin{scope}[xshift=10mm,yshift=5mm]
\node[draw=none,fill=none] at (2.5,2.7) {$|X|=2$:};
    \path (1.9,0.4) node (u1) {}
      (3.1,0.4) node (u2) {} edge (u1);
  \end{scope}

  \begin{scope}[xshift=40mm,yshift=5mm]
\node[draw=none,fill=none] at (4.4,2.7) {$|X|=3$:};
    \path (2.4,0.1) node (u1) {}
      (3.8,0.1) node (u2) {} edge (u1)
      (3.1,1.2) node (u3) {} edge (u2) edge (u1);
    \path (5.0,0.1) node (u1) {}
      (6.4,0.1) node (u2) {} edge[<-] (u1)
      (5.7,1.2) node (u3) {} edge[<-] (u2) edge[->] (u1);
  \end{scope}

  \begin{scope}[xshift=105mm,yshift=5mm]
\node[draw=none,fill=none] at (3.0,-.8) {$K_4$};
    \path (2.3,0) node (u1) {}
      (3.7,0) node (u2) {} edge (u1)
      (3.7,1.5) node (u3) {} edge (u2) edge (u1)
      (2.3,1.5) node (u4) {} edge (u3) edge (u2) edge (u1);
  \end{scope}

  \begin{scope}[xshift=130mm,yshift=5mm]
\node[draw=none,fill=none] at (4.5,2.7) {$|X|=4$:};
\node[draw=none,fill=none] at (3.0,-.8) {$F_4$};
    \path (2.3,0) node (u1) {}
      (3.7,0) node (u2) {} edge[dashed] (u1)
      (3.7,1.5) node (u3) {} edge (u2) edge[dotted] (u1)
      (2.3,1.5) node (u4) {} edge[dashed] (u3) edge[dotted] (u2) edge (u1);
  \end{scope}

  \begin{scope}[xshift=155mm,yshift=5mm]
\node[draw=none,fill=none] at (3.0,-.8) {$C_4$};
    \path (2.3,0) node (u1) {}
      (3.7,0) node (u2) {} edge (u1)
      (3.7,1.5) node (u3) {} edge (u2) edge[dotted] (u1)
      (2.3,1.5) node (u4) {} edge (u3) edge[dotted] (u2) edge (u1);
  \end{scope}

  \begin{scope}[xshift=180mm,yshift=5mm]
\node[draw=none,fill=none] at (3.0,-.8) {$\vec{C}_4$};
    \path (2.3,0) node (u1) {}
      (3.7,0) node (u2) {} edge[<-] (u1)
      (3.7,1.5) node (u3) {} edge[<-] (u2) edge[dotted] (u1)
      (2.3,1.5) node (u4) {} edge[<-] (u3) edge[dotted] (u2) edge[->] (u1);
  \end{scope}

\end{tikzpicture}
\caption{Cells $\ccc[X]$ on 2, 3, and 4 points.}
\label{fig:cells}
\end{figure}

If $X\ne Y$, we call the partition $\ccc[X,Y]$ an \emph{interspace} of \ccc.
Note that $R\in\ccc[X,Y]$ if and only if $R^*\in\ccc[Y,X]$.
In particular, $R\in\ccc[X,Y]$ for $X\ne Y$ implies that $R^*\ne R$.
If $|\ccc[X,Y]|=1$, that is, $X\times Y$ is a basis relation of \ccc,
then the interspace $\ccc[X,Y]$ will be called \emph{uniform}.
Otherwise, $\ccc[X,Y]$ will be referred to as \emph{non-uniform}.
The interspace $\ccc[X,Y]$ is uniform if and only if so is $\ccc[Y,X]$.

If $R\in\ccc[X,Y]$, then the number of arrows in $R$ from a point $x\in X$
is the same for each  $x$ in $X$. We call this number the \emph{valency} of $R$
and denote it by~$d(R)$.

\begin{lemma}\label{lem:coprime}
  Let $X,Y\in F(\ccc)$.
If $|X|$ and $|Y|$ are coprime, then $\ccc[X,Y]$ is uniform.
\end{lemma}

\begin{proof}
Let $R$ be a basis relation such that $R\subseteq\ccc[X,Y]$.
Recall that the valency $d(R)$ is equal to the
number of arrows in $R$ from each point $x\in X$. Note also that
the valency $d(R^*)$ of the transpose relation $R^*$
is equal to the number of arrows in $R$ to a point $y\in Y$;
it does not depend on the choice of $y$. It follows that
$$
d(R)|X|=|R|=d(R^*)|Y|.
$$
Since $|X|$ and $|Y|$ are coprime, $d(R)$ is divisible by $|Y|$.
Taking into account that $d(R)\le|Y|$, we obtain the equality $d(R)=|Y|$.
As a consequence, $R=X\times Y$.
\end{proof}

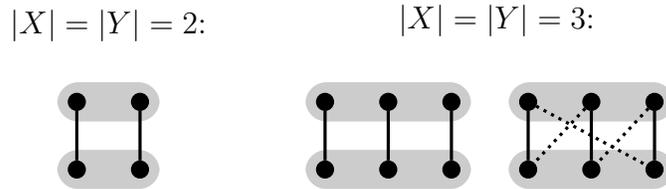
\begin{figure}[b]
  \centering
\begin{tikzpicture}[every node/.style={circle,draw,inner sep=2pt,fill=black},very thick,scale=.6]

  \begin{scope}[xshift=0mm,yshift=5mm]
\node[draw=none,fill=none] at (3.0,3.2) {$|X|=|Y|=2$:};
    \path (2.3,0) node (u1) {}
      (3.7,0) node (u2) {} 
      (3.7,1.5) node (u3) {} edge (u2)
      (2.3,1.5) node (u4) {} edge (u1);
\colclass{u1,u2}
\colclass{u3,u4}
  \end{scope}

  \begin{scope}[xshift=55mm,yshift=5mm]
\node[draw=none,fill=none] at (6.1,3.3) {$|X|=|Y|=3$:};
    \path (2.3,0) node (u1) {}
      (3.7,0) node (u2) {}
      (5.1,0) node (u3) {}
      (5.1,1.5) node (u4) {} edge (u3)
      (3.7,1.5) node (u5) {} edge (u2)
      (2.3,1.5) node (u6) {} edge (u1);
\colclass{u1,u2,u3}
\colclass{u4,u5,u6}
  \end{scope}

  \begin{scope}[xshift=100mm,yshift=5mm]
    \path (2.3,0) node (u1) {}
      (3.7,0) node (u2) {}
      (5.1,0) node (u3) {}
      (5.1,1.5) node (u4) {} edge (u3) edge[dotted] (u2)
      (3.7,1.5) node (u5) {} edge (u2) edge[dotted] (u1)
      (2.3,1.5) node (u6) {} edge (u1) edge[dotted] (u3);
\colclass{u1,u2,u3}
\colclass{u4,u5,u6}
  \end{scope}

\end{tikzpicture}
\caption{Non-uniform interspaces between fibers with at most 3 points.}
\label{fig:interspace3}
\end{figure}

Thus, all interspaces $\ccc[X,Y]$ with $|X|=1$ are uniform,
and so are also interspaces with $|X|=2$ and $|Y|=3$.
Figure \ref{fig:interspace3} shows all non-uniform interspaces $\ccc[X,Y]$ with 
$|X|\le3$ and $|Y|\le3$. Here we adhere to the following convention:
A pair $xy$ with $x\in X$ in $y\in Y$ is shown as an undirected edge,
as it is automatically ordered once the ordered pair of fibers $X,Y$ is given.
To facilitate visualization, one basis relation
in each picture is not shown. It is reconstructable by taking the
bipartite complement of the shown part. 

\begin{figure}
  \centering
\begin{tikzpicture}[every node/.style={circle,draw,inner sep=2pt,fill=black},very thick,xscale=.6,yscale=.65]

  \begin{scope}[xshift=0mm,yshift=5mm]
\node[draw=none,fill=none] at (4.3,3.0) {$|X|=2,|Y|=4$:};
    \path
      (3.0,0) node (u2) {}
      (5.8,0) node (u3) {}
      (6.5,1.5) node (u5) {} edge (u3)
      (5.1,1.5) node (u6) {} edge (u3)
      (3.7,1.5) node (u7) {} edge (u2)
      (2.3,1.5) node (u8) {} edge (u2);
\colclass{u2,u3}
\colclass{u5,u6,u7,u8}
\node[draw=none,fill=none] at (4.3,-1.3) {$2K_{1,2}$};
  \end{scope}

  \begin{scope}[xshift=75mm,yshift=5mm]
    \path (2.3,0) node (u1) {}
      (3.7,0) node (u2) {}
      (5.1,0) node (u3) {}
      (6.5,0) node (u4) {}
      (6.5,1.5) node (u5) {} edge (u4)
      (5.1,1.5) node (u6) {} edge (u3)
      (3.7,1.5) node (u7) {} edge (u2)
      (2.3,1.5) node (u8) {} edge (u1);
\colclass{u1,u2,u3,u4}
\colclass{u5,u6,u7,u8}
\node[draw=none,fill=none] at (4.3,-1.3) {$4K_{1,1}$};
  \end{scope}

  \begin{scope}[xshift=135mm,yshift=5mm]
\node[draw=none,fill=none] at (4.5,3.0) {$|X|=|Y|=4$, $|\ccc[X,Y]|=2$:};
    \path (2.3,0) node (u1) {}
      (3.7,0) node (u2) {}
      (5.1,0) node (u3) {}
      (6.5,0) node (u4) {}
      (6.5,1.5) node (u5) {} edge (u4) edge (u3)
      (5.1,1.5) node (u6) {} edge (u3) edge (u4)
      (3.7,1.5) node (u7) {} edge (u2) edge (u1)
      (2.3,1.5) node (u8) {} edge (u1) edge (u2);
\colclass{u1,u2,u3,u4}
\colclass{u5,u6,u7,u8}
\node[draw=none,fill=none] at (4.3,-1.3) {$2K_{2,2}$};
  \end{scope}

  \begin{scope}[xshift=195mm,yshift=5mm]
    \path (2.3,0) node (u1) {}
      (3.7,0) node (u2) {}
      (5.1,0) node (u3) {}
      (6.5,0) node (u4) {}
      (6.5,1.5) node (u5) {} edge (u4) edge (u3)
      (5.1,1.5) node (u6) {} edge (u3) edge (u2)
      (3.7,1.5) node (u7) {} edge (u2) edge (u1)
      (2.3,1.5) node (u8) {} edge (u1) edge (u4);
\colclass{u1,u2,u3,u4}
\colclass{u5,u6,u7,u8}
\node[draw=none,fill=none] at (4.3,-1.3) {$C_8$};
  \end{scope}

  \begin{scope}[xshift=0mm,yshift=-55mm]
\node[draw=none,fill=none] at (8.0,2.8) {$|X|=|Y|=4$, $|\ccc[X,Y]|=3$:};
    \path (2.3,0) node (u1) {}
      (3.7,0) node (u2) {}
      (5.1,0) node (u3) {}
      (6.5,0) node (u4) {}
      (6.5,1.5) node (u5) {} edge (u4) edge[dotted] (u3)
      (5.1,1.5) node (u6) {} edge (u3) edge[dotted] (u4)
      (3.7,1.5) node (u7) {} edge (u2) edge[dotted] (u1)
      (2.3,1.5) node (u8) {} edge (u1) edge[dotted] (u2);
\colclass{u1,u2,u3,u4}
\colclass{u5,u6,u7,u8}
  \end{scope}

  \begin{scope}[xshift=60mm,yshift=-55mm]
    \path (2.3,0) node (u1) {}
      (3.7,0) node (u2) {}
      (5.1,0) node (u3) {}
      (6.5,0) node (u4) {}
      (6.5,1.5) node (u5) {} edge (u4) edge[dotted] (u3)
      (5.1,1.5) node (u6) {} edge (u3) edge[dotted] (u2)
      (3.7,1.5) node (u7) {} edge (u2) edge[dotted] (u1)
      (2.3,1.5) node (u8) {} edge (u1) edge[dotted] (u4);
\colclass{u1,u2,u3,u4}
\colclass{u5,u6,u7,u8}
  \end{scope}

  \begin{scope}[xshift=135mm,yshift=-55mm]
\node[draw=none,fill=none] at (8.0,2.8) {$|X|=|Y|=4$, $|\ccc[X,Y]|=4$:};
    \path (2.3,0) node (u1) {}
      (3.7,0) node (u2) {}
      (5.1,0) node (u3) {}
      (6.5,0) node (u4) {}
      (6.5,1.5) node (u5) {} edge (u4) edge[dotted] (u3) edge[dashed,red] (u2)
      (5.1,1.5) node (u6) {} edge (u3) edge[dotted] (u4) edge[dashed,red] (u1)
      (3.7,1.5) node (u7) {} edge (u2) edge[dotted] (u1) edge[dashed,red] (u4)
      (2.3,1.5) node (u8) {} edge (u1) edge[dotted] (u2) edge[dashed,red] (u3);
\colclass{u1,u2,u3,u4}
\colclass{u5,u6,u7,u8}
  \end{scope}

  \begin{scope}[xshift=195mm,yshift=-55mm]
    \path (2.3,0) node (u1) {}
      (3.7,0) node (u2) {}
      (5.1,0) node (u3) {}
      (6.5,0) node (u4) {}
      (6.5,1.5) node (u5) {} edge (u4) edge[dotted] (u3) edge[dashed,red] (u1)
      (5.1,1.5) node (u6) {} edge (u3) edge[dotted] (u2) edge[dashed,red] (u4)
      (3.7,1.5) node (u7) {} edge (u2) edge[dotted] (u1) edge[dashed,red] (u3)
      (2.3,1.5) node (u8) {} edge (u1) edge[dotted] (u4) edge[dashed,red] (u2);
\colclass{u1,u2,u3,u4}
\colclass{u5,u6,u7,u8}
  \end{scope}

\end{tikzpicture}
\caption{Non-uniform interspaces between fibers with at most 4 points.}
\label{fig:interspace4}
\end{figure}

If we allow also fibers on 4 points, then the list of all non-uniform interspaces
(up to isomorphism of partitions) is completed in Figure \ref{fig:interspace4}.
Again, one basis relation is not shown in each case
as it is reconstructable by taking the bipartite complement.
If $|X|<|Y|=4$, then Lemma \ref{lem:coprime} implies that a non-uniform interspace
$\ccc[X,Y]$ is possible only for $|X|=2$. Such an interspace is unique.
Suppose that $|X|=|Y|=4$. As agreed, we represent a basis relation as
an undirected bipartite graph with vertex classes $X$ and $Y$ (tacitly assuming
the arrows in the direction from $X$ to $Y$). Note that this graph must be regular.
There is a unique basis relation of degree 1 (a 4-matching),
a unique basis relation of degree 3 (the bipartite complement of a 4-matching),
and there are two self-complementary basis relations of degree 2 (two disjoint 4-cycles
and a 8-cycle). This yields three non-uniform interspaces $\ccc[X,Y]$ with 2 basis 
relations, for which we will use names $4K_{1,1}$, $2K_{2,2}$, and $C_8$,
using the notation $\ccc[X,Y]\simeq 2K_{2,2}$ etc.
An interspace $\ccc[X,Y]$ with $|\ccc[X,Y]|=3$ consists of two basis relations
of degree 1 and one basis relation of degree 2. The latter can be either
an 8-cycle or the disjoint union of two 4-cycles. In each case, the two
basis relations of degree 1 are obtainable in a unique, up to combinatorial isomorphism, way by splitting
the complement into two 4-matchings. This yields two interspaces shown
in Figure \ref{fig:interspace4}. An interspace $\ccc[X,Y]$ with $|\ccc[X,Y]|=4$ 
is a factorization of $X\times Y$ into four 4-matchings.
One factor is unique up to isomorphism. Two factors can form together
either an 8-cycle or the disjoint union of two 4-cycles, as shown
in the picture for the preceding case of $|\ccc[X,Y]|=3$.
The bipartite complement of an 8-cycle is also an 8-cycle, which is uniquely factorizable into
two further 4-matchings, and, therefore, $\ccc[X,Y]$ is uniquely determined
in this case. The bipartite complement of 
two 4-cycles also consists of two 4-cycles. It is factorizable in two ways, but one of them
leads to two 4-cycles whose union is an 8-cycle, which is the case we already have. 
Thus, there are two interspaces with $|\ccc[X,Y]|=4$,
one with two factors forming an 8-cycle and one with every two factors
forming two 4-cycles.

\subsection{Direct sums}\label{ss:direct}

Extending our notation, for any $U\in F(\ccc)^\cup$ we let $\ccc[U]$ denote
the set of all basis relations of \ccc contained in $U^2$. Note that
$\ccc[U]$ is a coherent configuration on the point set $U$. Let $W=V\setminus U$.
We say that \ccc is the \emph{direct sum} of coherent configurations
$\ccc[U]$ and $\ccc[W]$ and write $\ccc=\ccc[U]\boxplus\ccc[W]$ if
the interspace $\ccc[X,Y]$ is uniform for every two fibers $X,Y\in F(\ccc)$ 
such that $X\subseteq U$ and $Y\subseteq W$.

\begin{lemma}[{see \cite[Corollary 3.2.8]{Ponomarenko-book}}]\label{lem:direct}
Suppose that $\ccc=\ccc_1\boxplus\ccc_2$. The coherent configuration \ccc
is separable if and only if both $\ccc_1$ and $\ccc_2$ are separable.
\end{lemma}

Lemma \ref{lem:direct} reduces the general separability problem to its
restriction for \emph{indecomposable} coherent configurations, that is,
those configurations which cannot be split into a direct sum.
Lemma \ref{lem:coprime} implies that an indecomposable coherent
configuration of maximum fiber size at most 4 either has 
maximum fiber size at most 3 or has only fibers of size 4 or~2.

\subsection{Algebraic isomorphisms and fibers}\label{ss:ai-fibers}

\begin{lemma}[{see \cite[Proposition 2.3.18]{Ponomarenko-book}}]\label{lem:refl}
  Let $f\function\ccc{\ccc'}$ be an algebraic isomorphism of coherent configurations.
If $R$ is a reflexive basis relations of \ccc, then $R^f$ is a reflexive basis relations 
of~$\ccc'$.
\end{lemma}

Given an algebraic isomorphism $f\function\ccc{\ccc'}$ and a fiber $X\in F(\ccc)$,
we will write $f(X)$ to denote the fiber of $\ccc'$ underlying the reflexive
basis relation $R^f\in\ccc'$ for the reflexive basis relation $R=\setdef{xx}{x\in X}$ in~\ccc.
Thus, the algebraic isomorphism $f$ determines a bijection $X\mapsto f(X)$ from $F(\ccc)$ to~$F(\ccc')$.

As it follows directly from the definitions, whenever we want to check that
a given rainbow is a coherent configuration or that a given map is
an algebraic isomorphism of coherent configurations, this is enough
to do locally for every triple of fibers. We formalize this observation
as follows.

  \begin{lemma}\label{lem:cl:b}
Let \ccc be a coherent configuration and $\ccc'$ be a rainbow.
Let $f\function{\ccc}{\ccc'}$ be a bijection that induces a one-to-one
correspondence between the reflexive relations of \ccc and the reflexive 
relations of $\ccc'$, that is, establishes a bijection $f\function{F(\ccc)}{F(\ccc')}$.
If $\ccc'[f(A)\cup f(B)\cup f(C)]$ is a coherent configuration for any
fibers $A,B,C\in F(\ccc)$ and, moreover,  the restriction of $f$ to $\ccc[A\cup B\cup C]$ 
is an algebraic isomorphism from $\ccc[A\cup B\cup C]$
to $\ccc'[A'\cup B'\cup C']$, then $\ccc'$ is also a coherent configuration
and $f$ is an algebraic isomorphism from \ccc to~$\ccc'$.
  \end{lemma}

Lemma \ref{lem:cl:b} readily implies the following fact.
Let \ccc be a coherent configuration on point set $V$.
For $X\in F(\ccc)$, we denote $\ccc\setminus X=\ccc[V\setminus X]$.

\begin{lemma}\label{lem:fibers}
  Let $f\function\ccc{\ccc'}$ be an algebraic isomorphism of coherent configurations.
If $X\in F(\ccc)$, then the restriction of $f$ to $\ccc\setminus X$
is an algebraic isomorphism from $\ccc\setminus X$ to $\ccc'\setminus f(X)$.
\end{lemma}

\section{Cutting it down: Interspaces with a matching}\label{s:excl-matching}

 Let $M$ be a basis relation of a coherent configuration \ccc.
Suppose that $M\in\ccc[X,Y]$ for distinct fibers $X$ and $Y$.
We call $M$ a \emph{matching} if both $M$ and its transpose
have valency 1, i.e., $d(M)=d(M^*)=1$. This means that $M$ determines a one-to-one
correspondence between $X$ and $Y$. In the case that $M\in\ccc[X]$ for a fiber $X$,
we call $M$ a \emph{matching} if $d(M)=d(M^*)=1$ and, additionally, $M$ is symmetric and irreflexive. 
In this case, $M$ determines a partition of $X$ into pairs of points.

\begin{lemma}\label{lem:excl-matching-both}
Suppose that a coherent configuration \ccc contains a matching basis relation in an interspace $\ccc[X,Y]$.
Then \ccc is separable if and only if $\ccc\setminus X$ is separable.
\end{lemma}

We split the proof of Lemma \ref{lem:excl-matching-both} into two parts,
Lemmas \ref{lem:excl-matching} and \ref{lem:excl-matching-conv} below.

Before proceeding to the proof, we need some definitions.
We call a rainbow $\calP$ \emph{fibrous} if for every basis relation $R\in\calP$
there are, not necessarily distinct, fibers $X,Y\in F(\calP)$ such that $R\subseteq X\times Y$.
Note that a coherent configuration is a fibrous rainbow. For a fibrous rainbow $\calP$,
we can use the notation $\calP[X,Y]=\setdef{R\in\calP}{R\subseteq X\times Y}$,
that was introduced for coherent configurations.

Let $\calP$ be a rainbow on point set $V$ and $\calQ$ be a rainbow on point set $U$.
We call a surjective function $\nu\function VU$ a \emph{folding map from $\calP$ to $\calQ$} if 
\begin{itemize}
\item 
$R^\nu\in\calQ$ for every basis relation $R\in\calP$, and
\item 
$|R^\nu|=|R|$ for every reflexive basis relation $R\in\calP$.
\end{itemize}
Note that, by the first condition, for every fiber $X$ of $\calP$,
its image $X^\nu$ is a fiber of $\calQ$. Taking into account
the second condition, we see that the restriction of $\nu$ to $X$
is a bijection from $X$ to $X^\nu$. Therefore, for every two (not necessarily distinct) 
fibers $X,Y\in F(\calP)$, the map $\nu$ induces a bijection from $X\times Y$
onto $X^\nu\times Y^\nu$. If $\calP$ is fibrous, this implies that $\nu$,
extended to a map from $V^2$ to $U^2$, induces a bijection from $R$ to $R^\nu$
for each $R\in\calP$. Moreover, $\nu$ determines a one-to-one correspondence
$R\mapsto R^\nu$ between $\calP[X,Y]$ and $\calQ[X^\nu,Y^\nu]$ (hence
$\calQ$ must be fibrous too).

We say that basis relations $R,S,T$ of a fibrous rainbow $\calP$ 
form a \emph{collocated triple} (or are \emph{collocated}) if there are fibers
$X,Y,Z\in F(\calP)$, not necessarily distinct, such that
$T\in\calP[X,Y]$, $R\in\calP[X,Z]$, and $S\in\calP[Z,Y]$.
If $\calP$ is a coherent configuration and $R,S,T$ are not collocated, then
obviously $p^T_{RS}=0$.

\begin{lemma}\label{lem:cover}
  If $\calQ$ is a coherent configuration and $\nu$ is a
folding map from a fibrous rainbow $\calP$ to $\calQ$, then
$\calP$ is also a coherent configuration. Moreover,
$$
p^{T}_{RS}=p^{T^\nu}_{R^\nu S^\nu}
$$
for every collocated triple $R,S,T\in\calP$.
\end{lemma}

\begin{proof}
If $R,S,T\in\calP$ are not collocated, then $p^{T}_{RS}$ is obviously
well defined and equal to 0, which means that Condition (\ref{item:C}) in the definition
of a coherent configuration is fulfilled for such triples. Suppose that $R,S,T$ is a collocated triple in $\calP$ 
and that this is certified by the fibers $X,Y,Z\in F(\calP)$. 
Note that $R^\nu,S^\nu,T^\nu$ is a collocated triple in $\calQ$,
which is certified by fibers $X^\nu,Y^\nu,Z^\nu\in F(\calQ)$. 
Let $xy\in T$ and consider the set $W$ of all $z\in Z$ 
such that $xz\in R$ and $zy\in S$. Note that $z\in W$ if and only if 
$x^\nu z^\nu\in R^\nu$ and $z^\nu y^\nu\in S^\nu$. Since $x^\nu y^\nu\in T^\nu$
and $\nu$ is a bijection from $Z$ onto $Z^\nu$, we conclude that
$|W|=p^{T^\nu}_{R^\nu S^\nu}$, not depending on the choice
of $xy$ in~$T$.
\end{proof}

Note that matching basis relations are \emph{thin} in the sense of \cite{Ponomarenko-book}
and, therefore, Part 1 of the following lemma can also be obtained from \cite[Example~2.2.2]{Ponomarenko-book}.

\begin{lemma}\label{lem:matching}
Suppose that an interspace $\ccc[X,Y]$ of a coherent configuration \ccc 
contains a matching basis relation $M$ and define a function $\nu\function{V(\ccc)}{V(\ccc)\setminus X}$ by
$\nu(x)=y$ for all $xy\in M$ and $\nu(z)=z$ for all $z\notin X$. Then the following is true:
  \begin{enumerate}[\bf 1.]
  \item 
$\nu$ is a combinatorial isomorphism from the cell $\ccc[X]$ to the cell~$\ccc[Y]$.
\item 
If $R\in\ccc[X,Y]\cup\ccc[Y,X]$, then $R^\nu\in\ccc[Y]$.
\item 
If $R\in\ccc[X,Z]\cup\ccc[Z,X]$, where $Z\in F(\ccc)$ and $Z\ne X,Y$, then $R^\nu\in\ccc[Y,Z]\cup\ccc[Z,Y]$.
\item 
$\nu$ is a folding map from $\ccc$ to~$\ccc\setminus X$.
  \end{enumerate}
\end{lemma}

\begin{proof}
Note that $\nu$ maps each fiber of $\ccc$ bijectively onto a fiber of $\ccc\setminus X$.
Throughout this proof, we write $ab\approx a'b'$ in the case that
the pairs $ab$ and $a'b'$ are in the same basis relation of \ccc.

\textit{1.}
We show that $\nu$ takes the partition $\ccc[X]$ of $X^2$ onto
the partition $\ccc[Y]$ of $Y^2$. 
Let $x_1x_2$ and $x_3x_4$ 
be two, not necessarily disjoint, 
pairs of points of $X$. Denote $y_i=\nu(x_i)$ for $i\le4$. 
Assuming that 
\begin{equation}
  \label{eq:xxxx}
x_1x_2\approx x_3x_4,  
\end{equation}
we need to show that
\begin{equation}
  \label{eq:yyyy}
y_1y_2\approx y_3y_4.
\end{equation}
Since $x_1y_1$ is the only arrow in $M$ from $x_1$ and
$x_3y_3$ is the only arrow in $M$ from $x_3$, Property (\ref{item:C}) of a coherent configuration
allows us to deduce from \refeq{xxxx} that
\begin{equation}
  \label{eq:yxyx}
y_1x_2\approx y_3x_4.
\end{equation}
Since $x_2y_2$ is the only arrow in $M$ from $x_2$ and
$x_4y_4$ is the only arrow in $M$ from $x_4$, the relation \refeq{yxyx} along with Property (\ref{item:C})
implies that $y_2y_1\approx y_4y_3$, which implies~\refeq{yyyy} by Property (\ref{item:B}) of a coherent configuration.

\textit{2.}
Suppose that $R\in\ccc[X,Y]$ (the case $R\in\ccc[Y,X]$ is symmetric).
Note that $\nu$ induces a bijection from $X\times Y$ onto $Y^2$. Therefore,
it suffices to prove that for arbitrary $x_1,x_2\in X$ and $y_1,y_2\in Y$:
$$
x_1y_1\approx x_2y_2\iff x^\nu_1y_1\approx x^\nu_2y_2.
$$
This readily follows from the fact $x_1x^\nu_1$ is the only arrow in $M$
from $x_1$ and $x_2x^\nu_2$ is the only arrow in $M$ from~$x_2$.

\textit{3.}
Suppose that $R\in\ccc[X,Z]$ (the case of $R\in\ccc[Z,X]$ is symmetric).
Note that $\nu$ determines a bijection from $X\times Z$ onto $Y\times Z$.
We have to show that $\nu$ takes the partition $\ccc[X,Z]$ of $X\times Z$ onto
the partition $\ccc[Y,Z]$ of $Y\times Z$. Let $x_1,x_2\in X$, $y_i=x_i^\nu$ for $i=1,2$,
and $z_1,z_2\in Z$. Assuming that $x_1z_1\approx x_2z_2$, we have to prove
that $y_1z_1\approx y_2z_2$. This follows from the fact that $x_1y_1$ is the only 
arrow in $M$ from $x_1$ and $x_2y_2$ is the only arrow in $M$ from~$x_2$.

\textit{4.}
The fact that $\nu$ is a folding map from $\ccc$ to $\ccc\setminus X$
is a direct consequence of the three preceding parts
along with the obvious observation that $R^\nu=R$ for all $R\in\ccc\setminus X$.
\end{proof}

\begin{lemma}\label{lem:extension}
Suppose that an interspace $\ccc[X,Y]$ of a coherent configuration \ccc 
contains a matching basis relation $M$.
Suppose also that an interspace $\ccc'[X',Y']$ of a coherent configuration $\ccc'$ 
contains a matching basis relation $M'$.
If $f_0$ is an algebraic isomorphism from $\ccc\setminus X$ to $\ccc'\setminus X'$
such that $f_0(Y)=Y'$, then $f_0$ extends to an algebraic isomorphism $f$ from $\ccc$ to $\ccc'$.
\end{lemma}

\begin{proof}
Define a function $\nu\function{V(\ccc)}{V(\ccc)\setminus X}$ by
$\nu(x)=y$ for all $xy\in M$ and $\nu(z)=z$ for all $z\notin X$.
A function $\nu'\function{V(\ccc')}{V(\ccc')\setminus X'}$ is defined similarly.

Taking into account Lemma \ref{lem:refl}, let $Z'=f_0(Z)$ denote the fiber
of $\ccc'\setminus X'$ corresponding to a fiber $Z$ of $\ccc\setminus X$ under $f_0$.
We define a bijection $f\function\ccc{\ccc'}$ that coincides with $f_0$ on $\ccc\setminus X$
and bijectively maps
\begin{itemize}
\item 
$\ccc[X]$ onto $\ccc'[X']$,
\item 
$\ccc[X,Y]\cup\ccc[Y,X]$ onto $\ccc'[X',Y']\cup\ccc'[Y',X']$,
\item 
$\ccc[X,Z]\cup\ccc[Z,X]$ onto $\ccc'[Y',Z']\cup\ccc'[Z',Y']$ for each $Z\in F(\ccc)$, $Z\ne X,Y$.
\end{itemize}
Moreover, we require that
\begin{equation}
  \label{eq:fRnu}
f(R^\nu)=(f(R))^{\nu'}.  
\end{equation}
This condition uniquely determines $f$ because, by Parts 1--3 of Lemma \ref{lem:matching},
the mappings $\nu$ and $\nu'$ are bijective when restricted to each of the domains
listed above; see Figure~\ref{fig:diagram}.

\begin{figure}
  \centering
$$
\begin{CD}
R\in\ccc[X] @>{f}>>R^f\in\ccc'[X']@.\\
@V{\nu}VV @VV{\nu'}V@.\\
R^\nu\in\ccc[Y] @>>{f_0}>R^\star\in\ccc'[Y'],@.\quad R^\star=f_0(R^\nu)=\nu'(R^f).
\end{CD}
$$
  \caption{Proof of Lemma \ref{lem:extension}: 
This commutative diagram uniquely determines $R^f$ for each $R\in\ccc[X]$.
The similar commutative diagram holds also if $R\in\ccc[X,Z]$
for each $Z\ne X$ (as well as for $R\in\ccc[Z,X]$).}
  \label{fig:diagram}
\end{figure}

It is evident from the definition of $f$ that a triple $R,S,T$ is collocated in \ccc
if and only if the triple $R^f,S^f,T^f$ is collocated in $\ccc'$.
If $R,S,T\in\ccc$ are not collocated, we therefore have
$$
p^{T}_{RS}=0=p^{T^f}_{R^fS^f}.
$$

Assume that $R,S,T\in\ccc$ form a collocated triple.
By Part 4 of Lemma \ref{lem:matching}, $\nu$ and $\nu'$ are folding maps.
According to Lemma \ref{lem:cover},
\begin{equation}
  \label{eq:gg}
p^T_{RS}=p^{T^\nu}_{R^\nu S^\nu}  
\end{equation}
and
\begin{equation}
  \label{eq:gfgf}
p^{T^f}_{R^fS^f}=p^{\nu'(T^f)}_{\nu'(R^f)\nu'(S^f)},  
\end{equation}
the last equality being true because $R^f,S^f,T^f$ are collocated as well.
Combining Equalities \refeq{fRnu}--\refeq{gfgf} and using the assumption that
$f_0$ is an algebraic isomorphism from $\ccc\setminus X$ to $\ccc'\setminus X'$,
we conclude that
$$
p^{T}_{RS}=p^{T^\nu}_{R^\nu S^\nu}=p^{f_0(T^\nu)}_{f_0(R^\nu) f_0(S^\nu)}=
p^{f(T^\nu)}_{f(R^\nu) f(S^\nu)}=p^{\nu'(T^f)}_{\nu'(R^f) \nu'(S^f)}=p^{T^f}_{R^fS^f}.
$$
Therefore, $f$ is an algebraic isomorphism from \ccc to~$\ccc'$.
\end{proof}

A function $\nu$ in Lemma \ref{lem:matching} is a kind of projection of a coherent configuration $\ccc$
onto the smaller coherent configuration $\ccc\setminus X$ along a matching
in an interspace $\ccc[X,Y]$. We now consider a kind of the reverse lifting operation.

\begin{lemma}\label{lem:lift}
Let \ccd be a coherent configuration on point set $U$.
Let $\mu\function YX$ be a bijection, where $Y\in F(\ccd)$ and $X\cap U=\emptyset$.
Construct a rainbow $\ccc$ on the point set $U\cup X$ such that $\ccc[U]=\ccd$ as follows:
For each basis relation $R\in\ccd[Y]$, the partition $\ccc$ of $(U\cup X)^2$ contains
\begin{itemize}
\item
the image $R^\mu\subset X^2$ of $R$ under $\mu$;
\item 
the relation $\hat R\subset X\times Y$ defined by 
$$
y_1^\mu y_2\in\hat R\iff y_1y_2\in R;
$$
\item 
the transpose of $\hat R$.
\end{itemize}
Furthermore, for each $Z\in F(\ccd)$, $Z\ne Y$, and for each basis relation $R\in\ccd[Y,Z]$, 
the partition $\ccc$ contains
\begin{itemize}
\item 
the relation $\hat R\subset X\times Z$ defined by 
$$
y^\mu z\in\hat R\iff yz\in R;
$$
\item 
the transpose of $\hat R$.
\end{itemize}
Define a map $\nu\function{U\cup X}U$ to be the inverse $\mu^{-1}$ on $X$ and
the identity elsewhere. Then the following is true:
  \begin{enumerate}[\bf 1.]
  \item 
\ccc is a fibrous rainbow, and $\nu$ is a folding map from \ccc to \ccd.
\item 
\ccc is a coherent configuration.
  \end{enumerate}
\end{lemma}

\begin{proof}
\textit{1.}
The fact that \ccc is fibrous is a straightforward consequence of the construction.
It is also straightforward to see that $\nu$ preserves the fibers and their cardinalities.
Moreover, we have
\begin{itemize}
\item 
$\nu(R^\mu)=R$ for every $R\in\ccc[Y]$;
\item 
$\nu(\hat R)=R$ and $\nu((\hat R)^*)=R^*$ for every $R\in\ccc[Y]$;
\item 
$\nu(\hat R)=R$ and $\nu((\hat R)^*)=R^*$ for every $R\in\ccc[Y,Z]$.
\end{itemize}
This shows that $\nu$ is a folding map.

\textit{2.}
By Part 1 and Lemma \ref{lem:cover}.
\end{proof}

\begin{lemma}\label{lem:excl-matching}
Suppose that an interspace $\ccc[X,Y]$ of a coherent configuration \ccc
contains a matching basis relation $M$. If \ccc is separable, then $\ccc\setminus X$
is also separable.
\end{lemma}

\begin{proof}
Let $f_0$ be an algebraic isomorphism from $\ccc\setminus X$ to
a coherent configuration \ccd. According to Lemma \ref{lem:refl} and the notation introduced
in Section \ref{ss:ai-fibers}, let $Y'=f_0(Y)$ denote the fiber of \ccd corresponding
to the fiber $Y$ of \ccc. Fix a bijection
$\mu'\function{Y'}{X'}$, where $X'\cap V(\ccd)=\emptyset$.
Based on \ccd and $\mu'$, we construct a coherent configuration $\ccc'$
as described in Lemma \ref{lem:lift}. 
Note that, according to this construction, the interspace $\ccc'[X',Y']$ contains a matching basis
relation, namely $M'=\setdef{y^{\mu'}y}{y\in Y'}$ (indeed,
$M'=\hat D$ for $D=\setdef{yy}{y\in Y'}$).
By Lemma \ref{lem:extension}, $f_0$ extends to an
algebraic isomorphism $f$ from \ccc to $\ccc'$.
Since \ccc is separable, 
$f$ is induced by a combinatorial isomorphism $\phi\function{V(\ccc)}{V(\ccc')}$
from \ccc to $\ccc'$. Denote the restriction of $\phi$ to $V(\ccc)\setminus X$
by $\phi_0$. Note that $\phi(X)=f(X)=X'$. 
Therefore, the map $\phi_0$ is a combinatorial isomorphism
from $\ccc\setminus X$ to \ccd. Since $f$ is induced by $\phi$,
the algebraic isomorphism $f_0$ is induced by the combinatorial isomorphism~$\phi_0$.
\end{proof}

The converse of Lemma \ref{lem:excl-matching} is known to be true due to
Evdokimov and Ponomarenko \cite[Lemma 9.4]{EvdokimovP02}. 
Since their proof uses a matrix language, for the reader's convenience
we give an independent proof of this fact, stated as Lemma \ref{lem:excl-matching-conv} below.
Our argument is based on the following lemma, which essentially says that an algebraic
isomorphism of coherent configurations preserves matching basis relations and
respects projections along them.

\begin{lemma}\label{lem:fCC}
Suppose that an interspace $\ccc[X,Y]$ of a coherent configuration \ccc 
contains a matching basis relation $M$. If $f\function\ccc{\ccc'}$
is an algebraic isomorphism of coherent configurations, then the following is true:
\begin{enumerate}[\bf 1.]
\item 
$M'=f(M)$ is a matching basis relation in the interspace $\ccc'[X',Y']$,
where $X'=f(X)$ and $Y'=f(Y)$.
\item 
$f$ maps bijectively
\begin{enumerate}[(a)]
\item 
$\ccc[X]$ onto $\ccc'[X']$,
\item 
$\ccc[X,Y]$ onto $\ccc'[X',Y']$ and $\ccc[Y,X]$ onto $\ccc'[Y',X']$,
\item 
$\ccc[X,Z]$ onto $\ccc'[X',Z']$ and $\ccc[Z,X]$ onto $\ccc'[Z',X']$ for each $Z\in F(\ccc)$, $Z\ne X,Y$,
where $Z'=f(Z)$.
\end{enumerate}
\item 
$f$ satisfies the condition 
\begin{equation}
  \label{eq:fRnu-2}
(f(R))^{\nu'}=f(R^\nu),  
\end{equation}
where $\nu\function{V(\ccc)}{V(\ccc)\setminus X}$ is defined as the one-to-one map
from $X$ to $Y$ according to $M$ and as the identity map elsewhere, and 
$\nu'\function{V(\ccc')}{V(\ccc')\setminus X'}$ is defined similarly.
Thus, each of the bijections in Part 2(a)--(c) is uniquely determined by 
the restriction of $f$ to~$\ccc\setminus X$.
\end{enumerate}
\end{lemma}

\begin{proof}
\textit{1.}
According to the notation introduced in Section \ref{ss:ai-fibers}, $X'$ and $Y'$ are fibers of $\ccc'$.
Note that $R\in\ccc[X,Y]$ if and only if $p^D_{RR^*}>0$ for $D=\setdef{xx}{x\in X}$ and
$p^E_{R^*R}>0$ for $E=\setdef{yy}{y\in Y}$. 
Since an algebraic isomorphism preserves the intersection numbers and the pairs of mutually 
transposed basic relations (see \cite[Proposition 2.3.18]{Ponomarenko-book}),
we conclude that 
\begin{equation}
  \label{eq:RXY}
R\in\ccc[X,Y]\text{ if and only if }R^f\in\ccc'[X^f,Y^f].  
\end{equation}
In particular, this shows that $M'\in\ccc'[X',Y']$. Furthermore, $M'$ is a matching 
basis relation because an algebraic isomorphism preserves the valency of a basis relation 
(see \cite[Corollary 2.3.20]{Ponomarenko-book}).

\textit{2.}
By the general property \refeq{RXY} of an algebraic isomorphism.

\begin{figure}
  \centering
\begin{tikzpicture}[every node/.style={circle,draw,inner sep=2pt,fill=black},
lab/.style={draw=none,fill=none,inner sep=0pt,rectangle},
very thick,xscale=.6,yscale=.65]
  \begin{scope}[xshift=0mm,yshift=5mm]
    \path
       (0,0) node (y1) {}
       (4,0) node (y2) {}
      (-1,0) node (ylab)[lab] {$Y$}
      (0,4)  node (x1) {}
      (4,4)  node (x2) {}
      (-1,4) node (xlab)[lab] {$X$}
      (y1) edge[->] node[midway,below=3mm,lab] {$S$} (y2)
      (x1) edge[->] node[midway,above=3mm,lab] {$S^\mu$} (x2)
      (x1) edge[->] node[midway,left=1mm,lab] {$M$} (y1)
      (x1) edge[->] node[midway,above=3mm,lab] {$\hat{S}$} (y2)
      (x2) edge[->] node[midway,right=1mm,lab] {$M$} (y2)
      ;
\colclass{y1,y2} \colclass{x1,x2}
  \end{scope}
  \begin{scope}[xshift=10cm,yshift=5mm]
    \path
      (0,0) node (yp1) {}
      (4,0) node (yp2) {}
      (5,0) node (yplab)[lab] {$Y'$}
      (0,4) node (xp1) {}
      (4,4) node (xp2) {}
      (5,4) node (xplab)[lab] {$X'$}
      (yp1) edge[->] node[midway,below=3mm,lab] {$f(S)$} (yp2)
      (xp1) edge[->] node[midway,above=3mm,lab] {$f(S)^{\mu'}$} (xp2)
      (xp1) edge[->] node[midway,left=1mm,lab] {$M'$} (yp1)
      (xp1) edge[->] node[near start,right=3mm,lab] {$\widehat{f(S)}$} (yp2)
      (xp2) edge[->] node[midway,right=1mm,lab] {$M'=M^f$} (yp2)
      ;
\colclass{yp1,yp2} \colclass{xp1,xp2}
  \end{scope}
\end{tikzpicture}
\caption{Proof of Lemma \ref{lem:fCC}.}
\label{fig:fCC}
\end{figure}

\textit{3.}
If $R\in\ccc\setminus X$, then \refeq{fRnu-2} is trivially true
because $\nu$ is the identity on $\ccc\setminus X$ and $\nu'$ is the identity on $\ccc'\setminus X'$.
Thus, we have to consider the cases that $R$ belongs to the cell $\ccc[X]$ or to
one of the interspaces listed in Part 2(b)--(c).
Consider first the case that $R\in\ccc[X,Y]$. 
By Part 2 of Lemma \ref{lem:matching}, $\nu$ takes $\ccc[X,Y]$ bijectively
onto $\ccc[Y]$ and, similarly, $\nu'$ takes $\ccc'[X',Y']$ bijectively
onto $\ccc[Y']$. Given $S\in\ccc[Y]$, let $\hat S$ 
denote the basis relation in $\ccc[X,Y]$ such that $(\hat S)^\nu=S$.
The similar notation will be used also for $\ccc'$. 
Since $\widehat{R^\nu}=R$ for every $R\in\ccc[X,Y]$ (and similarly in $\ccc'$), 
we actually have to prove that
\begin{equation}
  \label{eq:hatS}
f(\hat S)=\widehat{f(S)}  
\end{equation}
for any $S\in\ccc[Y]$. 
Note that $T=\hat S$ if and only if $p^T_{MS}=1$; see Figure \ref{fig:fCC}. 
Similarly, $T'=\widehat{f(S)}$ if and only if
$p^{T'}_{M'f(S)}=1$. This implies Equality \refeq{hatS} as
$$
p^{f(\hat S)}_{M'f(S)}=p^{f(\hat S)}_{f(M)f(S)}=p^{\hat S}_{MS}=1.
$$

Suppose now that $R\in\ccc[X]$. 
By Part 1 of Lemma \ref{lem:matching}, $\nu$ takes $\ccc[X]$ bijectively
onto $\ccc[Y]$, and similarly $\nu'$ takes $\ccc'[X']$ bijectively
onto $\ccc'[Y']$. Let $\mu\function YX$ be the inverse
of $\nu$ and, similarly, $\mu'\function{Y'}{X'}$ be the inverse
of $\nu'$. Now we have to prove that
\begin{equation}
  \label{eq:muS}
f(S^\mu)=f(S)^{\mu'}
\end{equation}
for any $S\in\ccc[Y]$. 
Note that $T=S^\mu$ if and only if $p^{\hat S}_{TM}=1$. 
Similarly, $T'=f(S)^{\mu'}$ if and only if
$p^{\widehat{f(S)}}_{T'M'}=1$. Using \refeq{hatS}, this implies Equality \refeq{muS} as
$$
p^{\widehat{f(S)}}_{f(S^\mu)M'}=p^{f(\hat S)}_{f(S^\mu)f(M)}=p^{\hat S}_{S^\mu M}=1.
$$

The other cases are handled similarly.
\end{proof}

\begin{lemma}\label{lem:excl-matching-conv}
Suppose that an interspace $\ccc[X,Y]$ of a coherent configuration \ccc contains a matching
basis relation. If $\ccc\setminus X$ is separable, then \ccc is separable too.
\end{lemma}

\begin{proof}
Let $f$ be an algebraic isomorphism from \ccc to $\ccc'$.  
Let $X'=f(X)$ be the fiber of $\ccc'$ corresponding to the fiber $X$ of \ccc.
Denote the restriction of $f$ to $\ccc\setminus X$ by $f_0$.
By Lemma \ref{lem:fibers}, $f_0$ is an algebraic isomorphism
from $\ccc\setminus X$ to $\ccc'\setminus X'$. By the assumption
that $\ccc\setminus X$ is separable, $f_0$ is induced by a combinatorial
isomorphism $\phi_0\function{V(\ccc)\setminus X}{V(\ccc')\setminus X'}$.

Let $Y'=f(Y)$. By Part 1 of Lemma \ref{lem:fCC},
the basis relation $M'=M^f$ is a matching in the interspace $\ccc'[X',Y']$.
Define a bijection $\nu\function XY$, as usually, by $\nu(x)=y$ for all $xy\in M$
and a bijection $\nu'\function{X'}{Y'}$ similarly. Let us extend $\phi_0$
to a bijection $\phi\function{V(\ccc)}{V(\ccc')}$ by setting
$$
\phi=(\nu')^{-1}\circ\phi_0\circ\nu
$$
on $X$, that is, requiring the diagram
$$
\begin{CD}
X @>{\phi}>>X'\\
@V{\nu}VV @VV{\nu'}V\\
Y @>>{\phi_0}> Y'
\end{CD}
$$
be commutative.

Extend $\nu$ to the whole point set $V(\ccc)$ by the identity on $V(\ccc)\setminus X$
and, similarly, extend $\nu'$ to $V(\ccc')$ by the identity on $V(\ccc')\setminus X'$.
Recall the properties of $\nu$ and $\nu'$ established in Parts 1--3 of Lemma \ref{lem:matching}.
The definition of $\phi$ implies that
$$
\phi\circ\nu=\nu'\circ\phi.
$$
It follows that
$$
\phi(R)^{\nu'}=\phi(R^\nu)
$$
for every basis relation $R\in\ccc$. 
Comparing this with Equality \refeq{fRnu-2} in Lemma \ref{lem:fCC}, we see that
$f(R)=\phi(R)$ for all $R\in\ccc$ as a consequence of Parts 2--3 of this lemma.
We conclude that $f$ is induced by~$\phi$.
\end{proof}

Thus, the proof of Lemma \ref{lem:excl-matching-both} is complete.

\begin{corollary}[{cf.~\cite[Exercise 3.7.20]{Ponomarenko-book}}]\label{cor:size3}
Every coherent configuration \ccd with maximum fiber size at most 3 is separable.  
\end{corollary}

\begin{proof}
  By Lemma \ref{lem:coprime}, \ccd decomposes in a direct sum of indecomposable 
coherent configurations each with fibers of the same size which can be 3, or 2, or 1.
An indecomposable coherent configuration with maximum fibers size 1 is actually
a single-point configuration and, hence, is separable. 
Let $s\in\{2,3\}$ and suppose that \ccc is an indecomposable
coherent configuration with all fibers of the same size $s$.
As it is seen from Figure \ref{fig:interspace3}, every non-uniform interspace
of \ccc contains a matching basis relation. 
Lemma \ref{lem:excl-matching-both} reduces deciding separability of
a \ccc to deciding separability of a smaller coherent configuration.
Applying this reduction repeatedly, we see that \ccc is separable if and only if the association scheme
$\ccc[X]$ for the only remaining fiber $X$ is separable. 
The 2- and 3-point association schemes are shown
in Figure \ref{fig:cells}; all of them are separable. Therefore, \ccc is separable.
By Lemma \ref{lem:direct}, we conclude that \ccd is separable.
\end{proof}

Note that Corollary \ref{cor:size3}, along with Theorem \ref{thm:reduction},
implies the Immerman-Lander result \cite{ImmermanL90} that every graph of color multiplicity 
at most 3 is amenable.

\section{Cutting it down: 2-Point fibers}\label{s:excl2points}

Corollary \ref{cor:size3}, along with Lemmas \ref{lem:direct} and \ref{lem:coprime},
reduces deciding whether a coherent configuration \ccc with
maximum fiber size 4 is separable to the case that \ccc has fibers only of size 4 or 2.
By Lemma \ref{lem:excl-matching-both}, we can also assume
that no interspace of \ccc contains a matching basis relation.
The following lemma makes further reduction.

\begin{lemma}\label{lem:excl2points}
  Let \ccc be an indecomposable coherent configuration on more than 2 points 
with fibers only of size 4 or 2. 
Suppose that no interspace of \ccc contains a matching basis relation.
Let $X\in F(\ccc)$ with $|X|=2$. Under these conditions,
\ccc is separable if and only if $\ccc\setminus X$ is separable.
\end{lemma}

We remark that Lemma \ref{lem:excl2points} is applicable to the multipede graphs of color multiplicity at most 4
(see Section \ref{s:intro} and Remark \ref{rem:multipede}).
In this setting, Neuen and Schweitzer \cite[Section 4.2]{NeuenS17}
use the operation of removing vertex color classes of size 2 in order to reduce the number of vertices
in their construction of benchmark graphs challenging for practical isomorphism solvers.

To prove Lemma \ref{lem:excl2points}, we first collect some structural information.

\subsection{Direct and skewed connections of interspaces}

We use the notation introduced in Section \ref{ss:f-i}; see, in particular, Figure~\ref{fig:interspace4}.

\begin{lemma}\label{lem:X2Y4}\hfill
  \begin{enumerate}[\bf 1.]
  \item 
Suppose that $\ccc[X,Y]\simeq 2K_{1,2}$. If
$\ccc[X,Y]$ contains a relation $R=\{x_1y_1,x_1y_2,\allowbreak x_2y_3,x_2y_4\}$,
then $\ccc[Y]$ contains a basis relation $S=\{y_1y_2,y_2y_1,y_3y_4,y_4y_3\}$.
  \item 
Suppose that $\ccc[X,Y]\simeq 2K_{2,2}$. If
$\ccc[X,Y]$ contains a relation $R=\{x_1,x_2\}\times\{y_1,y_2\}\cup\{x_3,x_4\}\times\{y_3,y_4\}$,
then $\ccc[Y]$ contains a basis relation $S=\{y_1y_2,y_2y_1,\allowbreak y_3y_4,y_4y_3\}$.
  \item 
Suppose that $\ccc[X,Y]\simeq C_8$. If
$\ccc[X,Y]$ contains a relation $R=\{x_1y_1,x_2y_1,x_2y_2,\allowbreak x_3y_2,x_3y_3,x_4y_3,x_4y_4,x_1y_4\}$,
then $\ccc[Y]\simeq C_4$ and $S=\{y_1y_3,y_3y_1,y_2y_4,y_4y_2\}$ is one of the two irreflexive
basis relations of~$\ccc[Y]$.
  \end{enumerate}
\end{lemma}

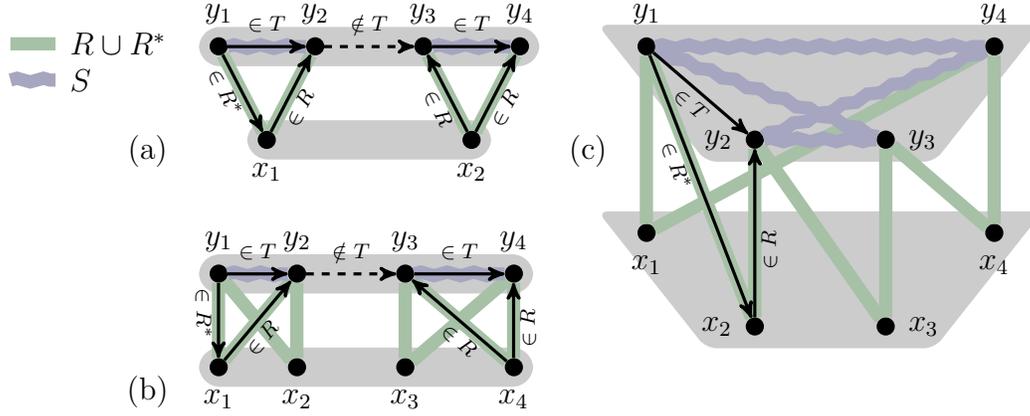
\begin{figure}
  \centering
\begin{tikzpicture}[every node/.style={circle,draw,black,
  inner sep=2pt,fill},very thick,
  elabel/.style={black,draw=none,fill=none,rectangle},
  every edge/.append style={every node/.append style={elabel}},
  int1/.style={edge node=
      {node [sloped,above=1mm]{\scriptsize $\in T$}}},
  int2/.style={edge node=
      {node [sloped,below]{\scriptsize $\in T$}}},
  inr/.style={edge node=
      {node [below,pos=0.4,sloped]{\scriptsize $\in R$}}},
  inrstar/.style={edge node=
      {node [pos=0.4,sloped,below]{\scriptsize $\in R^*$}}},
  nint/.style={dashed,edge node=
      {node [sloped,above=1mm]{\scriptsize $\notin T$}}},
  rrel/.style={ggrn,-,line width=5pt},
  srel/.style={gblu,-,line width=5pt,zz},
  be/.style={label position=below},
  bl/.style={label position=below left},
  ri/.style={label position=right},
  le/.style={label position=left},
  ]
  
  \matrixgraph[name=m1]{&[4mm] &[4mm] &[12mm] &[4mm] &[4mm]\\
    y_1 &         & y_2 & y_3 &         & y_4\\[1cm]
        & x_1[be] &     &     & x_2[be]      \\
  }{
    x_1 -- [rrel] {y_1,y_2};
    x_2 -- [rrel] {y_3;y_4};
    {y_1,y_3} --[srel,matching] {y_2,y_4};
    y_1->[int1] y_2;
    y_2->[nint] y_3;
    y_1 ->[inrstar] x_1;
    x_1 ->[inr] y_2;
    x_2 ->[inr] y_4;
    y_3->[int1] y_4;
    x_2 ->[inr] y_3;
  };
  \colclass{x_1,x_2}
  \colclass{y_1,y_2,y_3,y_4}
  
  \matrixgraph[name=m2,at={($(m1.south west)-(0mm,15mm)$)}]
  {&[8mm] &[12mm] &[12mm] &[4mm]\\
    y_1     & y_2     & y_3     & y_4\\[1cm]
    x_1[be] & x_2[be] & x_3[be] & x_4[be]      \\
  }{
    {x_1,x_2} -- [rrel] {y_1,y_2};
    {x_3,x_4} -- [rrel] {y_3;y_4};
    {y_1,y_3} --[srel,matching] {y_2,y_4};
    y_1->[int1] y_2;
    y_3->[int1] y_4;
    y_2->[nint] y_3;
    y_1 ->[inrstar] x_1;
    x_1 ->[inr] y_2;
    x_4 ->[inr] y_4;
    x_4 ->[inr] y_3;
  };
  \colclass{x_1,x_2,x_3,x_4}
  \colclass{y_1,y_2,y_3,y_4}
    
  \matrixgraph[name=m3,at={($(m1.north east)+(14mm,0mm)$)}]
    {&[12mm] &[15mm] &[12mm]\\
    y_1     &         &         &    y_4\\[10mm]
            & y_2[le] & y_3[ri] &        \\[10mm]
    x_1[be] &         &         & x_4[be]\\[10mm]
            & x_2[le] & x_3[ri] &        \\
  }{
    x_1 ->[rrel] {y_1,y_4};
    x_2 ->[rrel] {y_1,y_2};
    x_3 ->[rrel] {y_2,y_3};
    x_4 ->[rrel] {y_3,y_4};
    {y_1,y_2} --[srel] {y_3,y_4};
    y_1 ->[int2] y_2;
    y_1 ->[inrstar] x_2;
    x_2 ->[inr] y_2;
  };
  \colclasstrpz{y_1}{y_4}{y_3}{y_2}
  \colclasstrpz{x_1}{x_4}{x_3}{x_2}

  \legend[anchor=north east,at={($(m1.north west)+(-4mm,+2mm)$)}]{
    \legendrow{rrel}{$R\cup R^*$}
    \legendrow{srel}{$S$}
  };
  \node[name=a,elabel,at={($(m1.south west)+(-8mm,0mm)$)}] {(a)};
  \node[name=b,elabel,at={($(a)+(0pt,-32mm)$)}] {(b)};
  \node[name=b,elabel,at={($(a.east)+(55mm,0pt)$)}] {(c)};
\end{tikzpicture}
\caption{Proof of Lemma \ref{lem:X2Y4}: (a) Part 1; (b) Part 2; (c) Part 3.}
\label{fig:proofX2Y4}
\end{figure}

\begin{proof}
\textit{1.}  
Let $T$ be the basis relation of $\ccc[Y]$ containing the arrow $y_1y_2$.
We have $T\subseteq S$ because $p^T_{R^*R}>0$.
For example, $y_2y_3\notin T$ because  $y_1y_2$ extends to $y_1x_1y_2$
and $y_2y_3$ cannot be extended to a triangle of this kind. 
On the other hand, $S\subseteq T$ because $p^R_{RT}>0$.
For example, $y_3y_4\in T$ because otherwise, while $x_1y_2$ 
extends to $x_1y_1y_2$, the pair $x_2y_4$ could not 
be extended to a triangle of this kind; see Figure~\ref{fig:proofX2Y4}(a).

\textit{2.}  
Literally the same argument (but with $x_4y_4$ instead of
$x_2y_4$) applies also for this part; see Figure~\ref{fig:proofX2Y4}(b).

\textit{3.}  
Again, let $T$ be the basis relation containing the arrow $y_1y_2$.
We have $T\cap S=\emptyset$ because $p^T_{R^*R}>0$. Furthermore,
$p^R_{RT}>0$, and this implies that $T$ is exactly the irreflexive
complement of $S$ in $Y^2$; see Figure~\ref{fig:proofX2Y4}(c).
\end{proof}

The following definitions play an important role not only in the proof of Lemma \ref{lem:excl2points}
but also in the subsequent sections.
In the context of Lemma \ref{lem:X2Y4}, we say that $\ccc[X,Y]$
\emph{determines} a matching basis relation in $\ccc[Y]$ (namely $\{y_1y_2,y_2y_1,y_3y_4,y_4y_3\}$ in Parts 1--2
and $\{y_1y_3,y_3y_1,y_2y_4,y_4y_2\}$ in Part 3).
Suppose that $\ccc[X,Y]$ determines a matching $M$ in $Y$,
and $\ccc[Z,Y]$ determines a matching $M'$ in $Y$.
We say that $\ccc[X,Y]$ and $\ccc[Z,Y]$ have a \emph{direct connection at} 
$Y$ if $M=M'$. If $M\ne M'$, we say that
 $\ccc[X,Y]$ and $\ccc[Z,Y]$ have a \emph{skewed connection at} $Y$
(or are, respectively, \emph{directly} or \emph{askew connected at}~$Y$).

We conclude this subsection with a lemma that provides an important information on the structure
of matching-free coherent configurations.

\begin{lemma}[Transitivity of direct $2K_{2,2}$-connections]\label{lem:trans}
If $\ccc[X,Y]\simeq2K_{2,2}$ and $\ccc[Z,Y]\simeq2K_{2,2}$ are directly connected
at $Y$, then either $\ccc[X,Z]$ contains a matching basis relation or
$\ccc[X,Z]\simeq2K_{2,2}$ and the connections between $\ccc[Z,X]$ and $\ccc[Y,X]$
at $X$ and between $\ccc[X,Z]$ and $\ccc[Y,Z]$ at $Z$ are direct.
\end{lemma}

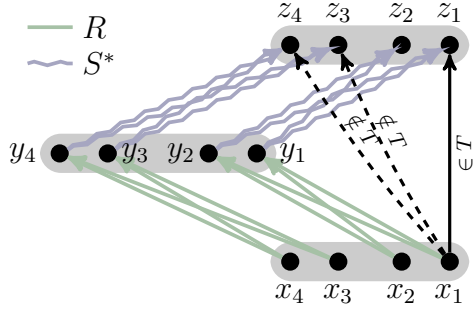
\begin{figure}
\centering
\begin{tikzpicture}[every node/.style={circle,draw,black,
  inner sep=2pt,fill=black},very thick,
  elabel/.style={black,draw=none,fill=none,rectangle},
  every edge/.append style={every node/.append style={elabel}},
  lab/.style={draw=none,fill=none,inner sep=0pt,rectangle},
  ab/.style={label position=above},
  br/.style={label position=below right},
  ri/.style={label position=right},
  le/.style={label position=left},
  cl/.style={label distance=0mm,label position=left},
  cr/.style={label distance=0mm,label position=right},
  nl/.style={nolabel},
  rrel/.style={ggrn,line width=1.4pt},
  srel/.style={gblu,line width=1.4pt,zz},
  trel/.style={gred,-,line width=1.4pt,bps},
  int/.style={edge node=
      {node [sloped,below]{\scriptsize $\in T$}}},
  nint/.style={dashed,bend left=5,edge node=
      {node [sloped,above,pos=0.6]{\scriptsize $\notin T$}}},
]
  \matrixgraph[name=m1,label position=below]
    {&[4mm]&[11mm]&[4mm]&[2mm]&[4mm]&[6mm]&[4mm]\\
    &&&& z_4[ab] & z_3[ab] & z_2[ab] & z_1[ab] \\[12mm]
         y_4[le] & y_3[cr] & y_2[cl] & y_1[ri] \\[12mm]
    &&&& x_4     & x_3     & x_2     & x_1     \\
  }{
    {x_1,x_2} ->[rrel] {y_1,y_2} ->[srel] {z_1,z_2};
    {x_3,x_4} ->[rrel] {y_3,y_4} ->[srel] {z_3,z_4};
    x_1 ->[int] z_1;
    x_1 ->[nint] {z_3,z_4};
  };
  \colclass{x_1,x_2,x_3,x_4}
  \colclass{y_1,y_2,y_3,y_4}
  \colclass{z_1,z_2,z_3,z_4}
  \legend[anchor=north west,at={($(m1.north west)+(-4mm,4mm)$)}]{
    \legendrow{rrel}{$R$}
    \legendrow{srel}{$S^*$}
  };
\end{tikzpicture}
\caption{Proof of Lemma \ref{lem:trans}.}
\label{fig:proofTrans}
\end{figure}

\begin{proof}
  Fix basis relations $R\in\ccc[X,Y]$ and $S\in\ccc[Z,Y]$. By assumption, they determine the same matching
in $\ccc[Y]$. Specifically, let 
$$
R={\{x_1,x_2\}\times\{y_1,y_2\}}\cup{\{x_3,x_4\}\times\{y_3,y_4\}}
$$
and
$$
S={\{z_1,z_2\}\times\{y_1,y_2\}}\cup{\{z_3,z_4\}\times\{y_3,y_4\}};
$$
see Figure \ref{fig:proofTrans}.
Let $T\in\ccc[X,Z]$ be the basis relation containing the arrow $x_1z_1$.
Note that $p^T_{RS^*}=2$. This equality implies that $T$ contains
neither $x_1z_3$ nor $x_1z_4$. Therefore, the valency of $T$ is either
1 or 2. In the former case $T$ is a matching basis relation.
Suppose that $d(T)=2$ and, hence, $x_1z_2\in T$. It remains to exclude
the possibility that $T$ is of $C_8$-type.

Using the fact that $z_2x_1\in T^*$ and repeating the same argument as above,
we conclude that also $z_2x_2\in T^*$. This yields $x_2z_2\in T$ and, repeating the same argument,
we also derive $x_2z_1\in T$. It follows that $T$ contains the set $\{x_1,x_2\}\times\{z_1,z_2\}$
and, hence, is of $2K_{2,2}$-type.

Finally, note that $T$ and $S^*$ determine the same matching in $\ccc[Z]$
and that $T^*$ and $R^*$ determine the same matching in~$\ccc[X]$.
\end{proof}

\subsection{Subconfigurations $\ccc[X\cup Y\cup Z]$ with $\ccc[X,Y]\simeq 2K_{1,2}$}

\begin{lemma}\label{lem:Z4}
Let $X,Y,Z\in F(\ccc)$ with $|X|=2$ and $|Y|=|Z|=4$.
Suppose that $\ccc[X,Y]\simeq 2K_{1,2}$.
\begin{enumerate}[\bf 1.]
\item 
If $\ccc[Z,Y]\simeq 2K_{2,2}$ with direct connection to $\ccc[X,Y]$ at $Y$,
then $\ccc[X,Z]\simeq 2K_{1,2}$ with direct connection to $\ccc[Y,Z]$ at $Z$.
\item 
If $\ccc[Z,Y]\simeq 2K_{2,2}$ with skewed connection to $\ccc[X,Y]$ at $Y$,
then $\ccc[X,Z]$ is uniform.
\item 
If $\ccc[Z,Y]\simeq C_8$, then $\ccc[X,Z]$ is uniform.
\item 
If $\ccc[Z,Y]$ is uniform, then $\ccc[X,Z]$ is uniform too.
\end{enumerate}
\end{lemma}

\begin{figure}[t]
  \centering
\begin{tikzpicture}[every node/.style={circle,draw,black,
  inner sep=2pt,fill=black},very thick,
  elabel/.style={black,draw=none,fill=none,rectangle},
  every edge/.append style={every node/.append style={elabel}},
  lab/.style={draw=none,fill=none,inner sep=0pt,rectangle},
  be/.style={label position=below,label distance=2mm},
  br/.style={label position=below right},
  ri/.style={label position=right},
  le/.style={label position=left},
  nl/.style={nolabel},
  rrel/.style={ggrn,-,line width=3pt},
  srel/.style={gblu,-,line width=3pt,zz},
]
  \begin{scope}[xshift=0mm,yshift=0mm,
    inr/.style={edge node=
      {node [below,pos=0.33,sloped]{\scriptsize $\in R$}}},
    ins/.style={edge node=
      {node [sloped,below=0.7mm]{\scriptsize $\in S$}}},
    int/.style={edge node=
      {node [sloped,below ]{\scriptsize $\in T$}}},
    nint/.style={dashed,edge node=
      {node [sloped,above,pos=0.72]{\scriptsize $\notin T$}}},]
    \path (0,0) node (tag)[lab] {(1)};
    \matrixgraph[name=m1,label position=right,
      at={($(tag)+(4mm,0mm)$)}]
      {&[3mm]&[6mm]&[3mm]&[-1mm]&[8mm]&[-1mm]&[3mm]&[6mm]&[3mm]\\
      &&& y_4[nl] &&& z_4    \\[5mm]
      && y_3[nl] &&&&& z_3    \\[8mm]
      & y_2[nl] &&&&&&& z_2[nl]\\[5mm]
      y_1[le] &&&&&&&&& z_1    \\[8mm]
      &&&&    x_1[be]&x_2[be]        \\
    }{
      {y_4,y_3} --[rrel] {z_4,z_3};
      {y_2,y_1} --[rrel] {z_2,z_1};
      {y_4,y_3} --[srel] {x_2};
      {y_2,y_1} --[srel] {x_1};
      {x_1} ->[nint] {z_4,z_3};
      x_1 ->[int] z_1;
      y_1 ->[inr] z_1;
      x_1 ->[ins] y_1;
    };
    \colclass{x_1,x_2}
    \colclass[52]{y_1,y_2,y_3,y_4}
    \colclass[-52]{z_1,z_2,z_3,z_4}
    \legend[anchor=north west,at={($(m1.north east)+(-12mm,+4mm)$)}]{
      \legendrow{rrel}{$R\cup R^*$}
      \legendrow{srel}{$S\cup S^*$}
    };
  \end{scope}
  \begin{scope}[xshift=7.6cm,yshift=0mm,
    inr/.style={edge node=
      {node [above,pos=0.8,sloped]{\scriptsize $\in R$}}},
    insstar/.style={edge node=
      {node [sloped,pos=0.42,below=1mm]{\scriptsize $\in S^*$}}},
    int/.style={edge node=
      {node [sloped,below ]{\scriptsize $\in T$}}},
    ]
    \path (0,0) node (tag)[lab] {(2)};
    \matrixgraph[name=m2,label position=right,
      at={($(tag)+(4mm,0mm)$)}]
      {&[3mm]&[6mm]&[3mm]&[-1mm]&[8mm]&[-1mm]&[3mm]&[6mm]&[3mm]\\
      &&& y_4[nl] &&& z_4[nl]\\[5mm]
      && y_3[nl] &&&&& z_3[nl]\\[8mm]
      & y_2[le] &&&&&&& z_2[nl]\\[5mm]
       y_1[le] &&&&&&&&& z_1    \\[8mm]
      &&&&    x_1[be]&x_2[be]        \\
    }{
      {y_2,y_1} --[rrel] {z_2,z_1};
      {y_4,y_2} --[srel] {x_2};
      {y_3,y_1} --[srel] {x_1};
      {y_4,y_3} --[rrel] {z_4,z_3};
      x_2 ->[int] z_1;x_1->z_1; 
      y_1 -> z_1;y_2 ->[inr] z_1;
      {y_1,y_2} ->[matching,insstar] {x_1,x_2};
    };
    \colclass{x_1,x_2}
    \colclass[52]{y_1,y_2,y_3,y_4}
    \colclass[-52]{z_1,z_2,z_3,z_4}
    \legend[anchor=north west,at={($(m2.north east)+(-12mm,+4mm)$)}]{
      \legendrow{rrel}{$R\cup R^*$}
      \legendrow{srel}{$S\cup S^*$}
    };
  \end{scope}
  \begin{scope}[xshift=0cm,yshift=-5.5cm,
    inr/.style={edge node=
      {node [below,pos=0.6]{\scriptsize $\in R$}}},
    insstar/.style={edge node=
      {node [sloped,pos=0.5,below=0.3mm]{\scriptsize $\in S^*$}}},
    int/.style={edge node=
      {node [sloped,above,pos=0.3]{\scriptsize $\in T$}}},
    ]
    \path (0,0) node (tag)[lab] {(3)};
    \matrixgraph[name=m3,label position=below left,
      label distance=0pt,
      every label/.append style={font=\scriptsize},
      at={($(tag)+(4mm,0mm)$)}]
      {&[3mm]&[6mm]&[3mm]&[-1mm]&[8mm]&[-1mm]&[3mm]&[6mm]&[3mm]\\
      &&&   1     &&& 8[br]\\[5mm]
      &&   5     &&&&& 4[br]\\[8mm]
      &   7     &&&&&&& 6[br]\\[5mm]
         3     &&&&&&&&& 2[br]\\[8mm]
      &&&&    x_1[be]&x_2[be]        \\
    }{
      {1,3,5,7} --[matching,rrel] {2,4,6,8};
      {1,3,5,7} --[matching,rrel] {8,2,4,6};
      x_1 --[srel] {7,3};
      x_2 --[srel] {1,5};
    };
    \alsolabel{7/left/y_1,5/left/y_2,6/right/z_1}
    \graph[use existing nodes]{
      x_2 ->[int] z_1;x_1->z_1; 
      y_1 -> z_1;y_2 ->[inr] z_1;
      {y_1,y_2} ->[matching,insstar] {x_1,x_2};
    };
    \colclass{x_1,x_2}
    \colclass[52]{1,5,7,3}
    \colclass[-52]{2,4,6,8}
    \legend[anchor=north west,at={($(m3.north east)+(-12mm,+4mm)$)}]{
      \legendrow{rrel}{$R\cup R^*$}
      \legendrow{srel}{$S\cup S^*$}
      \node[anchor=west,inner sep=0pt,fill=none,draw=none,
        rectangle,font=\scriptsize] {$1\,...\,8$};\&
      \node[anchor=west,inner sep=0pt,fill=none,draw=none,
        rectangle,font=\scriptsize,align=left] 
        {\,\,order on $C_8$};\\
    };
  \end{scope}
  \begin{scope}[xshift=7.6cm,yshift=-5.5cm,
    inr/.style={edge node=
      {node [below=1mm,pos=0.7]{\scriptsize $\in R$}}},
    insstar/.style={edge node=
      {node [sloped,pos=0.33,below=0.3mm]{\scriptsize $\in S^*$}}},
    int/.style={edge node=
      {node [sloped,below]{\scriptsize $\in T$}}},
    ]
    \path (0,0) node (tag)[lab] {(4)};
    \matrixgraph[name=m4,label position=right,
      at={($(tag)+(4mm,0mm)$)}]
      {&[3mm]&[6mm]&[3mm]&[-1mm]&[8mm]&[-1mm]&[3mm]&[6mm]&[3mm]\\
      &&& y_4[nl] &&& z_4[nl]\\[5mm]
      && y_2[le] &&&&& z_3[nl]\\[8mm]
      & y_3[nl] &&&&&&& z_2[nl]\\[5mm]
       y_1[le] &&&&&&&&& z_1    \\[8mm]
      &&&&    x_1[be]&x_2[be]        \\
    }{
      {y_4,y_2} --[srel] {x_2};
      {y_3,y_1} --[srel] {x_1};
      x_2 ->[int] z_1;x_1->z_1; 
      y_1 -> z_1;y_2 ->[inr] z_1;
      {y_1,y_2} ->[matching,insstar] {x_1,x_2};
    };
    \colclass{x_1,x_2}
    \colclass[52]{y_1,y_2,y_3,y_4}
    \colclass[-52]{z_1,z_2,z_3,z_4}
    \legend[anchor=north west,at={($(m4.north east)+(-12mm,+4mm)$)}]{
      \legendrow{srel}{$S\cup S^*$}
    };
  \end{scope}
\end{tikzpicture}
\caption{Proof of Lemma \ref{lem:Z4}. An auxiliary enumeration of the points in $Y\cup Z$
corresponds to the 8-cycle underlying the basis relation $R$ in Part~3.}
\label{fig:proofZ4}
\end{figure}
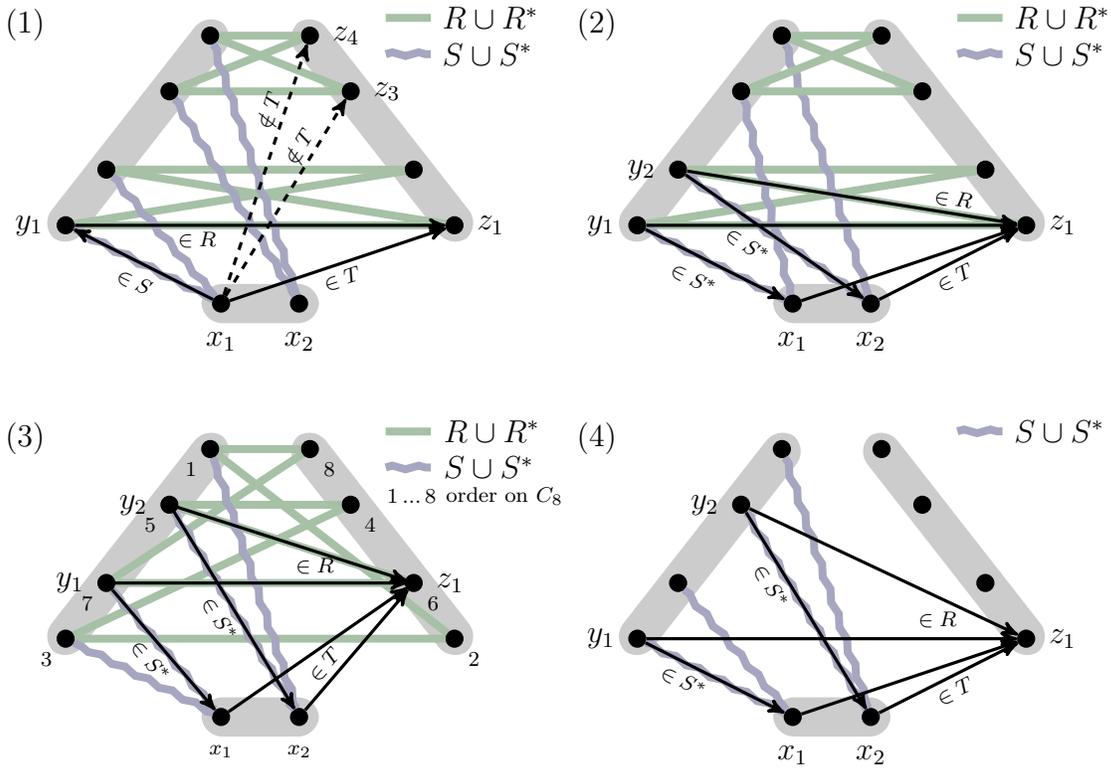

\begin{proof}
Fix a basis relation $S\in\ccc[X,Y]$. Also, fix a basis relation $R\in\ccc[Y,Z]$;
Figure \ref{fig:proofZ4} gives the corresponding picture for each part of the lemma.
We will refer to the points of  $\ccc[X\cup Y\cup Z]$ using these pictures.

\textit{1.}  
Let $T$ be the basis relation of $\ccc[X,Z]$ containing the arrow $x_1z_1$.
It suffices to show that neither $x_1z_3$ nor $x_1z_4$ is in $T$.
This follows from the fact that $p^T_{SR}>0$.

\textit{2.}  
It suffices to prove that all the arrows from $z_1$ to $X$, i.e.,
$z_1x_1$ and $z_1x_2$, are in the same basis relation.
We will prove that the transposed arrows $x_1z_1$ and $x_2z_1$ are in the same basis relation.
Denote the basis relation containing $x_1z_1$ by $T$.
Looking at the triple $y_1x_1z_1$, we see that $p^R_{S^*T}>0$.
Since $y_2z_1\in R$ and $x=x_2$ is the only point such that $y_2x\in S^*$,
we conclude that $x_2z_1\in T$.

\textit{3.}  
By Parts 1 and 3 of Lemma \ref{lem:X2Y4}, each of the interspaces $\ccc[X,Y]$ and $\ccc[Z,Y]$
determines a matching basis relation in the cell $\ccc[Y]$. Moreover, Part 3 of Lemma \ref{lem:X2Y4}
implies that $\ccc[Y]$ contains a unique matching relation. Therefore, the
connection of $\ccc[X,Y]$ and $\ccc[Z,Y]$ at $Y$ is exactly as shown in Figure \ref{fig:proofZ4}.
The rest of the proof is literally the same as in the preceding part,
even though the points $y_1$, $y_2$, and $z_1$ now have a different position, as shown in Figure~\ref{fig:proofZ4}(3).

\textit{4.}  
The proof is literally the same as in Part 2,
where the new position of the points $y_1$, $y_2$ is shown
in Figure~\ref{fig:proofZ4}(4).
\end{proof}

\subsection{Proof of Lemma \ref{lem:excl2points}}
  Since \ccc is indecomposable, it contains a non-uniform interspace $\ccc[X,Y]$.
It is impossible that $|Y|=2$ because then $\ccc[X,Y]$ would contain a matching;
see Figure \ref{fig:interspace3}.
It follows that $|Y|=4$. We conclude that $\ccc[X,Y]\simeq 2K_{1,2}$,
as this is the only, up to isomorphism, non-uniform
interspace between two fibers of sizes 2 and 4; see Figure~\ref{fig:interspace4}.

To be specific, let $X=\{x_1,x_2\}$ and $Y=\{y_1,y_2,y_3,y_4\}$, and
suppose that $\ccc[X,Y]$ consists of the relation
\begin{equation}
  \label{eq:CXY}
T={\{x_1\}\times\{y_1,y_2\}}\cup{\{x_2\}\times\{y_3,y_4\}}
\end{equation}
and its complement ${X\times Y}\setminus T$.
We now proceed to proving that \ccc is separable if and only if
$\ccc\setminus X$ is separable.

\smallskip

($\implies$)
Let $f_0$ be an algebraic isomorphism from $\ccc\setminus X$ to
a coherent configuration \ccd. It suffices to extend \ccd to a
coherent configuration $\ccc'$ such that $f_0$ extends to an
algebraic isomorphism $f$ from \ccc to $\ccc'$. Indeed,
since \ccc is separable, such an $f$ is induced by a combinatorial isomorphism
$\phi\function{V(\ccc)}{V(\ccc')}$, and the restriction of $\phi$
to $V(\ccc)\setminus X$ will give us a combinatorial isomorphism
from $\ccc\setminus X$ to \ccd inducing~$f_0$.

We construct $\ccc'$ and $f$ as follows. First of all,
take $X'=\{x'_1,x'_2\}$ such that $X'\cap V(\ccd)=\emptyset$.
This will be a fiber of $\ccc'$. The 2-point association scheme $\ccc'[X']$ is unique.
The map $f$ is defined on the two basis relations of the cell $\ccc[X]$ uniquely
in an obvious way: It maps the (ir)reflexive
relation of $\ccc[X]$ to the (ir)reflexive relation of $\ccc'[X']$.

Let $Y'=f_0(Y)$ denote the fiber of \ccd corresponding to the fiber $Y$ of $\ccc\setminus X$ under $f_0$. 
It follows from \refeq{CXY} by Part 1 of Lemma \ref{lem:X2Y4},
that $\ccc[X,Y]$ determines the matching basis relation 
\begin{equation}
  \label{eq:Myyy}
M=\{y_1y_2,y_2y_1,y_3y_4,y_4y_3\}  
\end{equation}
in $\ccc[Y]$. Note that $M'=f_0(M)$ is
a matching basis relation in $\ccc'[Y']$ (cf.\ the proof of Part 1 of Lemma \ref{lem:fCC}). 
To be specific, let $Y'=\{y'_1,y'_2,y'_3,y'_4\}$ and
$$
M'=\{y'_1y'_2,y'_2y'_1,y'_3y'_4,y'_4y'_3\}.
$$
We set the interspace $\ccc'[X',Y']$ to be the partition of $X'\times Y'$ into the relations
$$
T'={\{x'_1\}\times\{y'_1,y'_2\}}\cup{\{x'_2\}\times\{y'_3,y'_4\}}
$$
and ${X'\times Y'}\setminus T'$.
Furthermore, we set $f(T)=T'$ and $f({X\times Y}\setminus T)={X'\times Y'}\setminus T'$.
This will define $f$ also on $\ccc[Y,X]$ according to the general property
$f(R^*)=f(R)^*$ of an algebraic isomorphism.
Note that the extension $f$, as defined so far, is an algebraic isomorphism
from $\ccc[X\cup Y]$ to the current fragment $\ccc'[X'\cup Y']$ of~$\ccc'$.

Given $Z\in F(\ccc\setminus X)$, let $Z'=f_0(Z)$.
It remains, for each $Z\ne Y$, to construct $\ccc'[X',Z']$ and to define $f$
locally as a bijection from $\ccc[X,Z]$ to $\ccc'[X',Z']$.
If $\ccc[X,Z]$ is uniform, which is always the case when $|Z|=2$, we set $\ccc'[X',Z']$ also to be uniform,
and correspondingly define $f(X\times Z)=X'\times Z'$. Assume now that $\ccc[X,Z]$ is non-uniform
and, hence, $|Z|=4$ and $\ccc[X,Z]\simeq2K_{1,2}$.

In this case, Lemma \ref{lem:Z4} implies that $\ccc[Z,Y]\simeq 2K_{2,2}$ and that this interspace
is directly connected to $\ccc[X,Y]$ at $Y$. 
Fix a basis relation $S_Z\in\ccc[Z,Y]$. Fix an enumeration $z_1,z_2,z_3,z_4$ of the points of $Z$ such that
\begin{equation}
  \label{eq:SZ}
S_Z={\{z_1,z_2\}\times\{y_1,y_2\}}\cup{\{z_3,z_4\}\times\{y_3,y_4\}}.
\end{equation}
Since $f_0$ is an algebraic isomorphism from $\ccc\setminus X$ to $\ccd$,
we have $\ccd[Z',Y']\simeq2K_{2,2}$, and this interspace determines the matching $M'$ in $\ccd[Y']$.
Therefore, the points of $Z'$ can be enumerated so that
\begin{equation}
  \label{eq:fSZ}
f_0(S_Z)={\{z'_1,z'_2\}\times\{y'_1,y'_2\}}\cup{\{z'_3,z'_4\}\times\{y'_3,y'_4\}},
\end{equation}
and we fix such an enumeration $z'_1,z'_2,z'_3,z'_4$. 
Part 1 of Lemma \ref{lem:Z4} implies that $\ccc[X,Z]$ is directly connected
to $\ccc[Y,Z]$ at $Z$ and, therefore, consists of the relation
\begin{equation}
  \label{eq:RZ}
R_Z={\{x_1\}\times\{z_1,z_2\}}\cup{\{x_2\}\times\{z_3,z_4\}}  
\end{equation}
and its complement ${X\times Z}\setminus R_Z$.
We, therefore, define the interspace $\ccc'[X',Z']\simeq2K_{1,2}$ by
requiring that it determines the matching $\{z'_1z'_2,z'_2z'_1,z'_3z'_4,z'_4z'_3\}$
in $\ccc'[Z']$.
This condition defines $\ccc'[X',Z']$ uniquely and ensures that 
$\ccc'[X'\cup Y'\cup Z']$ is a coherent configuration.
There still remain two different possibilities to define $f$ on $\ccc[X,Z]$.
Our choice is this: We set 
\begin{equation}
  \label{eq:fRZ}
f(R_Z)={\{x'_1\}\times\{z'_1,z'_2\}}\cup{\{x'_2\}\times\{z'_3,z'_4\}}.
\end{equation}
Note that the basis relation $R_Z$ is chosen in \refeq{RZ} in such a way that
the set $T\cup T^*\cup S_Z\cup S_Z^*\cup R_Z\cup R_Z^*$, seen as a symmetric irreflexive relation,
forms a graph with two connected components, one contaning $x_1$ and the other containing $x_2$.
Likewise, the basis relation $f(R_Z)$ is defined by \refeq{fRZ} so that
$T'\cup (T')^*\cup f(S_Z)\cup f(S_Z)^*\cup f(R_Z)\cup f(R_Z)^*$ also has two connected components.
This ensures that $f$ is an algebraic isomorphism from $\ccc[X\cup Y\cup Z]$
to $\ccc'[X'\cup Y'\cup Z']$, see Figure~\ref{fig:proof-excl2points}.
\begin{figure}[t]
\centering
\begin{tikzpicture}[every node/.style={circle,draw,black,
  inner sep=2pt,fill=black},very thick,
  elabel/.style={black,draw=none,fill=none,rectangle},
  every edge/.append style={every node/.append style={elabel}},
  lab/.style={draw=none,fill=none,inner sep=0pt,rectangle},
  ab/.style={label position=above},
  br/.style={label position=below right},
  ri/.style={label position=right},
  le/.style={label position=210},
  nl/.style={nolabel},
  rrel/.style={ggrn,-,line width=1.4pt},
  srel/.style={gblu,-,line width=1.4pt,zz},
  trel/.style={gred,-,line width=1.4pt,bps},
]
  \matrixgraph[name=m1,label position=below]
    {&[3mm]&[6mm]&[3mm]&[-1mm]&[8mm]&[-1mm]&[3mm]&[6mm]&[3mm]\\
    &&&&    x_1[ab]&x_2[ab]        \\[8mm]
     z_1 &&&&&&&&& y_1    \\[5mm]
    & z_2 &&&&&&& y_2    \\[8mm]
    && z_3 &&&&& y_3    \\[5mm]
    &&& z_4 &&& y_4    \\
  }{
    {y_4,y_3} --[srel] {z_4,z_3};
    {y_2,y_1} --[srel] {z_2,z_1};
    {y_4,y_3} --[trel] {x_2};
    {y_2,y_1} --[trel] {x_1};
    {z_4,z_3} --[rrel] {x_2};
    {z_2,z_1} --[rrel] {x_1};
  };
  \colclass{x_1,x_2}
  \colclass[52]{y_1,y_2,y_3,y_4}
  \colclass[-52]{z_1,z_2,z_3,z_4}
  \alsolabel[label distance=1.2mm]{x_2/0/X,z_1/above/Z,y_1/above/Y}
  \legend[anchor=north east,at={($(m1.north west)+(18mm,15mm)$)}]{
    \legendrow{rrel}{$R_Z\cup R_Z^*$}
    \legendrow{srel}{$S_Z\cup S_Z^*$}
    \legendrow{trel}{$T\cup T^*$}
  };
  \matrixgraph[matrix anchor=north west,name=m2,label position=below,
    /mnn/label suffix={'},
    at={($(m1.north east)+(8mm,0mm)$)}]
    {&[3mm]&[6mm]&[3mm]&[-1mm]&[8mm]&[-1mm]&[3mm]&[6mm]&[3mm]\\
    &&&&    x_1[ab]&x_2[ab]        \\[8mm]
     z_1 &&&&&&&&& y_1    \\[5mm]
    & z_2 &&&&&&& y_2    \\[8mm]
    && z_3 &&&&& y_3    \\[5mm]
    &&& z_4 &&& y_4    \\
  }{
    {y_4,y_3} --[srel] {z_4,z_3};
    {y_2,y_1} --[srel] {z_2,z_1};
    {y_4,y_3} --[trel] {x_2};
    {y_2,y_1} --[trel] {x_1};
    {z_4,z_3} --[rrel] {x_2};
    {z_2,z_1} --[rrel] {x_1};
  };
  \colclass{x_1,x_2}
  \colclass[52]{y_1,y_2,y_3,y_4}
  \colclass[-52]{z_1,z_2,z_3,z_4}
  \alsolabel[label distance=1.2mm,/mnn/label suffix={'}]{
    x_1/left/X,z_1/above/Z,y_1/above/Y}
  \legend[anchor=north west,at={($(m2.north east)+(-17mm,15mm)$)}]{
    \legendrow{rrel}{$f(R_Z)\cup f(R_Z)^*$}
    \legendrow{srel}{$f(S_Z)\cup f(S_Z)^*$}
    \legendrow{trel}{$T'\cup (T')^*$ ($T'{=}f(T)$)}
  };
\end{tikzpicture}
\caption{Proof of Lemma \ref{lem:excl2points}: Defining the
algebraic isomorphism $f$ locally from $\ccc[X\cup Y\cup Z]$ to $\ccc[X'\cup Y'\cup Z']$.}
\label{fig:proof-excl2points}
\end{figure}
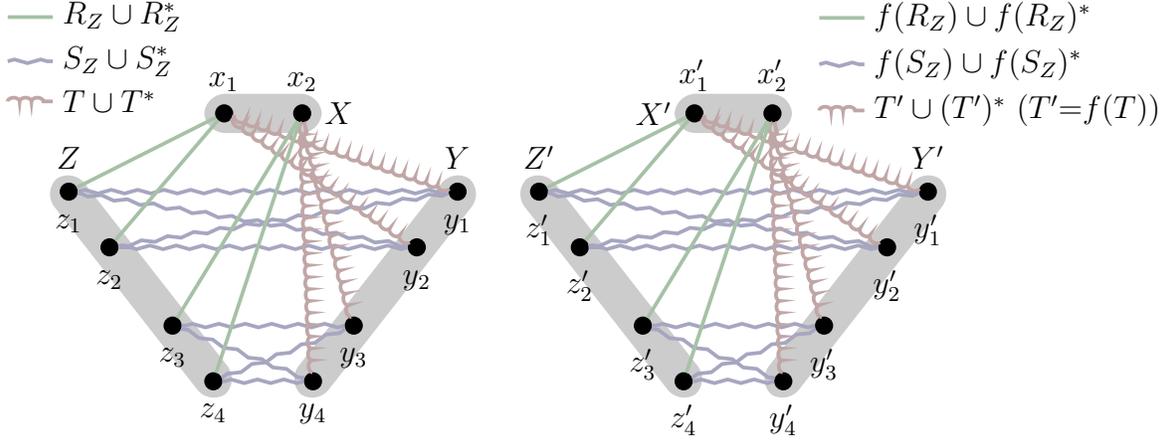

The construction of $\ccc'$ and $f\function{\ccc}{\ccc'}$ is complete. 
It remains to argue that $\ccc'$ is indeed a coherent configuration
and that $f$ is an algebraic isomorphism from \ccc to $\ccc'$.
For this purpose, we use Lemma~\ref{lem:cl:b}.

Let us check that the assumptions of Lemma \ref{lem:cl:b} are fulfilled.
They are trivially true if $A,B,C$ are fibers of $\ccc\setminus X$.
We, therefore, assume that $C=X$. The case that $Y\in\{A,B\}$ is already treated
above. Thus, it remains to consider subconfigurations $\ccc[X\cup A\cup B]$
for each pair $A,B\in F(\ccc)$ such that $A\ne B$ and neither $A$ nor $B$ is in $\{X,Y\}$.
Fix such a pair $A,B$. We will prove a stronger fact: 
There is a bijection  $\phi=\phi_{AB}$ from $X\cup A\cup B$ onto $X'\cup A'\cup B'$
that is a combinatorial isomorphism from $\ccc[X\cup A\cup B]$
to $\ccc'[X'\cup A'\cup B']$ and that induces the restriction of $f$ to $\ccc[X\cup A\cup B]$.

The restriction of $f_0$ to $\ccc[A\cup B]$ is an algebraic isomorphism from
$\ccc[A\cup B]$ to $\ccc'[A'\cup B']$. It is easy to deduce from Lemma \ref{lem:X2Y4}
that this algebraic isomorphism is induced by a combinatorial isomorphism
$\phi_0\function{A\cup B}{A'\cup B'}$. 
We will take $\phi$ to be an extension of $\phi_0$ to a bijection from 
$X\cup A\cup B$ to $X'\cup A'\cup B'$.
Note that $\phi$ can be defined on $X$ in two ways.
In both cases, this will be a combinatorial isomorphism from $\ccc[X\cup A\cup B]$
to $\ccc'[X'\cup A'\cup B']$. We only need to ensure that $\phi$ induces $f$.
To define $\phi$, we consider three cases.

\Case1{Both $\ccc[A,X]$ and $\ccc[B,X]$ are uniform.}
In this case, $\phi(X\times Z)=X'\times Z'=f(X\times Z)$
for both $Z\in\{A,B\}$ independently of how $\phi$ is defined on~$X$.

\Case2{Exactly one of the interspaces $\ccc[A,X]$ and $\ccc[B,X]$, say $\ccc[A,X]$,
is non-uniform.}
Note that $\ccc[X,A]\simeq2K_{1,2}$. Let $a_1,a_2,a_3,a_4$ be the enumeration of $A$
and $a'_1,a'_2,a'_3,a'_4$  be the enumeration of $A'$ fixed in the course of our construction
of $\ccc'$; cf.~\refeq{SZ} and \refeq{fSZ}.
Let
\begin{equation}
  \label{eq:MA}
M_A=\{a_1a_2,a_2a_1,a_3a_4,a_4a_3\}  
\end{equation}
and
\begin{equation}
  \label{eq:MMA}
M'_A=\{a'_1a'_2,a'_2a'_1,a'_3a'_4,a'_4a'_3\}
\end{equation}
be the matchings determined by the interspaces $\ccc[Y,A]$ in the cell $\ccc[A]$ 
and $\ccc'[Y',A']$ in $\ccc'[A']$.
Equalities \refeq{SZ} and \refeq{fSZ} applied to $Z=A$ show that $f_0(M_A)=M'_A$.
It follows that either
$$
\phi_0(\{a_1,a_2\})=\{a'_1,a'_2\}\text{ and }\phi_0(\{a_3,a_4\})=\{a'_3,a'_4\}
$$
or
$$
\phi_0(\{a_1,a_2\})=\{a'_3,a'_4\}\text{ and }\phi_0(\{a_3,a_4\})=\{a'_1,a'_2\}.
$$
In the former case, we extend $\phi_0$ to $\phi$ by $\phi(x_i)=x'_i$ for both $i=1,2$.
In the latter case, we have to swap the values of $\phi$ on $X$
by setting $\phi(x_1)=x'_2$ and $\phi(x_2)=x'_1$. By Equalities \refeq{RZ} and \refeq{fRZ}
applied to $Z=A$, this ensures that $\phi(R_A)=f(R_A)$ and, therefore, 
$f$ on $\ccc[A\cup B\cup X]$ is induced by~$\phi$.

\Case3{Both $\ccc[A,X]$ and $\ccc[B,X]$ are non-uniform.}
Thus, $\ccc[X,A]\simeq2K_{1,2}$ and $\ccc[X,B]\simeq2K_{1,2}$.
Like in the preceding case, let $a_1,a_2,a_3,a_4$ and $a'_1,a'_2,a'_3,a'_4$ be the enumerations of $A$
and $A'$ that were fixed in the course of our construction of $\ccc'$.
Similarly,  let $b_1,b_2,b_3,b_4$ and $b'_1,b'_2,b'_3,b'_4$ be the enumerations of $B$
and $B'$. Since the basis relations $S_A$ and $S_B$ (fixed in the course of our construction
of $\ccc'$) are directly connected at $\ccc[Y]$, Lemma \ref{lem:trans} implies that
the interspace $\ccc[A,B]$ consists of the basis relation
$$
Q_{AB}={\{a_1,a_2\}\times\{b_1,b_2\}}\cup{\{a_3,a_4\}\times\{b_3,b_4\}}
$$
and its complement ${A\times B}\setminus Q_{AB}$. Since the graph
$Q_{AB}\cup Q_{AB}^*\cup S_A\cup S_A^*\cup S_B\cup S_B^*$ has two connected components,
the graph $f_0(Q_{AB})\cup f_0(Q_{AB})^*\cup f_0(S_A)\cup f_0(S_A)^*\cup f_0(S_B)\cup f_0(S_B)^*$
must also have two connected components, which implies that
$$
f_0(Q_{AB})={\{a'_1,a'_2\}\times\{b'_1,b'_2\}}\cup{\{a'_3,a'_4\}\times\{b'_3,b'_4\}}.
$$
Using the fact that the coherent configuration
$\ccc[A\cup B\cup Y]$ has three $2K_{2,2}$-interspaces that are pairwise directly connected,
we see that the restriction of $f_0$ to an algebraic isomorphism
from $\ccc[A\cup B\cup Y]$ to $\ccc'[A'\cup B'\cup Y']$ is induced by
a combinatorial isomorphism $\phi_0$ (now $\phi_0$ is defined on a larger domain than $A\cup B$).
Since $\phi_0(Q_{AB})=f_0(Q_{AB})$, we have either
\begin{equation}
  \label{eq:right0}
\phi_0(\{a_1,a_2\})=\{a'_1,a'_2\}\text{ and }\phi_0(\{b_1,b_2\})=\{b'_1,b'_2\}  
\end{equation}
or
$$
\phi_0(\{a_1,a_2\})=\{a'_3,a'_4\}\text{ and }\phi_0(\{b_1,b_2\})=\{b'_3,b'_4\}.
$$
The coherent configuration $\ccc[A\cup B\cup Y]$ has a combinatorial automorphism
that maps each basis relation onto itself and
swaps the sets $\{a_1,a_2\}$ and $\{a_3,a_4\}$ and simultaneously
the sets $\{b_1,b_2\}$ and $\{b_3,b_4\}$. Applying this automorphism if necessary, we can
modify $\phi_0$ and ensure Equality \refeq{right0}. Now, extending $\phi_0$ to $\phi$ by $\phi(x_i)=x'_i$
for each $i=1,2$, we see that
$$
\phi(R_A)=f(R_A)\text{ and }\phi(R_B)=f(R_B)
$$
for the basis relations $R_A$ and $R_B$ introduced by \refeq{RZ}; see \refeq{fRZ}.
It follows that $\phi$ induces $f$ on $\ccc[A\cup B\cup X]$, as desired.

\smallskip

($\Longleftarrow$)
Let $f$ be an algebraic isomorphism from \ccc to a coherent configuration $\ccc'$.
For each fiber $A\in F(\ccc)$, let $A'=f(A)$ denote the corresponding fiber of $\ccc'$.
Denote the restriction of $f$ to $\ccc\setminus X$ by $f_0$ and note that
$f_0$ is an algebraic isomorphism from $\ccc\setminus X$ to the coherent configuration 
$\ccc'\setminus X'$. Since $\ccc\setminus X$ is separable,
$f_0$ is induced by a combinatorial isomorphism 
$\phi_0\function{V(\ccc)\setminus X}{V(\ccc')\setminus X'}$.
Extend $\phi_0$ to a map $\phi\function{V(\ccc)}{V(\ccc')}$ as follows.
We use the enumeration $x_1,x_2$ of $X$ and $y_1,y_2,y_3,y_4$ of $Y$ fixed in the beginning
of the proof. Recall that the interspace $\ccc[X,Y]$ consists of the relation $T$ specified by \refeq{CXY}.
Denote $y'_i=\phi_0(y_i)$ for each $i\le4$.
Number the points $x'_1$ and $x'_2$ of $X'$ so that
\begin{equation}
  \label{eq:fT}
f(T)=\{x'_1\}\times\{y'_1,y'_2\}\cup\{x'_2\}\times\{y'_3,y'_4\}.  
\end{equation}
Now, we set $\phi(x_1)=x'_1$ and $\phi(x_2)=x'_2$.
The bijection $\phi$ is therewith defined and we have to show
that $\phi$ is a combinatorial isomorphism from \ccc to $\ccc'$
inducing the algebraic isomorphism $f$.
It suffices to do this locally for subconfigurations $\ccc[X\cup Z]$
and $\ccc'[X'\cup Z']$, for each $Z\in F(\ccc\setminus X)$.
Since $|X|=2$, the restriction of $\phi$ to $X\cup Z$ is a combinatorial
isomorphism from $\ccc[X\cup Z]$ to $\ccc'[X'\cup Z']$ even irrespectively
of how $\phi$ is defined on $X$. Thus, we only have to prove that $\phi$
induces the restriction of $f$ to $\ccc[X\cup Z]$.
Our definition of $\phi$ ensures this for $Z=Y$, as it immediately follows 
from \refeq{CXY} and~\refeq{fT}.

Suppose that $Z\ne Y$. If $\ccc[X,Z]$ is uniform, we have nothing to prove.
Assume, therefore, that $\ccc[X,Z]$ is non-uniform. Recall that $\ccc[X,Z]\simeq 2K_{1,2}$ and
$\ccc[Z,Y]\simeq 2K_{2,2}$. We refer to the enumeration $z_1,z_2,z_3,z_4$ we have fixed
for each such $Z$. Denote $z'_i=\phi_0(z_i)$. Consider the basis relation 
$$
S_Z={\{z_1,z_2\}\times\{y_1,y_2\}}\cup{\{z_3,z_4\}\times\{y_3,y_4\}}
$$ 
as in \refeq{SZ}. Since $f_0$ is induced by $\phi_0$, we have
\begin{equation}
  \label{eq:fS}
f(S_Z)={\{z'_1,z'_2\}\times\{y'_1,y'_2\}}\cup{\{z'_3,z'_4\}\times\{y'_3,y'_4\}}.
\end{equation}
Consider now the basis relation 
$$
R_Z=\{x_1\}\times\{z_1,z_2\}\cup\{x_2\}\times\{z_3,z_4\}
$$
in $\ccc[X,Z]$ specified by \refeq{RZ}. Since $f$ provides an algebraic isomorphism from $\ccc[X\cup Y\cup Z]$
to $\ccc'[X'\cup Y'\cup Z']$, Equalities \refeq{fT} and \refeq{fS} readily imply that
$$
f(R_Z)=\{x'_1\}\times\{z'_1,z'_2\}\cup\{x'_2\}\times\{z'_3,z'_4\}.
$$
Thus, $f(R_Z)=\phi(R_Z)$, and $f$ is induced by $\phi$ on $\ccc[X,Z]$ and, hence, everywhere.

The proof of Lemma \ref{lem:excl2points} is complete.

\section{Cutting it down: Interspaces with an 8-cycle}\label{s:excl-C8}

Taking into account Lemma \ref{lem:excl2points}, our task now reduces to
deciding separability of a coherent configuration \ccc under
the following three conditions:
\begin{enumerate}[(1)]
\item 
\ccc is indecomposable;
\item 
all fibers of \ccc have size exactly 4;
\item 
the interspaces of \ccc do not contain any matching basis relation.
\end{enumerate}

Our next step is excluding $C_8$-interspaces.
For $X,Y\in F(\ccc)$, we denote $\ccc\setminus X,Y=\ccc[V(\ccc)\setminus(X\cup Y)]$.

\begin{lemma}\label{lem:excl-C8}
Let \ccc be a coherent configuration satisfying Conditions (1)--(3) above.
Suppose that \ccc has at least three fibers and that there is a $C_8$-interspace  
$\ccc[X,Y]$. Under these conditions, \ccc is separable if and only if
$\ccc\setminus X,Y$ is separable.
\end{lemma}

To prove Lemma \ref{lem:excl-C8}, we need further structural information.

\subsection{Isolation of $C_8$-interspaces}

The following lemma shows that in a matching-free coherent configuration
no two $C_8$-interspaces can share a fiber.

\begin{lemma}\label{lem:isol}
If $\ccc[X,Y]\simeq C_8$ and $\ccc[X,Z]\simeq C_8$, then 
the interspace $\ccc[Y,Z]$ contains a matching basis relation.
\end{lemma}

\begin{proof}
  Fix basis relations $R\in\ccc[Y,X]$ and $S\in\ccc[X,Z]$. Let
$X=\{x_1,x_3,x_5,x_7\}$, $Y=\{y_2,y_4,y_6,y_8\}$, and $Z=\{z_2,z_4,z_6,z_8\}$,
where the points are indexed according to the 8-cycles underlying $R$ and $S$; 
see Figure \ref{fig:proofIsol}.
Let $T\in\ccc[Y,Z]$ be the basis relation containing the arrow $y_2z_2$.
Note that $p^T_{RS}=2$. This equality prevents the membership in $T$
of the other arrows $y_2z_4$, $y_2z_6$, and $y_2z_8$ from $y_2$ to $Z$.
It follows that $T$ has valency 1, that is, it is a matching basis relation.
\end{proof}

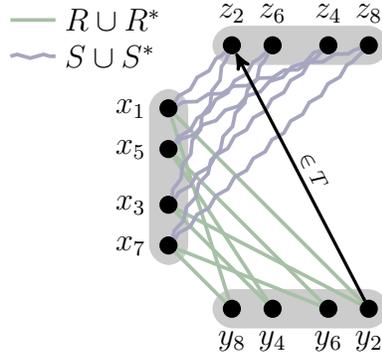
\begin{figure}
\centering
\begin{tikzpicture}[every node/.style={circle,draw,black,
  inner sep=2pt,fill=black},very thick,
  elabel/.style={black,draw=none,fill=none,rectangle},
  every edge/.append style={every node/.append style={elabel}},
  lab/.style={draw=none,fill=none,inner sep=0pt,rectangle},
  ab/.style={label position=above},
  br/.style={label position=below right},
  ri/.style={label position=right},
  le/.style={label position=left},
  cl/.style={label distance=0mm,label position=left},
  cr/.style={label distance=0mm,label position=right},
  nl/.style={nolabel},
  rrel/.style={ggrn,-,line width=1.4pt},
  srel/.style={gblu,-,line width=1.4pt,zz},
  trel/.style={gred,-,line width=1.4pt,bps},
  int/.style={edge node=
      {node [sloped,above]{\scriptsize $\in T$}}},
]
  \matrixgraph[name=m1,label position=below]
    {&[6mm]&[3mm]&[5mm]&[3mm]\\
            & z_2[ab] & z_6[ab] & z_4[ab] & z_8[ab] \\[6mm]
    x_1[le] \\[3mm]
    x_5[le] \\[5mm]
    x_3[le] \\[3mm]
    x_7[le] \\[6mm]
            & y_8     & y_4     & y_6     & y_2     \\
  }{
    {x_1,x_3,x_5,x_7} --[rrel,matching] {y_2,y_4,y_6,y_8};
    {x_1,x_3,x_5,x_7} --[rrel,matching] {y_8,y_2,y_4,y_6};
    {x_1,x_3,x_5,x_7} --[srel,matching] {z_2,z_4,z_6,z_8};
    {x_1,x_3,x_5,x_7} --[srel,matching] {z_8,z_2,z_4,z_6};
    y_2 ->[int] z_2;
  };
  \colclass{x_1,x_3,x_5,x_7}
  \colclass{y_2,y_4,y_6,y_8}
  \colclass{z_2,z_4,z_6,z_8}
  \legend[anchor=north east,at={($(m1.north west)+(2mm,5mm)$)}]{
    \legendrow{rrel}{$R\cup R^*$}
    \legendrow{srel}{$S\cup S^*$}
  };
\end{tikzpicture}
\caption{Proof of Lemma \ref{lem:isol}.}
\label{fig:proofIsol}
\end{figure}

\subsection{Proof of Lemma \ref{lem:excl-C8}}

Since \ccc is indecomposable, $X$ or $Y$ is connected by a non-uniform interspace
to a fiber $U$ of $\ccc\setminus X,Y$. To be specific, without loss of generality
we assume that there is a non-uniform interspace $\ccc[U,X]$.
If possible, we fix $U$ such that $\ccc[U,Y]$ is also non-uniform
and also set $W=U$ in this case; this is Case (a) in Figure \ref{fig:proofExcl-C8}. 
Otherwise, we fix a fiber $W$ of $\ccc\setminus X,Y$
such that $\ccc[W,Y]$ is non-uniform if such a fiber exists. 
Then $W\ne U$; this is Case (b) in Figure \ref{fig:proofExcl-C8}. 
If such a fiber does not exist, we again set $W=U$. 
In the last case, $\ccc[W,Y]$ is uniform.

Since \ccc contains no interspace with a matching,
Lemma \ref{lem:isol} implies that $\ccc[U,X]\simeq2K_{2,2}$.
By the same reason we also have $\ccc[W,Y]\simeq2K_{2,2}$, unless $\ccc[W,Y]$ is uniform.
If $U=W$ and both $\ccc[U,X]$ and $\ccc[U,Y]$ are non-uniform, then
Lemma \ref{lem:trans} implies that these interspaces have skewed connection at $U$;
see Figure~\ref{fig:proofExcl-C8}(a).

\begin{figure}
\centering
\begin{tikzpicture}[every node/.style={circle,draw,black,
  inner sep=2pt,fill=black},very thick,
  elabel/.style={black,draw=none,fill=none,rectangle},
  every edge/.append style={every node/.append style={elabel}},
  lab/.style={draw=none,fill=none,inner sep=0pt,rectangle},
  be/.style={label position=below,label distance=2mm},
  br/.style={label position=below right},
  ri/.style={label position=right},
  le/.style={label position=left},
  nl/.style={nolabel},
  rrel/.style={ggrn,-,line width=1.4pt},
  srel/.style={gblu,-,line width=1.4pt,zz},
  trel/.style={gred,-,line width=1.4pt,bps},
]
  \matrixgraph[name=m1,label position=below]
    {&[3mm]&[6mm]&[3mm]&[-1mm]&[3mm]&[6mm]&[3mm]
    &[-1mm]&[3mm]&[6mm]&[3mm]\\
    &&& x_7[le] && & && y_8[ri]\\[5mm]
    && x_3[le] &&&&  &&& y_4[ri]\\[8mm]
    & x_5[le] &&&& & &&&& y_6[ri]\\[5mm]
     x_1[le] &&&&&&  &&&&& y_2[ri]\\[8mm]
    &&&&    u_1&u_2&u_3&u_4        \\
  }{
    {x_1,x_3,x_5,x_7} --[rrel,matching] {y_2,y_4,y_6,y_8};
    {x_1,x_3,x_5,x_7} --[rrel,matching] {y_8,y_2,y_4,y_6};
    {u_1,u_2} --[srel] {x_1,x_5};
    {u_3,u_4} --[srel] {x_3,x_7};
    {u_1,u_3} --[trel] {y_4,y_8};
    {u_2,u_4} --[trel] {y_2,y_6};
  };
  \colclass{u_1,u_2,u_3,u_4}
  \colclass[-52]{y_2,y_6,y_8,y_4}
  \colclass[52]{x_1,x_5,x_3,x_7}
  \legend[anchor=north west,at={($(m1.north east)+(-4mm,3mm)$)}]{
    \legendrow{rrel}{$T\cup T^*$}
    \legendrow{srel}{$T_X\cup T_X^*$}
    \legendrow{trel}{$T_Y\cup T_Y^*$}
  };
  \matrixgraph[name=m2,matrix anchor=north west,
    label position=below,
    at={($(m1.north east)+(20mm,0mm)$)}]
    {&[3mm]&[6mm]&[3mm]&[12mm]&[3mm]&[6mm]&[3mm]\\
    &&&          x_7[le] & y_8[ri]\\[5mm]
    &&       x_3[le] &   &   & y_4[ri]\\[8mm]
    &    x_5[le] &   &   &   &   & y_6[ri]\\[5mm]
     x_1[le] &   &   &   &   &   &   & y_2[ri]\\[8mm]
     u_1&u_2&u_3&u_4&      w_1&w_3&w_2&w_4        \\
  }{
    {x_1,x_3,x_5,x_7} --[rrel,matching] {y_2,y_4,y_6,y_8};
    {x_1,x_3,x_5,x_7} --[rrel,matching] {y_8,y_2,y_4,y_6};
    {u_1,u_2} --[srel] {x_1,x_5};
    {u_3,u_4} --[srel] {x_3,x_7};
    {w_1,w_3} --[trel] {y_4,y_8};
    {w_2,w_4} --[trel] {y_2,y_6};
  };
  \colclass{u_1,u_2,u_3,u_4}
  \colclass{w_1,w_2,w_3,w_4}
  \colclass[-52]{y_2,y_6,y_8,y_4}
  \colclass[52]{x_1,x_5,x_3,x_7}
  \node[name="alab",lab,at={(m1.south west)}] {(a)};
  \node[lab,at={($(m2.south west)+(-8mm,0mm)$)}] {(b)};
\end{tikzpicture}
\caption{Proof of Lemma \ref{lem:excl-C8}. Case (a): $U=W$. Case (b): $U\ne W$.
In both cases, $\ccc[Y,W]$ is non-uniform.}
\label{fig:proofExcl-C8}
\end{figure}
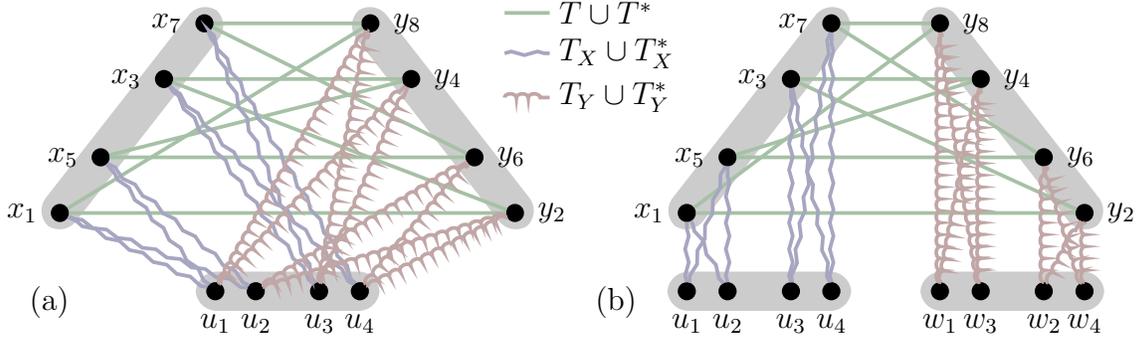

Fix basis relations $T\in\ccc[X,Y]$, $T_X\in\ccc[U,X]$, and $T_Y\in\ccc[W,Y]$.
Enumerate the points in the fibers $X=\{x_1,x_3,x_5,x_7\}$ and
$Y=\{y_2,y_4,y_6,y_8\}$ so that the indices correspond to the 8-cycle
underlying $T$, as in Figure \ref{fig:proofExcl-C8}. 
By Part 3 of Lemma \ref{lem:X2Y4}, the cell $\ccc[X]$
contains a unique matching, namely 
\begin{equation}
  \label{eq:NX}
N_X=\{x_1x_5,x_5x_1,x_3x_7,x_7x_3\}  
\end{equation}
(corresponding to the two pairs of antipodal odd points on the 8-cycle).
Therefore, Part 2 of the same lemma implies that $\ccc[U,X]$ determines 
exactly this matching in $\ccc[X]$. We enumerate the points of $U=\{u_1,u_2,u_3,u_4\}$ 
so that 
$$
T_X={\{u_1,u_2\}\times\{x_1,x_5\}}\cup{\{u_3,u_4\}\times\{x_3,x_7\}}.
$$
Note that $\ccc[X,U]$ determines the matching
\begin{equation}
  \label{eq:MX}
M_X=\{u_1u_2,u_2u_1,u_3u_4,u_4u_3\}  
\end{equation}
in the cell~$\ccc[U]$.

If $U=W$ and both $\ccc[U,X]$ and $\ccc[U,Y]$ are non-uniform, then
we suppose that
$$
T_Y={\{u_1,u_3\}\times\{y_4,y_8\}}\cup{\{u_2,u_4\}\times\{y_2,y_6\}},
$$
as in Figure \ref{fig:proofExcl-C8}(a). Thus, $T_Y$ determines the matching
$$
M_Y=\{u_1u_3,u_3u_1,u_2u_4,u_4u_2\}
$$
in $\ccc[W]=\ccc[U]$.
This assumption causes no loss of generality because the coherent configuration 
$\ccc[X\cup Y\cup U]$ under the conditions $\ccc[X,Y]\simeq C_8$, $\ccc[X,U]\simeq2K_{2,2}$,
and $\ccc[Y,U]\simeq2K_{2,2}$ is unique up to combinatorial isomorphism.
Indeed, the points of $\ccc[X\cup Y\cup U]$ can obviously be enumerated so that
the fragments $\ccc[X\cup Y]$ and $\ccc[X\cup U]$ will look
exactly as in Figure \ref{fig:proofExcl-C8}(a). 
Once this is fixed, the interspace $\ccc[U,Y]$ must determine
the matching $\{y_2y_6,y_6y_2,y_4y_8,y_8y_4\}$ in $\ccc[Y]$ corresponding to
the two pairs of antipodal even points on the 8-cycle.
By Lemma \ref{lem:trans}, the interspace $\ccc[Y,U]$ must determine
a matching in $\ccc[U]$ different from $M_X$.
One of these matchings, namely $M_Y$, appears in Figure \ref{fig:proofExcl-C8}(a),
and the other of them results actually in the same picture by transposing
the points $u_1$ and~$u_2$.

If $U\ne W$, then we enumerate $W=\{w_1,w_2,w_3,w_4\}$ so that
$$
T_Y={\{w_1,w_3\}\times\{y_4,y_8\}}\cup{\{w_2,w_4\}\times\{y_2,y_6\}}.
$$
In this case,
$$
M_Y=\{w_1w_3,w_3w_1,w_2w_4,w_4w_2\},
$$
where $M_Y$, as above, denotes the matching determined by $\ccc[Y,W]$ in $\ccc[W]$.
For notational convenience, in the case that $U=W$ we set $w_i=u_i$ for each $i\le4$.
Note that $M_Y$ is well defined irrespectively of whether $W=U$ or $W\ne U$.

\medskip

($\implies$)
Let $f_0$ be an algebraic isomorphism from $\ccc\setminus X,Y$ to
a coherent configuration \ccd. Like in the proof of Lemma \ref{lem:excl2points}, 
it suffices to extend \ccd to a coherent configuration $\ccc'$ and 
to extend $f_0$ to an algebraic isomorphism $f$ from \ccc to~$\ccc'$. 

For a fiber $A$ of $\ccc\setminus X,Y$, let $A'=f_0(A)$ denote the fiber of \ccd 
corresponding to $A$ under $f_0$. We fix an enumeration $U'=\{u'_1,u'_2,u'_3,u'_4\}$
so that
\begin{equation}
  \label{eq:f0MX}
f_0(M_X)=\{u'_1u'_3,u'_3u'_1,u'_2u'_4,u'_4u'_2\}. 
\end{equation}
If $W=U$, then $W'=U'$, and we set $w'_i=u'_i$ for $i\le4$. If $W\ne U$, then
we fix an enumeration $W'=\{w'_1,w'_2,w'_3,w'_4\}$ so that
$$
f_0(M_Y)=\{w'_1w'_3,w'_3w'_1,w'_2w'_4,w'_4w'_2\}.
$$

We now construct $\ccc'$ and $f$ as follows.
Let $\phi_{UW}$ be the bijection from $U\cup W$ onto $U'\cup W'$ defined
by $\phi_{UW}(u_i)=u'_i$ and $\phi_{UW}(w_i)=w'_i$. Moreover, we take
$X'=\{x'_1,x'_3,x'_5,x'_7\}$ and $Y=\{y'_2,y'_4,y'_6,y'_8\}$ such that
$X'\cap Y'=\emptyset$ and $(X'\cup Y')\cap V(\ccd)=\emptyset$
and extend $\phi_{UW}$ to a bijection from $U\cup W\cup X\cup Y$ onto
$U'\cup W'\cup X'\cup Y'$ by setting $\phi_{UW}(x_i)=x'_i$ and $\phi_{UW}(y_i)=y'_i$.
The sets $X'$ and $Y'$ will be fibers of $\ccc'$. 
We build the fragments $\ccc'[U'\cup X'\cup Y']$ and $\ccc'[W'\cup X'\cup Y']$
as isomorphic copies of $\ccc[U\cup X\cup Y]$ and $\ccc[W\cup X\cup Y]$
under the map $\phi_{UW}$. Moreover, for any relation 
$R\in\ccc[U\cup X\cup Y]\cup\ccc[W\cup X\cup Y]$ we set $f(R)=\phi_{UW}(R)$.

It remains, for each $Z\in F(\ccc)$ such that $Z\notin\{X,Y,U,W\}$
to construct $\ccc'[X',Z']$ and $\ccc'[Y',Z']$ and to define $f$
locally as a bijection from $\ccc[X,Z]$ to $\ccc'[X',Z']$
and  $\ccc[Y,Z]$ to $\ccc'[Y',Z']$.
If $\ccc[X,Z]$ is uniform, we set $\ccc'[X',Z']$ also to be uniform,
and correspondingly define $f(X\times Z)=X'\times Z'$. 
Similarly, if $\ccc[Y,Z]$ is uniform, we set $\ccc'[Y',Z']$ to be uniform
and define $f(Y\times Z)=Y'\times Z'$.

Assume that $\ccc[X,Z]$ is non-uniform. By Lemma \ref{lem:isol},
$\ccc[X,Z]\simeq2K_{2,2}$. Recall that the cell $\ccc[X]$ contains a
unique matching basis relation, namely $N_X$. By Part 2 of Lemma \ref{lem:X2Y4},
$\ccc[Z,X]$ determines $N_X$ in $\ccc[X]$ and, hence, is directly connected
to $\ccc[U,X]$ at $X$. Lemma \ref{lem:trans} implies that $\ccc[Z,U]\simeq 2K_{2,2}$
and that the interspace $\ccc[Z,U]$ has direct connections to $\ccc[U,X]$ at $U$
and to $\ccc[X,Z]$ at $Z$. In particular, $\ccc[Z,U]$ determines 
the same matching in $\ccc[U]$ as $\ccc[X,U]$, namely $M_X$. Note that $\ccc[Z,Y]$ must be uniform.
If $W\ne U$, this follows by the choice of $W$ and, if $W=U$, the non-uniformity
of $\ccc[Z,Y]$, by the argument similar to the above, would imply that 
$\ccc[Z,U]$ in $\ccc[U]$ determines the matching $M_Y$ rather than~$M_X$.

\begin{figure}
\centering
\begin{tikzpicture}[every node/.style={circle,draw,black,
  inner sep=2pt,fill=black},very thick,
  elabel/.style={black,draw=none,fill=none,rectangle},
  every edge/.append style={every node/.append style={elabel}},
  lab/.style={draw=none,fill=none,inner sep=0pt,rectangle},
  be/.style={label position=below,label distance=2mm},
  br/.style={label position=below right},
  ri/.style={label position=right},
  le/.style={label position=left},
  nl/.style={nolabel},
  rrel/.style={ggrn,-,line width=1.4pt},
  srel/.style={gblu,-,line width=1.4pt,zz},
  trel/.style={gred,-,line width=1.4pt,bps},
  tyrel/.style={gorg,-,line width=1.4pt,sw},
]
  \matrixgraph[name=m1,label position=above]
    {&[3mm]&[6mm]&[3mm]&[-1mm]&[3mm]&[6mm]&[3mm]
    &[-1mm]&[3mm]&[6mm]&[3mm]\\
    &&&&    x_1&x_5&x_3&x_7        \\[8mm]
     z_4[le] &&&&&&  &&&&& u_4[ri]\\[5mm]
    & z_3[le] &&&& & &&&& u_3[ri]\\[8mm]
    && z_2[le] &&&&  &&& u_2[ri]\\[5mm]
    &&& z_1[le] && & && u_1[ri]\\
  }{
    {z_1,z_2} --[srel] {u_1,u_2};
    {z_3,z_4} --[srel] {u_3,u_4};
    {z_1,z_2} --[rrel] {x_1,x_5};
    {z_3,z_4} --[rrel] {x_3,x_7};
    {u_1,u_2} --[trel] {x_1,x_5};
    {u_3,u_4} --[trel] {x_3,x_7};
  };
  \colclass{x_1,x_5,x_3,x_7}
  \colclass[-52]{z_1,z_2,z_3,z_4}
  \colclass[52]{u_1,u_2,u_3,u_4}
  \legend[anchor=north west,at={($(m1.north east)+(-4mm,15mm)$)}]{
    \legendrow{rrel}{$R_Z\cup R_Z^*$}
    \legendrow{srel}{$S_Z\cup S_Z^*$}
    \legendrow{trel}{$T_X\cup T_X^*$}
    \legendrow{tyrel}{$T_Y\cup T_Y^*$}
  };
  \matrixgraph[name=m2,matrix anchor=north west,
    label position=above,
    at={($(m1.north east)+(20mm,0mm)$)}]
    {&[3mm]&[6mm]&[3mm]&[-1mm]&[3mm]&[6mm]&[3mm]
    &[-1mm]&[3mm]&[6mm]&[3mm]\\
    &&&&    y_8&y_4&y_6&y_2        \\[8mm]
     z_4[le] &&&&&&  &&&&& w_4[ri]\\[5mm]
    & z_3[le] &&&& & &&&& w_2[ri]\\[8mm]
    && z_2[le] &&&&  &&& w_3[ri]\\[5mm]
    &&& z_1[le] && & && w_1[ri]\\
  }{
    {z_1,z_2} --[srel] {w_1,w_3};
    {z_3,z_4} --[srel] {w_2,w_4};
    {z_1,z_2} --[rrel] {y_8,y_4};
    {z_3,z_4} --[rrel] {y_2,y_6};
    {w_1,w_3} --[tyrel] {y_8,y_4};
    {w_2,w_4} --[tyrel] {y_2,y_6};
  };
  \colclass[-52]{z_1,z_2,z_3,z_4}
  \colclass[52]{w_1,w_2,w_3,w_4}
  \colclass{y_2,y_6,y_8,y_4}
  \node[name="alab",lab,at={(m1.south west)}] {(a)};
  \node[lab,at={($(m2.south west)+(-2mm,0mm)$)}] {(b)};
\end{tikzpicture}
\caption{Proof of Lemma \ref{lem:excl-C8}.
(a) $\ccc[X,Z]$ is non-uniform, and $\ccc[Y,Z]$ is uniform: relations are named as in the proof;
(b) $\ccc[X,Z]$ is uniform, and $\ccc[Y,Z]$ is non-uniform: now $S_Z\in\ccc[Z,W]$ and $R_Z\in\ccc[Y,Z]$.}
\label{fig:proofExclC8-2}
\end{figure}
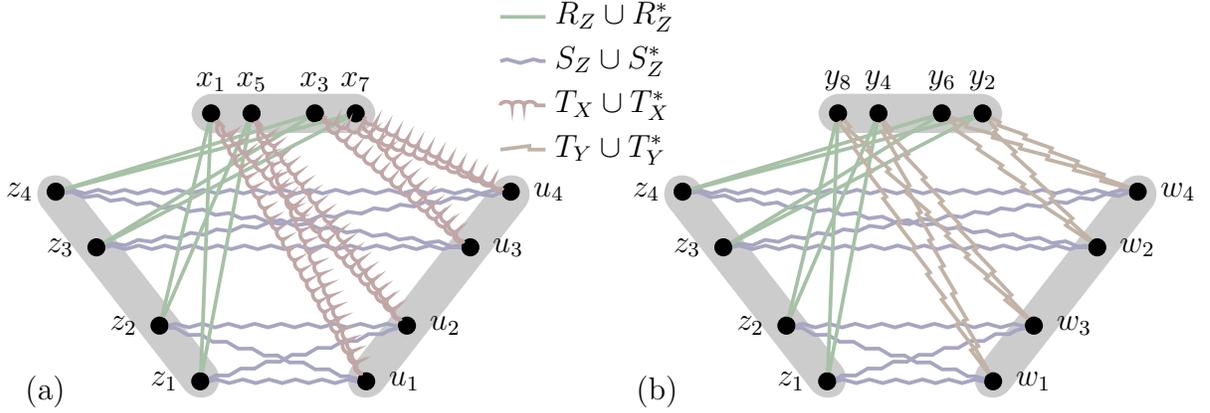

The rest of our argument is similar to the proof of Lemma \ref{lem:excl2points}.
Fix a basis relation $S_Z\in\ccc[Z,U]$. Fix an enumeration $Z=\{z_1,z_2,z_3,z_4\}$ such that
\begin{equation}
  \label{eq:SZc}
S_Z={\{z_1,z_2\}\times\{u_1,u_2\}}\cup{\{z_3,z_4\}\times\{u_3,u_4\}};
\end{equation}
see Figure \ref{fig:proofExclC8-2}(a).
Since $f_0$ is an algebraic isomorphism from $\ccc\setminus X,Y$ to $\ccd$
and $S_Z$ determines the matching $M_X$ in $\ccc[U]$,
the basis relation $f_0(S_Z)$ determines the matching $f_0(M_X)$ in $\ccd[U']$.
Taking into account \refeq{f0MX}, the points of $Z'$ can be enumerated so that
\begin{equation}
  \label{eq:fSZc}
f_0(S_Z)={\{z'_1,z'_2\}\times\{u'_1,u'_2\}}\cup{\{z'_3,z'_4\}\times\{u'_3,u'_4\}},
\end{equation} 
and we fix such an enumeration $z'_1,z'_2,z'_3,z'_4$.
Note that $\ccc[X,Z]$ consists of the basis relation
\begin{equation}
  \label{eq:RZc}
R_Z={\{x_1,x_5\}\times\{z_1,z_2\}}\cup{\{x_3,x_7\}\times\{z_3,z_4\}}  
\end{equation}
and its complement ${X\times Z}\setminus R_Z$.
We, therefore, define the interspace $\ccc'[X',Z']$ as
consisting of the relation
\begin{equation}
  \label{eq:RRZc}
R'_Z={\{x'_1,x'_5\}\times\{z'_1,z'_2\}}\cup{\{x'_3,x'_7\}\times\{z'_3,z'_4\}}
\end{equation}
and its complement ${X'\times Z'}\setminus R'_Z$.
This ensures that $\ccc'[X'\cup U'\cup Z']$ is a coherent configuration
combinatorially isomorphic to $\ccc[X\cup U\cup Z]$.
Moreover, we define $f$ on $\ccc[X,Z]$ by setting
\begin{equation}
  \label{eq:fRZc}
f(R_Z)=R'_Z.
\end{equation}
This ensures that $f$ is an algebraic isomorphism from $\ccc[X\cup U\cup Z]$
to $\ccc'[X'\cup U'\cup Z']$.

Assume now that $\ccc[Y,Z]$ is non-uniform.
As was noticed, in this case the interspace $\ccc[X,Z]$ must be uniform
and, therefore, the fiber $Z$ was not handled above.
We construct $\ccc'[Y',Z']$ and extend $f$ to a map from $\ccc[Y,Z]$ 
to $\ccc'[Y',Z']$ similarly to the above, considering the interspaces
$\ccc[Y,W]\simeq2K_{2,2}$ and $\ccc[Z,W]\simeq2K_{2,2}$;
see Figure \ref{fig:proofExclC8-2}(b).

The construction of $\ccc'$ and $f\function{\ccc}{\ccc'}$ is therewith complete. 
It remains to argue that $\ccc'$ is indeed a coherent configuration
and that $f$ is indeed an algebraic isomorphism from \ccc to $\ccc'$.
We use Lemma~\ref{lem:cl:b}.

If $A,B,C$ are fibers of $\ccc\setminus X,Y$, then
the assumptions of Lemma \ref{lem:cl:b} are trivially true.

Let $B=X$ and $C=Y$. If $A=U$ or $A=W$, then the assumptions of Lemma \ref{lem:cl:b}
are ensured by the construction. Suppose that $A=Z$ is a fiber of $\ccc\setminus X,Y$
different from $U$ and $W$. Recall that at least one of the interspaces
$\ccc[Z,X]$ and $\ccc[Z,Y]$ is uniform. To be specific, assume that $\ccc[Z,Y]$ is uniform;
the other case is symmetric. 
Using the fact that any association scheme on 4 points is separable,
we consider a bijection $\phi_Z\function Z{Z'}$ that is a combinatorial
isomorphism from $\ccc[Z]$ to $\ccd[Z']$ inducing the restriction of $f_0$
to $\ccc[Z]$. We stick to
the enumeration $z_1,z_2,z_3,z_4$ and $z'_1,z'_2,z'_3,z'_4$ of points in $Z$
and $Z'$ fixed while constructing $\ccc'[Z',X']$. 
If $\ccc[Z,X]$ is non-uniform, then either
\begin{equation}
  \label{eq:phiz}
\phi_Z(\{z_1,z_2\})=\{z'_1,z'_2\}
\end{equation}
or $\phi_Z(\{z_1,z_2\})=\{z'_3,z'_4\}$. 
The association scheme $\ccc[Z]$ has a combinatorial automorphism that
maps each basis relation onto itself and swaps the sets $\{z_1,z_2\}$ and $\{z_3,z_4\}$.
Using this automorphism if needed, we can ensure Equality \refeq{phiz}.
Extend $\phi_Z$ to a bijection $\phi_Z\function{X\cup Y\cup Z}{X'\cup Y'\cup Z'}$ by
setting $\phi(x_i)=x'_i$ and $\phi(y_i)=y'_i$. 
For the basis relations $R_Z$ and $R'_Z$ introduced by \refeq{RZc} and \refeq{RRZc},
from the definition of $\phi_Z$ on $X\cup Y$ and Equality \refeq{phiz} we derive 
that $\phi_Z(R_Z)=R'_Z$. Along with \refeq{fRZc}, this shows
that $\phi_Z$ is a combinatorial isomorphism from $\ccc[X\cup Y\cup Z]$
to $\ccc'[X'\cup Y'\cup Z']$ and that $\phi_Z$ induces the restriction of $f$
to $\ccc[X\cup Y\cup Z]$. Thus, the assumptions of Lemma \ref{lem:cl:b} are
fulfilled also in this case.

Assume now that $C=X$ and $A$ and $B$ are fibers of $\ccc\setminus X,Y$. 
If $U\in\{A,B\}$, then the assumptions of Lemma \ref{lem:cl:b} are ensured 
by the construction. Suppose, therefore, that neither $A$ nor $B$ is
equal to $U$. Our goal is to construct a bijection  $\phi_{AB}$ from 
$X\cup A\cup B$ onto $X'\cup A'\cup B'$
that is a combinatorial isomorphism from $\ccc[X\cup A\cup B]$
to $\ccc'[X'\cup A'\cup B']$ and that induces the restriction of $f$ to $\ccc[X\cup A\cup B]$.
Like in the proof of Lemma \ref{lem:excl2points}, we split our analysis into three cases.

\Case1{Both $\ccc[A,X]$ and $\ccc[B,X]$ are uniform.}
The interspace $\ccc[A,B]$ can be uniform or of $2K_{2,2}$- or $C_8$-type.
The structure of the subconfiguration $\ccc[A\cup B]$ in the last two cases
is described by Parts 2 and 3 of Lemma \ref{lem:X2Y4}.
In each case, it is easy to see that the restriction of $f_0$ to
an algebraic isomorphism from $\ccc[A\cup B]$ to $\ccc'[A'\cup B']$
is induced by a combinatorial isomorphism $\phi_0\function{A\cup B}{A'\cup B'}$. 
We extend $\phi_0$ to $\phi_{AB}$ by setting $\phi_{AB}(x_i)=x'_i$.

\Case2{Exactly one of the interspaces $\ccc[A,X]$ and $\ccc[B,X]$, say $\ccc[A,X]$,
is non-uniform.}
Like in the first case, $\ccc[A,B]$ can be uniform or of $2K_{2,2}$- or $C_8$-type.
Again, let $\phi_0\function{A\cup B}{A'\cup B'}$ be a combinatorial isomorphism
$\ccc[A\cup B]$ to $\ccc'[A'\cup B']$ inducing the restriction of $f_0$ to $\ccc[A\cup B]$.
By Lemma \ref{lem:isol}, $\ccc[A,X]\simeq2K_{2,2}$, and we consider the enumeration $a_1,a_2,a_3,a_4$ of $A$
and the enumeration $a'_1,a'_2,a'_3,a'_4$  of $A'$ that we have fixed for each such $A$; cf.~\refeq{SZc} and \refeq{fSZc}.
Let
$$
M_A=\{a_1a_2,a_2a_1,a_3a_4,a_4a_3\}
$$
and
$$
M'_A=\{a'_1a'_2,a'_2a'_1,a'_3a'_4,a'_4a'_3\}
$$
be the matchings determined by the interspaces $\ccc[U,A]$ in the cell $\ccc[A]$ 
and $\ccc'[U',A']$ in $\ccc'[A']$.
By \refeq{SZc} and \refeq{fSZc} applied to $Z=A$, we have $f_0(M_A)=M'_A$.
It follows that either
$$
\phi_0(\{a_1,a_2\})=\{a'_1,a'_2\}\text{ and }\phi_0(\{a_3,a_4\})=\{a'_3,a'_4\}
$$
or
$$
\phi_0(\{a_1,a_2\})=\{a'_3,a'_4\}\text{ and }\phi_0(\{a_3,a_4\})=\{a'_1,a'_2\}.
$$
In the former case, we extend $\phi_0$ to $\phi_{AB}$ by $\phi_{AB}(x_i)=x'_i$.
In the latter case, however, we have to swap the values of $\phi_{AB}$ on $X$
so that $\phi_{AB}(\{x_1,x_5\})=\{x'_3,x'_7\}$ and $\phi_{AB}(\{x_3,x_7\})=\{x'_1,x'_5\}$.
This ensures that $\phi_{AB}(R_A)=R'_A$
for the basis relations $R_A$ and $R'_A$ as in \refeq{RZc} and \refeq{RRZc}.
It follows from \refeq{fRZc} that the restriction of $f$ to
$\ccc[A\cup B\cup X]$ is induced by $\phi_{AB}$.

\Case3{Both $\ccc[A,X]$ and $\ccc[B,X]$ are non-uniform.}
Recall that, by Lemma \ref{lem:isol}, both $\ccc[A,X]\simeq2K_{2,2}$ and $\ccc[B,X]\simeq2K_{2,2}$.
Moreover, both $\ccc[A,X]$ and $\ccc[B,X]$ are directly connected to $\ccc[U,X]\simeq2K_{2,2}$ 
at $X$. It follows by Lemma \ref{lem:trans} that both $\ccc[A,U]\simeq2K_{2,2}$ and $\ccc[B,U]\simeq2K_{2,2}$
are directly connected to $\ccc[X,U]$ at $U$ and, therefore, also to each other.
Applying Lemma \ref{lem:trans} once again, we conclude that $\ccc[A,B]\simeq2K_{2,2}$
and the connections between $\ccc[A,B]$ with $\ccc[U,A]$ at $A$ and $\ccc[U,B]$ at $B$
are direct. Using the enumeration of the fibers $A$, $B$, $A'$, and $B'$
fixed in the course of our construction of $\ccc'$, we see that
the interspace $\ccc[A,B]$ consists of the basis relation
$$
Q_{AB}={\{a_1,a_2\}\times\{b_1,b_2\}}\cup{\{a_3,a_4\}\times\{b_3,b_4\}}
$$
and its complement ${A\times B}\setminus Q_{AB}$.
Taking into account Equalities \refeq{SZc} and \refeq{fSZc} for $Z=A$ and $Z=B$, we
conclude that 
$$
f_0(Q_{AB})={\{a'_1,a'_2\}\times\{b'_1,b'_2\}}\cup{\{a'_3,a'_4\}\times\{b'_3,b'_4\}}.
$$
We know the structure of the coherent configuration
$\ccc[A\cup B\cup U]$ up to the types of its cells $\ccc[A]$, $\ccc[B]$, and $\ccc[U]$.
In each case, the restriction of $f_0$ to an algebraic isomorphism
from $\ccc[A\cup B\cup U]$ to $\ccc'[A'\cup B'\cup U']$ is induced by
a combinatorial isomorphism $\phi_0\function{A\cup B\cup U}{A'\cup B'\cup U'}$.
Since $\phi_0(Q_{AB})=f_0(Q_{AB})$, we have either
\begin{equation}
  \label{eq:right}
\phi_0(\{a_1,a_2\})=\{a'_1,a'_2\}\text{ and }\phi_0(\{b_1,b_2\})=\{b'_1,b'_2\}  
\end{equation}
or
$$
\phi_0(\{a_1,a_2\})=\{a'_3,a'_4\}\text{ and }\phi_0(\{b_1,b_2\})=\{b'_3,b'_4\}.
$$
Applying, if necessary, an appropriate combinatorial automorphism of $\ccc[A\cup B\cup U]$, we can
ensure Equality \refeq{right}. We now extend $\phi_0$ to $\phi_{AB}$ by $\phi_{AB}(x_i)=x'_i$.
Note that
$$
\phi_{AB}(R_A)=R'_A\text{ and }\phi_{AB}(R_B)=R'_B
$$
for the basis relations introduced by \refeq{RZc} and \refeq{RRZc}.
Based on \refeq{fRZc}, we conclude that the restriction of $f$ to
$\ccc[A\cup B\cup X]$ is induced by $\phi_{AB}$.

The analysis of fiber triples $A,B,C$ such that $C=X$ and $A,B\in F(\ccc\setminus X,Y)$
is complete. The triple of fibers consisting of $C=Y$ and $A,B\in F(\ccc\setminus X,Y)$
are treated similarly.

\medskip

($\Longleftarrow$)
Let $f$ be an algebraic isomorphism from \ccc to a coherent configuration $\ccc'$.
For each fiber $A\in F(\ccc)$, let $A'=f(A)$ denote the corresponding fiber of $\ccc'$.
Like in the proof of Lemma \ref{lem:excl2points},
denote the restriction of $f$ to $\ccc\setminus X,Y$ by $f_0$ and note that
$f_0$ is an algebraic isomorphism from $\ccc\setminus X,Y$ to the coherent configuration 
$\ccc'\setminus X',Y'$. Since $\ccc\setminus X,Y$ is separable,
$f_0$ is induced by a combinatorial isomorphism 
$\phi_0\function{V(\ccc)\setminus(X\cup Y)}{V(\ccc')\setminus(X'\cup Y')}$.
We have to extend $\phi_0$ to a combinatorial isomorphism $\phi\function{V(\ccc)}{V(\ccc')}$
from \ccc to $\ccc'$ that induces~$f$.

We first solve a more modest task of defining $\phi$ on $X\cup Y$ so that
$\phi$ will be a combinatorial isomorphism from $\ccc[X\cup Y\cup U\cup W]$
to $\ccc'[X'\cup Y'\cup U'\cup W']$ inducing the restriction of $f$ to
an algebraic isomorphism between these subconfigurations. Let
$$
u'_i=\phi_0(u_i)\text{ and }w'_i=\phi_0(w_i)\text{ for }i\le4.
$$
Since $f$ is an algebraic isomorphism, $f(M_X)$ is a matching basis relation in $\ccc'[U']$;
see \refeq{MX} and Figure \ref{fig:proofExcl-C8}.
Since $\phi_0$ induces the restriction of $f$ to $\ccc[U]$,
$$
f(M_X)=\{u'_1u'_2,u'_2u'_1,u'_3u'_4,u'_4u'_3\}.
$$
Since $f$ is an algebraic isomorphism, $f(T_X)$ determines $f(M_X)$ in $\ccc'[U']$.
If $\ccc[W,Y]$ is non-uniform, then we similarly have
$$
f(M_Y)=\{w'_1w'_3,w'_3w'_1,w'_2w'_4,w'_4w'_2\},
$$
and $f(T_Y)$ determines $f(M_Y)$ in $\ccc'[W']$, irrespectively of whether $W=U$ or $W\ne U$.
We enumerate $X'=\{x'_1,x'_3,x'_5,x'_7\}$ and $Y'=\{y'_2,y'_4,y'_6,y'_8\}$ so that
\begin{equation}
  \label{eq:fTX}
f(T_X)={\{u'_1,u'_2\}\times\{x'_1,x'_5\}}\cup{\{u'_3,u'_4\}\times\{x'_3,x'_7\}}
\end{equation}
and
$$
f(T_Y)={\{w'_1,w'_3\}\times\{y'_4,y'_8\}}\cup{\{w'_2,w'_4\}\times\{y'_2,y'_6\}},
$$
the last equality under the assumption that $\ccc[W,Y]$ is non-uniform.
Our first concern is to satisfy the constraints
\begin{equation}
  \label{eq:phixxyy}
\phi(\{x_1,x_5\})=\{x'_1,x'_5\}\text{ and }\phi(\{y_2,y_6\})=\{y'_2,y'_6\}.
\end{equation}
This will ensure that
$$
\phi(T_X)=f(T_X)\text{ and }\phi(T_Y)=f(T_Y),
$$
accomplishing our task locally on $\ccc[X\cup U]$ and $\ccc[Y\cup W]$,
the latter also if $\ccc[W,Y]$ is uniform.
We fulfill the first equality in \refeq{phixxyy} immediately just by setting
$$
\phi(x_1)=x'_1\text{ and }\phi(x_5)=x'_5.
$$
It remains to ensure the second equality in \refeq{phixxyy} as well as the equality
\begin{equation}
  \label{eq:phiTfT}
\phi(T)=f(T)  
\end{equation}
for the 8-cycles $T\in\ccc[X,Y]$ and $f(T)\in\ccc'[X',Y']$; see Figure~\ref{fig:proofExcl-C8}.

Since $f$ is an algebraic isomorphism, Equality \refeq{fTX} implies that
$$
f(N_X)=\{x'_1x'_5,x'_5x'_1,x'_3x'_7,x'_7x'_3\}
$$
for the unique matching basis relation $N_X\in\ccc[X]$ introduced by \refeq{NX}.
Therefore $x'_1,x'_5$ and $x'_3,x'_7$ are the two pairs of diametrically opposite points
on the 8-cycle $f(T)$ that belong to $X'$. Similarly, if $\ccc[W,Y]$ is non-uniform, then $y'_2,y'_6$ and $y'_4,y'_8$ 
are the two pairs of antipodal points on $f(T)$ belonging to $Y'$. In fact, we can
suppose this also if $\ccc[W,Y]$ is uniform, as $Y'$ can be enumerated arbitrarily in this case.
Let $y'$ be the common neighbor of $x'_1$ and $x'_3$ on $f(T)$. We set
$$
\phi(x_3)=x'_3\text{ and }\phi(x_7)=x'_7\text{ if }y'\in\{y'_2,y'_6\}
$$
or
$$
\phi(x_3)=x'_7\text{ and }\phi(x_7)=x'_3\text{ if }y'\in\{y'_4,y'_8\}.
$$
Assignment of the four values $\phi(x_i)$ uniquely determines the four values $\phi(y_i)$. For example,
$\phi(y_1)$ is the point in $Y'$ lying on $f(T)$ between $\phi(x_1)$ and $\phi(x_3)$.
In each case, Conditions \refeq{phixxyy} and \refeq{phiTfT} are fulfilled.

Our modest task is fulfilled.
Now, let $Z\ne U,W$ be another fiber of $\ccc\setminus X,Y$.
We have to verify that $\phi$ is a combinatorial isomorphism from
$\ccc[X\cup Z]$ to $\ccc'[X'\cup Z']$ and from $\ccc[Y\cup Z]$ to $\ccc'[Y'\cup Z']$
and that $\phi$ induces $f$ on these subconfigurations. We do it for $\ccc[X\cup Z]$;
the argument for $\ccc[Y\cup Z]$ is similar.

If $\ccc[X,Z]$ is uniform, we have nothing to do.
Assume, therefore, that $\ccc[X,Z]$ is non-uniform and, hence,
$\ccc[X,Z]\simeq 2K_{2,2}$.  Recall that the connection between $\ccc[Z,X]$ and $\ccc[U,X]$
at $X$ must be direct; see Figure \ref{fig:proofExclC8-2}(a). Moreover, $\ccc[Z,U]\simeq2K_{2,2}$
is directly connected to $\ccc[X,U]$ at $U$ and to $\ccc[X,Z]$ at $Z$.
Let $Z=\{z_1,z_2,z_3,z_4\}$ and, without loss of generality, suppose that
$\ccc[U,Z]$ and $\ccc[X,Z]$ determine a matching $\{z_1z_2,z_2z_1,z_3z_4,z_4z_3\}$
in $\ccc[Z]$. Thus, $\ccc[Z,U]$ consists of the basis relation
$$
S_Z={\{z_1,z_2\}\times\{u_1,u_2\}}\cup{\{z_3,z_4\}\times\{u_3,u_4\}}
$$
and its complement ${Z\times U}\setminus S_Z$, while $\ccc[X,Z]$ consists of the basis relation
$$
R_Z={\{x_1,x_5\}\times\{z_1,z_2\}}\cup{\{x_3,x_7\}\times\{z_3,z_4\}}  
$$
and its complement ${X\times Z}\setminus R_Z$. Since $f_0$ is induced by $\phi_0$, we have
\begin{equation}
  \label{eq:fSZcc}
f(S_Z)={\{z'_1,z'_2\}\times\{u'_1,u'_2\}}\cup{\{z'_3,z'_4\}\times\{u'_3,u'_4\}},
\end{equation}
where $z'_i=\phi_0(z_i)$.
Since $f$ provides an algebraic isomorphism from $\ccc[X\cup U\cup Z]$
to $\ccc'[X'\cup U'\cup Z']$, Equalities \refeq{fSZcc} and \refeq{fTX} imply that
$$
f(R_Z)={\{x'_1,x'_5\}\times\{z'_1,z'_2\}}\cup{\{x'_3,x'_7\}\times\{z'_3,z'_4\}},
$$
see Figure \ref{fig:proofExclC8-2}(a). Along with the first equality in \refeq{phixxyy}, this shows that 
$$
f(R_Z)=\phi(R_Z).
$$
We see that, as claimed, $\phi$ is a combinatorial isomorphism from
$\ccc[X\cup Z]$ to $\ccc'[X'\cup Z']$ inducing~$f$.

The proof of Lemma \ref{lem:excl-C8} is complete.

\section{Irredundant configurations: Preliminaries}\label{s:2K22}

Along with two other Cut-Down Lemmas (i.e., Lemmas \ref{lem:excl-matching-both} and \ref{lem:excl2points}), 
Lemma \ref{lem:excl-C8} reduces our task to deciding separability of 
a coherent configuration \ccc under the following three conditions:
\begin{enumerate}[(1)]
\item 
\ccc is indecomposable;
\item 
all fibers of \ccc have size 4;
\item 
every non-uniform interspace of \ccc is of type $2K_{2,2}$.
\end{enumerate}
A coherent configuration satisfying Conditions (1)--(3) will be called \emph{irredundant}.

\subsection{Strict algebraic automorphisms}\label{ss:strict}

We begin with noticing that, for irredundant configurations, every algebraic isomorphism $f$
gives rise to a combinatorial isomorphism $\phi$, even though $\phi$ does not need
to induce $f$ on the whole coherent configuration.

\begin{lemma}\label{lem:alg-comb}
Suppose that a coherent configuration \ccc is irredundant.
If $f$ is an algebraic isomorphism from $\ccc$ to a coherent configuration $\ccc'$, then there
exists a combinatorial isomorphism $\phi$ from $\ccc$ to $\ccc'$ such
that $\phi$ induces $f$ on the cell $\ccc[X]$ for each fiber $X\in F(\ccc)$.
\end{lemma}

\begin{proof}
For a fiber $X\in F(\ccc)$, let $X'=f(X)$ denote the corresponding fiber of $\ccc'$. 
  We construct $\phi$ locally as a bijection $\phi\function X{X'}$ for each $X\in F(\ccc)$.
The restriction of $f$ to $\ccc[X]$ is an algebraic isomorphism from the cell $\ccc[X]$
to the cell $\ccc'[X']$. It is easy to check that all 4-point association schemes
are separable. Using this fact, we set $\phi\function X{X'}$ to be a combinatorial isomorphism
from $\ccc[X]$ to $\ccc'[X']$ inducing $f$ on $\ccc[X]$. It remains to show that
$\phi$ defined in this way is also a partition isomorphism from $\ccc[X,Y]$ to
$\ccc'[X',Y']$ for any two fibers $X,Y\in F(\ccc)$.

Assume first that the interspace $\ccc[X,Y]$ is uniform.
Since an algebraic isomorphism preserves the valency of a basis relation 
(see \cite[Corollary 2.3.20]{Ponomarenko-book}), the interspace $\ccc'[X',Y']$
is also uniform. Thus, $\ccc[X,Y]$ consists of the single basis relation
$X\times Y$, and $\ccc'[X',Y']$ consists of the single basis relation
$X'\times Y'$. We are done just because $\phi(X\times Y)=X'\times Y'$,
as trivially follows from the equalities $\phi(X)=X'$ and $\phi(Y)=Y'$.

Assume now that the interspace $\ccc[X,Y]$ is non-uniform, that is, $\ccc[X,Y]\simeq2K_{2,2}$.
Since $f$ preserves the valency of each basis relation in $\ccc[X,Y]$,
we must have either $\ccc'[X',Y']\simeq2K_{2,2}$ or $\ccc'[X',Y']\simeq C_8$.
The latter possibility is actually excluded. Indeed, let $R$ be a basis relation in $\ccc[X,Y]$.
For the matching basis relation $M$ determined by $R$ in the cell $\ccc[X]$
according to Part 2 of Lemma \ref{lem:X2Y4}, we have $p^M_{RR^*}=2$.
If $\ccc'[X',Y']\simeq C_8$, then Part 3 of Lemma \ref{lem:X2Y4}
implies that the cell $\ccc'[X']$ contains a single matching basis relation $M'$.
For this matching we, however, have $p^{M'}_{f(R)f(R)^*}=0$ and, therefore, $M'\ne f(M)$.

Thus, $\ccc'[X',Y']\simeq2K_{2,2}$.
Specifically, suppose that $\ccc[X,Y]$ consists of the basis relation
$$
R={\{x_1,x_2\}\times\{y_1,y_2\}}\cup{\{x_3,x_4\}\times\{y_3,y_4\}}  
$$
and its complement ${X\times Y}\setminus R$, and $\ccc'[X',Y']$ consists of the basis relation
$$
f(R)={\{x'_1,x'_2\}\times\{y'_1,y'_2\}}\cup{\{x'_3,x'_4\}\times\{y'_3,y'_4\}}  
$$
and its complement ${X'\times Y'}\setminus f(R)$.
By Part 2 of Lemma \ref{lem:X2Y4}, $\ccc[X,Y]$ determines the matching basis relations 
$$
M=\{x_1x_2,x_2x_1,x_3x_4,x_4x_3\} 
$$
in the cell $\ccc[X]$ and 
$$
N=\{y_1y_2,y_2y_1,y_3y_4,y_4y_3\} 
$$ 
in the cell $\ccc[Y]$, while $\ccc'[X',Y']$ determines 
$$
M'=\{x'_1x'_2,x'_2x'_1,x'_3x'_4,x'_4x'_3\}
$$
in $\ccc'[X']$ and 
$$
N'=\{y'_1y'_2,y'_2y'_1,y'_3y'_4,y'_4y'_3\}
$$ 
in $\ccc'[Y']$. Since $f$ is an algebraic isomorphism, we have $f(M)=M'$ and $f(N)=N'$.
Since $\phi$ induces $f$ both on $\ccc[X]$ and $\ccc[Y]$, this implies that
$\phi(M)=M'$ and $\phi(N)=N'$. Therefore, either $\phi(\{x_1,x_2\})=\{x_1,x_2\}$
or $\phi(\{x_1,x_2\})=\{x_3,x_4\}$, and either $\phi(\{y_1,y_2\})=\{y_1,y_2\}$
or $\phi(\{y_1,y_2\})=\{y_3,y_4\}$. There are four cases altogether. In two of them
we have $\phi(R)=f(R)$, while $\phi(R)={X'\times Y'}\setminus f(R)$ in the other two cases.
In each case, $\phi$ is a partition isomorphism from $\ccc[X,Y]$ to
$\ccc'[X',Y']$. We conclude that $\phi$ is a combinatorial isomorphism from \ccc to~$\ccc'$.
\end{proof}

Lemma \ref{lem:alg-comb} implies that, if a coherent configuration \ccc is irredundant,
then $\ccc\aiso\ccc'$ implies $\ccc\ciso\ccc'$.
This has the following practical consequence:
An irredundant configuration \ccc is separable 
if and only if every algebraic \emph{automorphism} of \ccc
is induced by a combinatorial automorphism of \ccc.
Moreover, we call an algebraic automorphism $f$ of \ccc \emph{strict}
if $f$ is the identity on each cell $\ccc[X]$ for $X\in F(\ccc)$.

\begin{lemma}
An irredundant coherent configuration \ccc is separable 
if and only if every strict algebraic automorphism of \ccc
is induced by a combinatorial automorphism of~\ccc. 
\end{lemma}

\begin{proof}
The direction `only if' follows directly from the definition of a separable coherent configuration.
For the other direction,
assume that every strict algebraic automorphism of \ccc is induced by a combinatorial automorphism.
We have to prove that \ccc is separable. Let $f$ be an algebraic isomorphism from $\ccc$ to $\ccc'$.
Let $\phi$ be a combinatorial isomorphism $\phi\function{V(\ccc)}{V(\ccc')}$ from $\ccc$ to $\ccc'$ 
as in Lemma \ref{lem:alg-comb}. Consider the composition $g=\phi^{-1}\circ f$, where $\phi^{-1}$
is understood as the induced map from $\ccc'$ to $\ccc$. Since $\phi^{-1}$ is an algebraic automorphism
from $\ccc'$ to $\ccc$, the composition $g$ is an algebraic automorphism of \ccc.
Since $\phi$ induces $f$ on each cell of \ccc, this algebraic automorphism is strict.
Note that $f=\phi\circ g$, where $\phi$ is understood as the induced map from $\ccc$ to $\ccc'$.
By the assumption, $g$ is induced by a combinatorial automorphism $\psi$ of \ccc.
It follows that $f$ is induced by the combinatorial isomorphism $\phi\circ\psi$ from $\ccc$ to~$\ccc'$.
\end{proof}

Denote the set of strict algebraic automorphisms of a coherent configuration
\ccc by $\saa\ccc$. Note that $\saa\ccc$ is a group of permutations of the set \ccc.
Furthermore, let $\saai\ccc$ denote the set of those strict algebraic automorphisms of \ccc
which are induced by combinatorial automorphisms of~\ccc.

\begin{lemma}
$\saai\ccc$ is a subgroup of $\saa\ccc$.
\end{lemma}

\begin{proof}
If $f$ and $h$ are in $\saai\ccc$, then they are induced by combinatorial automorphisms
$\phi$ and $\psi$ respectively. Note that the product $fh$ is induced by the combinatorial
isomorphism $\phi\psi$ and, therefore, $fh$ is in $\saai\ccc$ as well.
\end{proof}

Similarly, we call a combinatorial automorphism $\phi$ of \ccc \emph{strict}
if $\phi$ takes every basis relation in each cell $\ccc[X]$ onto itself. 
Obviously, a strict combinatorial automorphism induces a strict algebraic automorphism.
Moreover, if a combinatorial automorphism $\phi$ induces a strict algebraic automorphism, 
then $\phi$ must be strict itself. We denote the set of strict combinatorial automorphisms
of a coherent configuration \ccc by $\sca\ccc$. Note that $\sca\ccc$ is a group
of permutations on the point set $V(\ccc)$. Furthermore, we call a combinatorial 
automorphism $\phi$ of \ccc \emph{color-preserving} if $R^\phi=R$ for every basis relation
$R\in\ccc$. The term is justified by the fact that $\phi$ is a color-preserving automorphism
of $\ccc$ if any only if $\phi$ is an automorphism of an (arbitrarily chosen)
colored version $\tilde\ccc$ of $\ccc$.
Note that $\phi$ is color-preserving if it induces the identity
$\mathrm{id}_\ccc$. Obviously, a color-preserving combinatorial automorphism is strict,
and the set $\scac\ccc$ of all color-preserving automorphisms is a subgroup of~$\sca\ccc$.

\begin{lemma}\label{lem:quotient}
$\sca\ccc/\scac\ccc\cong\saai\ccc$, where $\cong$ denotes isomorphism of groups.
\end{lemma}

\begin{proof}
Suppose that a permutation $\phi$ of the point set $V(\ccc)$ is a strict combinatorial
automorphism of \ccc. In this case, let $\bar\phi$ denote
the induced permutation of \ccc. The map $\phi\mapsto\bar\phi$ is a homomorphism
from the group $\sca\ccc$ onto the group $\saai\ccc$ whose kernel is $\scac\ccc$.
The lemma immediately follows from the first isomorphism theorem.
\end{proof}

\begin{lemma}\label{lem:sca}
If \ccc is irredundant, then
$\sca\ccc\cong\prod_{X\in F(\ccc)}\scac{\ccc[X]}$.
\end{lemma}

\begin{proof}
To prove that $\sca\ccc$ is isomorphic to a subgroup of $\prod_{X\in F(\ccc)}\scac{\ccc[X]}$, 
consider an arbitrary $\phi$ in $\sca\ccc$.
Since each fiber $X$ is invariant under $\phi$, this permutation is split into
the product $\phi=\prod_{X\in F(\ccc)}\phi_X$, where $\phi_X$ is the identity outside $X$.
By the definition of a strict automorphism, the restriction of $\phi_X$ to $X$
belongs to~$\scac{\ccc[X]}$.

Let us prove that $\prod_{X\in F(\ccc)}\scac{\ccc[X]}$ is isomorphic to a subgroup of $\sca\ccc$.
Consider $\phi=\prod_{X\in F(\ccc)}\phi_X$, where $\phi_X\in\scac{\ccc[X]}$ is defined
outside $X$ by identity. To show that $\phi\in\sca\ccc$,
it is enough to prove that each $\phi_X$ is a strict combinatorial
automorphism of \ccc. This reduces to proving that, for any fiber $Y\ne X$,
the partition $\ccc[X,Y]$ is invariant under $\phi_X$. This is clear if $\ccc[X,Y]$ is
uniform. If the interspace $\ccc[X,Y]$ is non-uniform, that is, consists of two
basis relations $R_1$ and $R_2$, then $\ccc[X,Y]$ determines a matching
basis relation $M$ in the cell $\ccc[X]$. Since $\phi_X(M)=M$ and $\phi_X$
is the identity on $Y$, we see that $\phi_X$ either fixes each of $R_1$ and $R_2$
or transposes them.
\end{proof}

If a coherent configuration \ccc is irredundant,
then every cell $\ccc[X]$ is either of $F_4$-, or $C_4$-, or $\vec C_4$-type; 
see Figure \ref{fig:cells}. In each of these cases, the group $\scac{\ccc[X]}$ is
clear. In particular, if $\ccc[X]$ is of $F_4$-type, i.e., the factorization of $X$ into three 
matchings, then $\scac{\ccc[X]}$ is isomorphic to the Klein four-group. 
Specifically, this is the group $K(X)$ of all such permutations
$\phi\function XX$ that, for every two points $x,x'\in X$, either 
$\phi(\{x,x'\})=\{x,x'\}$ or $\phi(\{x,x'\})=X\setminus\{x,x'\}$.
If $X=\{x_1,x_2,x_3,x_4\}$, then 
$$
K(X)=\{\mathrm{id}_X,(x_1x_2)(x_3x_4),(x_1x_3)(x_2x_4),(x_1x_4)(x_2x_3)\}. 
$$ 
Denote the three matching relations on $X$ by $M$, $N$, and $L$, say,
\begin{eqnarray*}
M   &=&\{x_1x_2,x_2x_1,x_3x_4,x_4x_3\},\\
\cN &=&\{x_1x_3,x_3x_1,x_2x_4,x_4x_2\},\\
\cL &=&\{x_1x_4,x_4x_1,x_2x_3,x_3x_2\}.
\end{eqnarray*}
\ifcolored We will use the depicted colors for these relations
in relevant figures. \fi
Any permutation $\phi$ in $K(X)$ preserves each of the matchings, that is,
$\phi(M)=M$, $\phi(N)=N$, $\phi(L)=L$. If $\phi$ preserves a matching,
then two cases are possible: the matched pairs are either preserved or swapped.
We say that $\phi$ \emph{fixes} the matching in the former case and that
$\phi$ \emph{flips} the matching in the latter case. For example,
$\phi=(x_1x_2)(x_3x_4)$ fixes $M$ and flips each of $N$ and $L$.
To emphasize on this, we will use also the notation $\phi_{NL}=\phi_{LN}=(x_1x_2)(x_3x_4)$.
In this notation,
\begin{equation}
  \label{eq:K}
K(X)=\{\mathrm{id}_X,\phi_{NL},\phi_{ML},\phi_{MN}\},
\end{equation}
where $\phi_{ML}=(x_1x_3)(x_2x_4)$ fixes $N$ and flips both $M$ and $L$,
and similarly for~$\phi_{MN}$. 

Thus, Lemma \ref{lem:sca} gives us a complete explicit description of
the group $\sca\ccc$. A representation of the subgroup $\scac\ccc$
by a set of generators is efficiently computable as explained in \cite{ArvindK06} as this is the automorphism
group of a graph of color multiplicity 4 underlying any colored version $\tilde\ccc$
of \ccc; see Section \ref{ss:gen-case}. 
By Lemma \ref{lem:quotient}, this makes the group $\saai\ccc$ fully comprehensible.
Since \ccc is separable exactly when $\saa\ccc=\saai\ccc$, we can decide
separability if we can efficiently find an explicit description of the group $\saa\ccc$.
From now on, we focus on this task.

Call a permutation $f$ on \ccc \emph{bound} if $f$ is the identity on each cell,
maps each interspace onto itself, and satisfies the condition $f(R^*)=f(R)^*$ 
for every basis relation $R$ of \ccc.
Since the last condition is obeyed by any algebraic isomorphism,
every strict algebraic automorphism is bound.
If $\ccc[X,Y]\simeq 2K_{2,2}$, then for a bound permutation $f$ there are two possibilities. 
Specifically, suppose that $\ccc[X,Y]$ partitions $X\times Y$ into
two parts $R_1$ and $R_2$. We say that $f$ \emph{fixes} $\ccc[X,Y]$ if
$$
f(R_i)=R_i
$$
and that $f$ \emph{switches} $\ccc[X,Y]$ if
$$
f(R_i)=R_{3-i}
$$
for $i=1,2$.
Note that, if $f$ switches $\ccc[X,Y]$, then it switches also $\ccc[Y,X]$.
Given a set $S$ of pairs $\{X,Y\}$ such that $\ccc[X,Y]$ is non-uniform, let $f_S$ denote the bijection 
from \ccc onto itself which switches the interspace $\ccc[X,Y]$ as well
as the interspace $\ccc[Y,X]$ for each $\{X,Y\}\in S$ and leaves the rest
of \ccc fixed. Thus, every bound permutation of \ccc
coincides with $f_S$ for some $S$. Conversely, every $f_S$ is a bound permutation,
but not all $f_S$ must be algebraic automorphisms.

Thus, we have to describe the class of those $S$ for which $f_S$
is an algebraic automorphism. 
We begin our analysis with two instructive special cases in
Subsections \ref{ss:CFI} and \ref{ss:3-reg}, and then consider
the general case in Subsection \ref{ss:gen-case}. Prior to this
we do some preliminary work in Subsection~\ref{ss:3fibers}.

\subsection{The case of three fibers}\label{ss:3fibers}

Let $f$ be a bound permutation of \ccc.
Lemma \ref{lem:cl:b} reduces verification of whether or not $f$
is a strict algebraic automorphism of \ccc to local verification of
this on all 3-fiber subconfigurations $\ccc[X\cup Y\cup Z]$.
Thus, the case of coherent configurations with three fibers is
quite important and we consider it here.

We call an irredundant configuration $\ccc$ \emph{skew-connected} if $\ccc$ 
contains no directly connected interspaces.

\begin{lemma}\label{lem:skew+ddirect}
Let  \ccc be an irredundant coherent configuration with $F(\ccc)=\{X,Y,Z\}$ and 
$f$ be a bound permutation of the set of basis relations of \ccc.
  \begin{enumerate}
  \item 
If \ccc is skew-connected, then $f$ is an algebraic automorphism of~\ccc.
\item 
Suppose that \ccc is not skew-connected.
Then $f$ is an algebraic automorphism of \ccc if and only if either $f=\mathrm{id}_\ccc$ or
$f$ makes exactly two switches of interspaces
(switching an interspace and its transpose is counted as a single switch).
  \end{enumerate}
\end{lemma}

\begin{proof}
\textit{1.} 
It suffices to show that $f$ is induced by some combinatorial automorphism $\phi$ of \ccc.
Let $f_{XY}$ denote the restriction of $f$ to $\ccc[X,Y]\cup\ccc[Y,X]$ and extend it
to the whole \ccc by identity. Define $f_{YZ}$ and $f_{XZ}$ similarly.
Since $f=f_{XY}\circ f_{YZ}\circ f_{XZ}$, it is enough to check that
each of the three permutations of \ccc are induced by a combinatorial automorphism of \ccc.
We show this for $f_{XY}$, and the same argument applies as well for the other two cases.
If $\ccc[X,Y]$ is uniform or if $f$ fixes $\ccc[X,Y]$, then $f=\mathrm{id}_{\ccc}$ is induced by $\mathrm{id}_{X\cup Y\cup Z}$.
Suppose that $f$ switches $\ccc[X,Y]$. Let $M$ be the matching basis relation of
$\ccc[X]$ determined by the interspace $\ccc[Y,X]$ according to Lemma \ref{lem:X2Y4}.
If the interspace $\ccc[Z,X]$ is non-uniform, it determines another matching relation $L$ in $\ccc[X]$.
Let $N$ be the matching relation on $X$ different from $M$ and also from $L$ if the last exists.
Consider the permutation $\phi_{MN}$ of $X$ and extend it also to $Y\cup Z$ by identity.
It remains to notice that $\phi_{MN}$ is a combinatorial automorphism of \ccc and that it induces~$f_{XY}$.

\textit{2.} 
Since \ccc is not skew-connected, it contains two non-uniform directly connected interspaces
and, therefore, Lemma \ref{lem:trans} implies that all interspaces of \ccc are non-uniform
and all connections between them are direct. 
Recall that $f$ is an algebraic automorphism if and only if
$p^R_{ST}=p^{f(R)}_{f(S)f(T)}$ for all triples of basis relations $R,S,T\in\ccc$.
This equality holds for every bound $f$ if all three relations
$R,S,T$ are either in $\ccc[X\cup Y]$, or in $\ccc[X\cup Z]$, or in $\ccc[Y\cup Z]$;
cf.\ the proof of Part 1.
The only situation that requires some care is when every interspace of \ccc contains
one of the relations $R$, $S$, and $T$ or their transposes.
In any case of this kind, we have either $p^R_{ST}=2$ or $p^R_{ST}=0$.
The lemma follows from the observation that the value of $p^R_{ST}$ switches
from 2 to 0 or vice versa whenever we complement one of the relations.
Specifically, let $R^c={X\times Y}\setminus R$ for a basis relation $R\in\ccc[X,Y]$.
Using this notation also for other interspaces, we have
\begin{equation}
  \label{eq:gammaRc}
p^R_{ST}\ne p^{R^c}_{ST}=p^R_{S^cT}=p^R_{ST^c}  
\end{equation}
for any triple $R,S,T$ under consideration.
The inequality in \refeq{gammaRc} implies that any $f$ making exactly one switch
is not an algebraic automorphism. Applying \refeq{gammaRc} twice, we see that
$f$ making two switches is an algebraic automorphism. Applying \refeq{gammaRc} once
again, we see that $f$ making three switches is not an algebraic automorphism.
\end{proof}

\section{Irredundant configurations: The CFI case}\label{ss:CFI}

Let \ccc be an irredundant configuration. Like in the case of reduced Klein configurations \cite{EvdokimovP99},
we define the \emph{fiber graph} of \ccc, denoted by $\fg\ccc$, as follows:
\begin{itemize}
\item 
The vertices of $\fg\ccc$ are the fibers of \ccc, i.e., $V(\fg\ccc)=F(\ccc)$;
\item 
Two fibers $X$ and $Y$ are adjacent in $\fg\ccc$ if the interspace $\ccc[X,Y]$
is non-uniform.
\end{itemize}

\noindent
Recall that irredundant configurations are indecomposable.
This implies that $\fg\ccc$ is connected.

As usually, $\Delta(G)$ (resp., $\delta(G)$) denotes the maximum (resp., minimum)
degree of a vertex in the graph~$G$.

We now consider coherent configurations corresponding to the classical CFI construction \cite{CaiFI92}
of graphs of color multiplicity 4 not amenable to \kWL. 
Recall that an irredundant configuration $\ccc$ is skew-connected if $\ccc$
contains no directly connected interspaces.
Note that $\Delta(\fg\ccc)\le3$ in this case.
Part 3 of the following lemma is reminiscent of \cite[Lemma 6.2]{CaiFI92}.

\begin{lemma}\label{lem:CFI}
If \ccc is skew-connected, then the following is true.
\begin{enumerate}[\bf 1.]
\item 
$\saa\ccc=\setdef{f_S}{S\subseteq E(\fg\ccc)}$.  
\item 
If $\delta(\fg\ccc)\le2$, then every $f_S$ is induced by
a combinatorial automorphism of \ccc.
\item 
If $\delta(\fg\ccc)=3$, i.e., $\fg\ccc$ is a regular graph of degree 3,
then $f_S$ is induced by a combinatorial automorphism of \ccc
exactly when $|S|$ is even.
\end{enumerate}
\end{lemma}

\begin{proof}
\textit{1.}
This part follows directly from Lemma \ref{lem:cl:b} and Part 1 of Lemma~\ref{lem:skew+ddirect}.

\textit{2.}
Note that $f_S\circ f_{S'}=f_{S\triangle S'}$, where $\circ$ denotes the group operation
in $\saa\ccc$, i.e., the composition of permutations. This implies that the
set $\saa\ccc=\Set{f_S}_S$ is a commutative group with every element having order 2. 
For $s\in E(\fg\ccc)$, denote
$f_s=f_{\{s\}}$. The group $\saa\ccc$ is generated by the set
$\setdef{f_s}{s\in E(\fg\ccc)}$. Indeed, if $S=\{s_1,\ldots,s_k\}$, then
obviously $f_S=f_{s_1}\circ\ldots\circ f_{s_k}$. Therefore, it suffices to prove that,
if $\delta(\fg\ccc)\le2$, then each $f_s$ is induced by a combinatorial automorphism of~\ccc.

Let $s=\{X,Y\}$.
Suppose first that the degree of $X$ in $\fg\ccc$ is 2.
This means that $X$ is incident to two non-uniform interspaces of \ccc.
One of them is $\ccc[X,Y]$, and let $\ccc[X,Z]$ be the other one.
By Part 2 of Lemma \ref{lem:X2Y4}, each of the interspaces $\ccc[Y,X]$ and $\ccc[Z,X]$
determines a matching in the cell $\ccc[X]$. Denote these matchings by $M$ and $L$ respectively
and note that $M\ne L$ because \ccc is skew-connected. Therefore, $\ccc[X]\simeq F_4$.
Let $N$ be the third matching in $\ccc[X]$. Consider the permutation $\phi_{MN}$
as in \refeq{K} and extend it to a permutation of the entire point set $V=V(\ccc)$
by identity outside $X$. Since $N$ is not determined by any incident interspace, $f_s$ is induced by~$\phi_{MN}$.

Suppose now that $X$ has degree 1 in $\fg\ccc$. As above,
let $M$ be the matching determined in $\ccc[X]$ by the interspace $\ccc[Y,X]$.
Furthermore, let $N$ be another matching relation on $X$.
If $\ccc[X]\simeq F_4$ or $\ccc[X]\simeq C_4$, then $f_s$ is induced by $\phi_{MN}$
by the same reason as above. If $\ccc[X]\simeq \vec C_4$, then this does not work
because $\phi_{MN}$ is not a strict combinatorial automorphism of $\ccc[X]$.
In this case, let $x_1,x_2,x_3,x_4$ be an enumeration of $X$ along a non-matching
basis relation of $\ccc[X]$ (which is a directed 4-cycle). 
Then $f_s$ is induced by the cyclic permutation $(x_1x_2x_3x_4)$
because, as easily seen, this permutation flips~$M$.

The case that both $X$ and $Y$ have degree 3 in $\fg\ccc$ can be reduced to
the case above. Indeed, suppose that $s_1=\{A,B\}$ and $s_2=\{A,C\}$ are two
adjacent edges in $\fg\ccc$. Let $M$ be the matching in the cell $\ccc[A]$ determined
by the interspace $\ccc[B,A]$ and $N$ be the matching in $\ccc[A]$ determined
by $\ccc[C,A]$. Note that
\begin{equation}
  \label{eq:ffphi}
  f_{s_1}\circ f_{s_2}=\phi_{MN},
\end{equation}
where $\phi_{MN}$ is the permutation of \ccc induced by the permutation of $V$
that flips each of $M$ and $N$ and is the identity outside $A$.
Since all permutations under consideration are involutive, we infer from \refeq{ffphi} that
$$
f_{s_1}=\phi_{MN}\circ f_{s_2}.
$$
It immediately follows that $f_{s_1}$ is induced by a combinatorial automorphism
if and only if $f_{s_2}$ is induced by a combinatorial automorphism.
By the connectedness of $\fg\ccc$, this implies that all $f_s$ are induced 
by combinatorial automorphisms if this is true for at least one of them,
which is the case as we already know.

\textit{3.}
We first prove by induction on $n$ that, if $|S|=2n$, then $f_S$ is induced by a combinatorial 
automorphism. If $n=0$, then $f_\emptyset=\mathrm{id}_{\ccc}$, and the claim is trivially true.

Consider next the case that $n=1$, i.e., $S=\{s_1,s_2\}$.
If $s_1$ and $s_2$ are adjacent in $\fg\ccc$, then $f_S=f_{s_1}\circ f_{s_2}$
is induced by a combinatorial automorphism $\phi_{MN}$ as in \refeq{ffphi}.
Otherwise, consider a sequence $s_1,r_1,\ldots,r_k,s_2$ of successive edges
along a path in $\fg\ccc$. Such a path exists because $\fg\ccc$ is connected.
Note that
\begin{multline*}
f_S=f_{s_1}\circ f_{s_2}=f_{s_1}\circ (f_{r_1}\circ f_{r_1})\circ\ldots\circ(f_{r_k}\circ f_{r_k})\circ f_{s_2}\\=
(f_{s_1}\circ f_{r_1})\circ(f_{r_1}\circ f_{r_2})\circ \ldots
\circ(f_{r_{k-1}}\circ f_{r_k})\circ(f_{r_k}\circ f_{s_2}).
\end{multline*}
Since each of the factors $f_{s_1}\circ f_{r_1}$, $f_{r_i}\circ f_{r_{i+1}}$, and $f_{r_k}\circ f_{s_2}$
is induced by a combinatorial automorphism, this is true also for~$f_S$.

Suppose now that $n>1$. Let $s_1$ and $s_2$ be two elements of $S$. We have
$$
f_S=f_{\{s_1,s_2\}}\circ f_{S\setminus\{s_1,s_2\}}.
$$
By the induction assumption, both factors are induced by a combinatorial automorphism, 
so this must be true as well for~$f_S$.

For the other direction, assume that $|S|$ is odd. Let $s\in S$. We have
$$
f_S=f_s\circ f_{S\setminus\{s\}}.
$$
We already know that the second factor is induced by a combinatorial automorphism.
This implies that $f_S$ is induced by a combinatorial automorphism if and only if this is so
for $f_s$. Therefore, it suffices\footnote{As an alternative argument, note that, if we show
that at least one $f_s$ does not belong to $\saai\ccc$, then this will imply that $\saai\ccc$
is a subgroup of $\saa\ccc$ of index~2.} 
to prove that $f_s$ is not induced by a combinatorial automorphism for any~$s$. 

Assume to the contrary that $f_s$ is induced by a 
combinatorial automorphism $\phi$. Note that $\phi$ must be strict and, by Lemma \ref{lem:sca},
\begin{equation}
  \label{eq:phiprod}
\phi=\prod_{X\in F(\ccc)}\phi_X,  
\end{equation}
where $\phi_X\in K(X)$ is defined outside $X$ by identity.
Consider a factor $\phi_X$ that is non-identity; at least one such factor
must exist. Suppose that $\phi_X=\phi_{MN}$ for matchings $M$ and $N$ in $\ccc[X]$.
Since $X$ has degree 3 in $\fg\ccc$, the matching $M$ must be determined by an interspace $\ccc[Y,X]$,
and $N$ must be determined by an interspace $\ccc[Z,X]$. Let $s_1=\{X,Y\}$ and
$s_2=\{X,Z\}$. Accordingly with \refeq{ffphi}, we now have
$$
\phi_X=f_{s_1}\circ f_{s_2},
$$
where $\phi_X$ is understood as the induced permutation of \ccc.
Along with \refeq{phiprod}, this implies that
$$
f_s=f_{s_1}\circ\ldots\circ f_{s_{2k}}
$$
or, equivalently,
$$
f_s\circ f_{s_1}\circ\ldots\circ f_{s_{2k}}=\mathrm{id}_{\ccc}.
$$
After all possible cancellations of pairs of equal factors,
the product in the left hand side is non-empty and consists of pairwise
distinct factors, which yields a contradiction.
\end{proof}

\begin{corollary}\label{cor:skew}
  A skew-connected coherent configuration \ccc is separable if and only if $\delta(\fg\ccc)\le2$.
\end{corollary}

\begin{remark}
Let \ccc be a skew-connected coherent configuration with $\delta(\fg\ccc)=3$.
Denote the number of fibers in \ccc by $n$. Then the fiber graph $\fg\ccc$ has $n$ vertices
and $m=\frac32\,n$ edges. We have $|\saa\ccc|=2^m$ by Part 1 of Lemma \ref{lem:CFI}
and $|\saai\ccc|=2^{m-1}$ by Part 3 of this lemma. From Lemma \ref{lem:sca} it follows
that $|\sca\ccc|=4^n$. Lemma \ref{lem:quotient}, therefore, implies that
$|\scac\ccc|=2^{2n-m+1}=2^{m-n+1}$. This equality agrees with the fact, which can 
be derived from~\refeq{ffphi}, that color-preserving automorphisms of \ccc are in 
one-to-one correspondence with Eulerian subgraphs of $\fg\ccc$. 
Recall that a graph is called Eulerian if its every vertex has even degree.
All Eulerian subgraphs of a connected graph with $n$ vertices and $m$ edges
form the cycle space, which is a vector space of dimension $m-n+1$ over the two-element field.
\end{remark}

\section{Irredundant configurations: The 3-harmonious case}\label{ss:3-reg}

\subsection{The hypergraph of direct connections}\label{ss:fg-dch}

Let \ccc be an irredundant configuration.
Suppose that $\ccc[X,Y]$ is a non-uniform interspace.
We define $D(X,Y)$ to be the set of fibers
consisting of $X$, $Y$, and all $Z$ such that $\ccc[Z,X]$ is non-uniform and 
directly connected with $\ccc[Y,X]$. Let $\dcc\ccc$ denote
the family of all sets $D(X,Y)$ over non-uniform interspaces $\ccc[X,Y]$.
We regard $\dcc\ccc$ as a hypergraph on $F(\ccc)$ and call it the
\emph{hypergraph of direct connections} of~\ccc.

Recall that the degree of a vertex $v$ in a hypergraph $H$ is the number
of hyperedges of $H$ containing $v$. Similarly to graphs, $\Delta(H)$ (resp., $\delta(H)$) 
denotes the maximum (resp., minimum) degree of a vertex in the hypergraph~$H$.
The following properties of irredundant configurations are known for 
\emph{reduced Klein configurations} \cite{MuzychukP12}; see also \cite[Section 4.1.2]{Ponomarenko-book}.
The two classes of coherent configurations are closely related but not identical.
In particular, a reduced Klein configuration cannot contain $C_4$-cells.

\begin{lemma}[{cf.~\cite[Lemma 4.1.18]{Ponomarenko-book}}]\label{lem:dcc}\hfill
  \begin{enumerate}[\bf 1.]
\item 
$1\le\delta(\dcc\ccc)\le\Delta(\dcc\ccc)\le3$. Moreover, every edge $\{X,Y\}$ of $\fg\ccc$
can be extended to a hyperedge of~$\dcc\ccc$.
  \item 
Every hyperedge of $\dcc\ccc$ is a clique in $\fg\ccc$,
and all interspace connections within this clique are direct.
\item 
Any two hyperedges of $\dcc\ccc$ have at most one common vertex.
  \end{enumerate}
\end{lemma}

\begin{proof}
Part 1 follows from the definitions and the obvious fact
that a cell contains at most three matchings where interspaces can be
directly connected to each other. Lemma \ref{lem:trans} implies that,
if $A$ and $B$ are two fibers in $D(X,Y)$, then the interspace
$\ccc[A,B]$ is non-uniform and $D(A,B)=D(X,Y)$. This yields Parts 2 and~3.
\end{proof}

Lemma \ref{lem:dcc} shows that $\dcc\ccc$ is a clique edge partition of $\fg\ccc$.
This implies that the fiber graph is reconstructable from the hypergraph of direct connections.
Indeed, $\fg\ccc$ is the Gaifman graph of $\dcc\ccc$, that is, two fibers are
adjacent in $\fg\ccc$ if and only if they are both contained in some hyperedge of~$\dcc\ccc$.

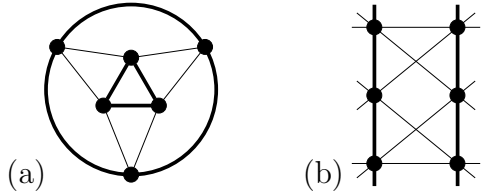
\begin{figure}
\centering
\begin{tikzpicture}[every node/.style={circle,draw,black,
  inner sep=2pt,fill=black},
  lab/.style={draw=none,fill=none,inner sep=0pt,rectangle},
  longer/.style={shorten >=-4mm,shorten <=-4mm},
  vt/.style={line width=1.4pt},
  line width=0.2pt,
  outer sep=0pt,
]
  \node[inner sep=8mm,fill=none,vt] (outer) {};
  \node[inner sep=3mm,fill=none,draw=none] (inner) {};
  \path 
  (outer.150) node (1) {}
  (outer.30)  node (2) {}
  (outer.270) node (3) {}
  (inner.90)  node (4) {}
  (inner.210) node (5) {}
  (inner.330) node (6) {};
  \draw[vt] (4) -- (5) -- (6) -- (4);
  \draw (1) -- (4) -- (2) -- (6) -- (3) -- (5) -- (1);
  \node[lab,anchor=east,
    at={(outer.south -| outer.west)}] {(a)};
  \matrixgraph[name=m1,nolabel,matrix anchor=north west,
  at={($(outer.north -| outer.east)+(20mm,-2mm)$)}]
    {&[9mm]\\
    1 & 4\\[7mm]
    2 & 5\\[7mm]
    3 & 6\\
  }{
    1 --[vt,longer] 3;
    4 --[vt,longer] 6;
    1 --[longer]  {4,5};
    2 --[longer]  {4,6};
    3 --[longer]  {5,6};
  };
  \node[lab,anchor=east,
    at={($(outer.south -| outer.east)+(17mm,0mm)$)}] {(b)};
\end{tikzpicture}
\caption{(a) A hypergraph of direct connections $\dcc\ccc$ shown as a family of 
two 3-cliques, marked in bold, and six 2-cliques in the fiber graph~$\fg\ccc$.
(b) A geometric representation of~$\dcc\ccc$.}
\label{fig:PLS}
\end{figure}

Part 3 of Lemma \ref{lem:dcc} says exactly that $\dcc\ccc$ is a \emph{linear hypergraph}.
Linear hypergraphs with each hyperedge of size at least 2 are known
in incidence geometry \cite{Bruyn16,Metsch91}
as \emph{partial linear spaces}. Here vertices of a hypergraph are interpreted
as \emph{points} and hyperedges as \emph{lines}; see Figure \ref{fig:PLS},
though not every partial linear space admits a geometric realization.
A relationship between reduced Klein configurations and partial linear spaces
was noticed in \cite[Corollary 4.1.19]{Ponomarenko-book}.
Lemma \ref{lem:equi-chain} below shows that, under certain conditions, a coherent configuration is
uniquely determined by its hypergraph of direct connections, and that
partial linear spaces are a rich source of templates for constructing
coherent configurations.
The following elementary fact will be useful in the proof of Lemma \ref{lem:equi-chain} and also later.

\begin{lemma}\label{lem:matchmatch}\hfill
  \begin{enumerate}[\bf 1.]
  \item 
Let $M_1,M_2,M_3$ be the three matching relations on a 4-element set $X$
numbered in an arbitrary order. Then there is a permutation $\phi$ of $X$
such that $\phi(M_1)=M_1$, $\phi(M_2)=M_3$, and $\phi(M_3)=M_2$.
\item 
Let $M_1,M_2,M_3$ be the three matching relations on a 4-element set $X$
and $M'_1,M'_2,M'_3$ be the three matching relations on a 4-element set $X'$,
numbered in an arbitrary order. Then there is a bijection $\psi$ from $X$
onto $X'$ such that $\psi(M_i)=M'_i$ for each $i=1,2,3$.
  \end{enumerate}  
\end{lemma}

\begin{proof}
Part 1 is straightforward; cf.\ the discussion in Section \ref{ss:strict}.
Part 2 easily follows from Part~1.
\end{proof}

Recall that a hypergraph is called \emph{connected} if its Gaifman graph is connected.

\begin{lemma}\label{lem:equi-chain}\hfill
  \begin{enumerate}[\bf 1.]
  \item 
Let \ccc be an irredundant configuration.
If $\ccc\aiso\ccc'$, then $\dcc{\ccc}\cong\dcc{\ccc'}$, where $\cong$ denotes isomorphism of hypergraphs.
\item 
Under the condition $\delta(\dcc{\ccc})\ge2$, $\dcc{\ccc}\cong\dcc{\ccc'}$ implies that $\ccc\ciso\ccc'$.
\item 
For any connected partial linear space $D$ with $\Delta(D)\le3$ there is an irredundant configuration \ccc
such that $\dcc{\ccc}\cong D$.
  \end{enumerate}
\end{lemma}

\begin{proof}
\textit{1.}  
This part follows from the fact that an algebraic isomorphism respects fibers, non-uniformity of interspaces,
and direct connections of interspaces.

\textit{2.}  
Let $h\function{F(\ccc)}{F(\ccc')}$ be an isomorphism from the hypergraph
$\dcc{\ccc}$ to the hypergraph $\dcc{\ccc'}$.
Based on $h$, we define a bijection $\bar h$ from the set of all matching basis relations
of \ccc to the set of all matching basis relations of $\ccc'$.
Consider a fiber $X\in F(\ccc)$. Let $C_1$ and $C_2$ be two hyperedges of $\dcc{\ccc}$
containing $X$. All interspaces $\ccc[Y,X]$ for $Y\in C_1$ determine the same matching
in the cell $\ccc[X]$, which we denote by $M_1$. All interspaces $\ccc[Y,X]$ for $Y\in C_2$ 
determine a matching $M_2$, different from $M_1$. 
Thus, $\ccc[X]\simeq F_4$. Denote the third matching in $\ccc[X]$
by $M_3$. Similarly, the interspaces $\ccc'[Y',h(X)]$ for $Y'\in h(C_1)$ determine a matching
$M'_1$, and the interspaces $\ccc'[Y',h(X)]$ for $Y'\in h(C_2)$ 
determine a matching $M'_2\ne M'_1$ in $\ccc'[h(X)]$. Denote the third matching in $\ccc'[h(X)]$
by $M'_3$ and set $\bar h(M_i)=M'_i$ for $i=1,2,3$.
Let $\psi_X$ be a bijection from $X$ onto $h(X)$ such that $\psi_X(M)=\bar h(M)$ for each
matching $M$ in $\ccc[X]$.
Such a bijection exists by Part 2 of Lemma \ref{lem:matchmatch}.
Combining all $\psi_X$ over $X\in F(\ccc)$, we obtain a bijection from $V(\ccc)$ onto $V(\ccc')$ 
which is a combinatorial isomorphism from \ccc to~$\ccc'$.

\textit{3.}
Given $D$, we construct \ccc as follows. Each point $p$ of $D$ gives rise to
a 4-point fiber $X_p$ in \ccc, with the cell $\ccc[X_p]$ being of type $F_4$. 
With each hyperedge $C$ of $D$ containing $p$, 
we associate a matching relation $M_{p,C}$ in $\ccc[X_p]$ such that $M_{p,C}\ne M_{p,C'}$
if $C\ne C'$. For each pair of points $p$ and $q$ in the same hyperedge $C$,
we make the interspace $\ccc[X_p,X_q]$ non-uniform so that it determines
the matching $M_{p,C}$ in $\ccc[X_p]$ and the matching $M_{q,C}$ in $\ccc[X_q]$.
\end{proof}

Without the assumption $\delta(\dcc{\ccc})\ge2$ in Part 2 of Lemma \ref{lem:equi-chain},
a coherent configuration \ccc cannot be uniquely reconstructed from $\dcc{\ccc}$
because, if a fiber $X$ has degree 1 in $\dcc{\ccc}$, then the cell $\ccc[X]$
can be not only of type $F_4$ but also of type $C_4$ or~$\vec C_4$.

Note that \ccc is skew-connected exactly when $|C|=2$ for all $C\in\dcc\ccc$,
that is, $\dcc\ccc$ is just the edge set of the graph $\fg\ccc$.
Part 2 of Lemma \ref{lem:equi-chain}, therefore, implies that, if 
\ccc is a skew-connected coherent configuration with $\delta(\fg\ccc)\ge2$, then
the isomorphism $\ccc\ciso\ccc'$ is equivalent to the isomorphism $\fg{\ccc}\cong\fg{\ccc'}$.

\begin{remark}\label{rem:multipede}
Curiously, Lemma \ref{lem:equi-chain} reveals a connection between
irredundant coherent configurations and the multipede graphs introduced
by Neuen and Schweitzer in \cite{NeuenS17}. Let \ccc be an irredundant configuration
and assume for the hypergraph of direct connections of \ccc that
$\delta(\dcc{\ccc})=3$. Consistently with the notation in \cite{NeuenS17},
denote the incidence graph of the hypergraph $\dcc{\ccc}$ by $G=G(V,W)$,
where $V=F(\ccc)$ is the vertex set of $\dcc{\ccc}$, i.e., the set of all
fibers of \ccc, and $W$ is the set of the hyperedges  of $\dcc{\ccc}$, i.e.,
the cliques of directly connected fibers. Two vertices $v\in V$ and $w\in W$
are adjacent in $G$ if $v$ belongs to $w$. Thus, every vertex in $V$ has
degree 3 in $G$. Any such bipartite graph $G$ determines a multipede graph
denoted in \cite{NeuenS17} by $R(G)$. This is a vertex-colored graph with
vertex classes of size 4 and 2. Since we started from an irredundant
configuration \ccc, the coloring of $R(G)$ is not refinable by \WL2,
and each color class of $R(G)$ stays as a fiber in the coherent closure $\ccc(R(G))$.
Let $\ccc'$ be the coherent configuration obtained from $\ccc(R(G))$
by cutting down all fibers of size 2 (cf.\ Lemma \ref{lem:excl2points}).
Lemma \ref{lem:equi-chain} implies that $\ccc'$ is combinatorially isomorphic to~\ccc.
\end{remark}

\subsection{Separability of 3-harmonious configurations}

We say that a coherent configurations $\ccc$ is \emph{$3$-harmonious} if 
the following three conditions are met:
\begin{itemize}
\item 
\ccc is irredundant;
\item 
Every fiber of \ccc belongs to exactly three cliques in $\dcc\ccc$
(that is, $\dcc\ccc$ is a \emph{$3$-regular hypergraph});
\item 
$|C|=3$ for all $C\in\dcc\ccc$ (that is, $\dcc\ccc$ is a \emph{$3$-uniform hypergraph}).
\end{itemize}

If \ccc is 3-harmonious, then the incidence graph of the hypergraph $\dcc\ccc$
is a cubic bipartite graph, which readily implies the equality
\begin{equation}\label{eq:DFeq}
|\dcc\ccc|=|F(\ccc)|.  
\end{equation}

\begin{lemma}\label{lem:3-switches}
 If \ccc is 3-harmonious, then $f_S\in\saa\ccc$ exactly in the following case:
For each clique $C\in\dcc\ccc$, the switch set $S$ contains two or no
edges of~$C$.
\end{lemma}

\begin{proof}
    ($\implies$)
Suppose that $C\in\dcc\ccc$ and $C=\{X,Y,Z\}$. 
If $f_S\in\saa\ccc$, then the restriction of $f_S$ to $\ccc[X\cup Y\cup Z]$
is a strict automorphism of this subconfiguration.
By Part 2 of Lemma \ref{lem:skew+ddirect}, $f_S$ is either identity on $\ccc[X\cup Y\cup Z]$
or switches exactly two edges of~$C$.

  ($\Longleftarrow$)
This part follows directly from Lemma \ref{lem:skew+ddirect} by
Lemma~\ref{lem:cl:b}.
\end{proof}

\begin{lemma}\label{lem:Gccc}
A 3-harmonious coherent configuration \ccc is separable if and only if
the group of color-preserving automorphisms $\scac\ccc$ is trivial.
\end{lemma}

\begin{proof}
Lemma \ref{lem:3-switches} implies that 
$$
|\saa\ccc|=4^{|\dcc\ccc|}.
$$
Since every cell of a 3-harmonious coherent configuration is of type $F_4$,
Lemma \ref{lem:sca} implies that $\sca\ccc\cong\prod_{X\in F(\ccc)}K(X)$
and, therefore, 
$$
|\sca\ccc|=4^{|F(\ccc)|}.
$$ 
Taking into account Equality \refeq{DFeq},
we conclude by Lemma \ref{lem:quotient} that $\saai\ccc=\sca\ccc$
if and only if $|\scac\ccc|=1$.
\end{proof}

Lemma \ref{lem:Gccc} provides an efficient separability test for 3-harmonious
coherent configurations. Given a 3-harmonious coherent configuration \ccc,
we construct a vertex colored graph $G(\ccc)=G$ as follows:
\begin{itemize}
\item 
$V(G)=V(\ccc)$.
\item 
The vertex color classes of $G$ are exactly the fibers of~\ccc.
\item 
Every vertex color class of $G$ is an independent set.
\item 
For two disjoint sets $X$ and $Y$ of vertices of $G$, let $G[X,Y]$ denote the
subgraph of $G$ on the vertex set $X\cup Y$ formed by the edges between a vertex in $X$
and a vertex in $Y$. For each non-uniform interspace $\ccc[X,Y]$, we set
$G[X,Y]$ to be one of the two $2K_{2,2}$ graphs underlying the basis relations of $\ccc[X,Y]$.
We do not specify which of the two relations in $\ccc[X,Y]$ shall be
used to construct $G[X,Y]$ as this is irrelevant for our purpose.
\end{itemize}

Note that $\phi$ is a color-preserving automorphism of \ccc exactly when
$\phi$ is an automorphism of $G(\ccc)$. Since $G(\ccc)$ has color multiplicity 4,
whether it has a non-trivial automorphism is efficiently verifiable by the known techniques~\cite{ArvindK06,Luks86}.

We now show that there exist 3-harmonious coherent configurations of both sorts ---
separable and non-separable. Recall that the hypergraph of direct connections $\dcc\ccc$
can be viewed as a partial linear space. Moreover, if \ccc is 3-harmonious, then every line
contains exactly 3 points and every point is incident to exactly 3 lines.
Partial linear spaces with $n$ points having these properties are known as \emph{$(n_3)$-configurations};
see \cite{Gruenbaum,PisanskiS}. Lemma \ref{lem:equi-chain} implies a one-to-one correspondence
between $(n_3)$-configurations and 3-harmonious coherent configurations with $n$ fibers.

\begin{figure}
\centering
\begin{tikzpicture}[every node/.style={circle,draw,black,
  inner sep=2pt,fill=black},
  lab/.style={draw=none,fill=none,inner sep=0pt,rectangle},
  longer/.style={shorten >=-4mm,shorten <=-4mm},
  vt/.style={line width=1.4pt},
  line width=0.2pt,
  outer sep=0pt,
]
  \node[isosceles triangle,
    inner sep=6mm,
    isosceles triangle apex angle=60,
    shape border uses incircle,
    shape border rotate=90,vt,
    fill=none] (fanotr) {}
  ;
  \path 
    (fanotr.center) node (0) {}
    (fanotr.30)     node (4) {}
    (fanotr.90)     node (1) {}
    (fanotr.150)    node (2) {}
    (fanotr.210)    node (6) {}
    (fanotr.270)    node (5) {}
    (fanotr.330)    node (3) {}
  ;
  \draw[vt] (0) -- (2) -- (3)
    (4) -- (6) -- (0)
    (5) -- (0) -- (1)
  ;
  \node[at=(0),inner sep=6mm,fill=none,vt] {};
  
  \path (43mm,0pt) node[isosceles triangle,
    inner sep=6mm,
    isosceles triangle apex angle=60,
    shape border uses incircle,
    shape border rotate=90,vt,
    fill=none] (mktr) {};
  \path 
    (mktr.30)       node (3) {}
    (mktr.90)       node (4) {}
    (mktr.150)      node (7) {}
    (mktr.210)      node (6) {}
    (mktr.270)      node (0) {}
    (mktr.330)      node (1) {}
    ($(7)!0.5!(0)$) node (5) {}
    ($(3)!0.5!(0)$) node (2) {}
  ;
  \draw[vt] (0) -- (2) -- (3)
    (3) -- (5) -- (6)
    (5) -- (7) -- (0)
    (7) -- (1) -- (2)
  ;
  \node[inner sep=7mm,anchor=north,
    at=(mktr.90),fill=none,vt] {};
    
  \matrixgraph[name=pappus,nolabel,matrix anchor=north west,
  at={($(mktr.north -| mktr.east)+(18mm,-4mm)$)}]
    {&[7mm]&[7mm]&[7mm]&[7mm]\\
    1 &   & 2 &  & 3\\[6.2mm]
      & 4 & 5 & 6   \\[6.2mm]
    7 &   & 8 &  & 9\\
  }{
    1 --[vt,longer] 3;
    4 --[vt,longer] 6;
    7 --[vt,longer] 9;
    1 --[longer,vt]  5 --[longer,vt] 9;
    3 --[longer,vt]  5 --[longer,vt] 7;
    1 --[longer,vt]  4 --[longer,vt] 8;
    2 --[longer,vt]  4 --[longer,vt] 7;
    2 --[longer,vt]  6 --[longer,vt] 9;
    3 --[longer,vt]  6 --[longer,vt] 8;
  };
  
  \path (0mm,-30mm) node [inner sep=9mm,draw,
    fill=none,draw=none] (d7) {};
  \draw 
    (d7.90)  node (0) {} --
    (d7.39)  node (1) {} --
    (d7.347) node (2) {}
    (d7.296) node (3) {}
    (d7.244) node (4) {} --
    (d7.193) node (5) {} --
    (d7.141) node (6) {} --
    (0)
  ;
  \draw[vt] (0) -- (2) -- (3) -- (0);
  \draw[dashed] (1) -- (3) -- (4) -- (1);
  \draw ($(1)+(7mm,-7mm)$) 
    edge[->,bend left=51,vt] ($(2)+(2mm,-7mm)$);
    
  \path (43mm,-30mm) node [inner sep=9mm,draw,
    fill=none,draw=none] (d8) {};
  \draw 
    (d8.90)  node (0) {} --
    (d8.45)  node (1) {} --
    (d8.0)   node (2) {}
    (d8.315) node (3) {}
    (d8.270) node (4) {} --
    (d8.225) node (5) {} --
    (d8.180) node (6) {} --
    (d8.135) node (7) {} --
    (0)
  ;
  \draw[vt] (0) -- (2) -- (3) -- (0);
  \draw[dashed] (1) -- (3) -- (4) -- (1);
  \draw ($(1)+(7mm,-7mm)$) 
    edge[->,bend left=45,vt] ($(2)+(2mm,-7mm)$);
  \path[every node/.style={fill=none,draw=none}]
    (-22mm,-10mm) node {(a)}
     (22mm,-10mm) node {(b)}
     (65mm,-10mm) node {(c)}
    (-22mm,-40mm) node {(d)}
  ;
\end{tikzpicture}
\caption{(a) The Fano plane.
(b) The Möbius-Kantor configuration.
One 3-point ``line'' in (a) and in (b) is drawn as a circle.
(c) The Pappus configuration.
(d) Construction of the cyclic versions $D_7$ and~$D_8$
of the Fano and the Möbius-Kantor configurations.}
\label{fig:3conf}
\end{figure}
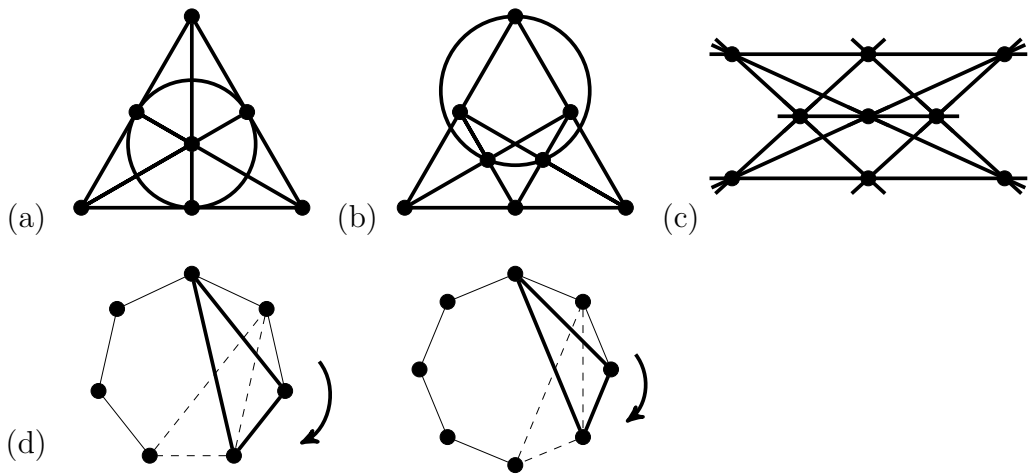

There is no $(n_3)$-configuration for $n\le6$.
There are a unique $(7_3)$-configuration, namely the
\emph{Fano plane}, and a unique $(8_3)$-configuration, namely the
\emph{Möbius-Kantor configuration}; see Figure \ref{fig:3conf}.
We denote the corresponding 3-harmonious coherent configurations by
$\ccc_\fa$ and $\ccc_\mk$ respectively.

\begin{theorem}\label{thm:3-reg}
$\ccc_\fa$ is non-separable, and $\ccc_\mk$ is separable.
\end{theorem}

In fact, we prove a more general fact.
Let $n\ge7$. The \emph{cyclic $(n_3)$-configuration} $D_n$ is 
constructed as follows \cite[Section 2.1]{Gruenbaum}. 
Let $F_n$ be the Cayley graph of $\integers_n$ with 
the difference set $\{\pm1,\pm2,\pm3\}$ and $D_n$ be the hypergraph 
formed by 3-cliques $\{i,i+2,i+3\}$ in $F_n$, where $i\in \integers_n$. 
It is straightforward to see that $D_n$ is really an $(n_3)$-configuration.
By the uniqueness of $(n_3)$-configurations for $n=7,8$ (see, e.g., \cite[Theorem 5.13]{PisanskiS}), 
the Fano plane is isomorphic, as a hypergraph, to $D_7$, and the Möbius-Kantor configuration
is isomorphic to $D_8$. Let $\ccc_n$ be the coherent configuration 
constructed from $D_n$ as in the proof of Part 3 of Lemma \ref{lem:equi-chain}.
This lemma implies that $\ccc_\fa\ciso\ccc_7$ and $\ccc_\mk\ciso\ccc_8$.
Thus, Theorem \ref{thm:3-reg} is equivalent to the statement that $\ccc_n$ is non-separable 
if $n=7$ and separable if $n=8$.

\begin{theorem}
Let $n\ge7$. The coherent configuration $\ccc_n$ is non-separable if and only if
$n$ is a multiple of~$7$.
\end{theorem}

\begin{proof}
Fix a vertex-colored graph $G_n=G(\ccc_n)$ as described above.
By Lemma \ref{lem:Gccc}, it suffices to show that $G_n$ has a 
non-trivial automorphism if and only if $n$ is divisible by~$7$.
Recall that the graph $G_n$ is not uniquely determined
(not even up to isomorphism). It can be constructed in many non-isomorphic ways, 
and any variant is suitable for our purposes
(as the automorphism group of $G_n$ always coincides with $\scac{\ccc_n}$, that is, is the same for any particular 
implementation of the construction). To fix the notation,
we turn back to construction of $\ccc_n$ from $D_n$ and make some specifications.

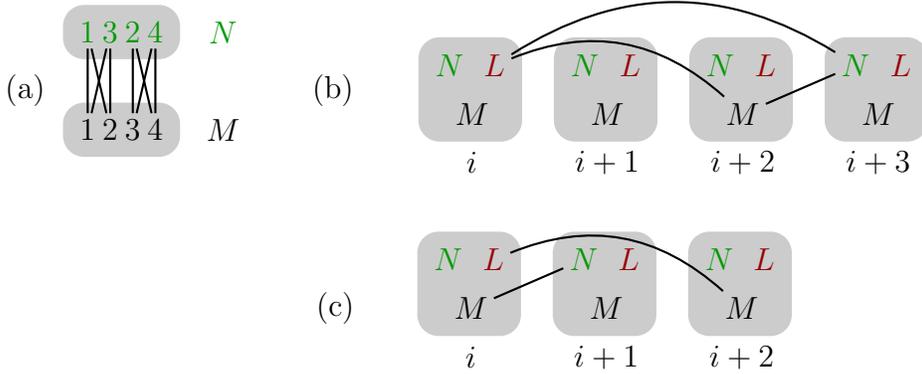
\begin{figure}
\hspace{5mm}(a)\hspace{1mm}
\begin{tikzpicture}[every node/.style={inner sep=2pt},
lab/.style={draw=none,fill=none,inner sep=0pt,rectangle},
thick,baseline=12pt,
xscale=.6,yscale=.65]
  \begin{scope}[xshift=0mm]
    \path
        (0  ,0)   node            (u0)  {1}
        (0.5,0)   node            (u1)  {2}
        (1  ,0)   node            (u2)  {3}
        (1.5,0)   node            (u3)  {4}
        (3  ,0)   node[lab]       (ul)  {$M$}
        (0  ,2)   node[Ngrn]      (o0)  {1}
        (0.5,2)   node[Ngrn]      (o2)  {3}
        (1  ,2)   node[Ngrn]      (o1)  {2}
        (1.5,2)   node[Ngrn]      (o3)  {4}
        (3  ,2)   node[lab]       (ol)  {$\cN$}
        (o0) edge (u0) edge (u1)
        (o2) edge (u0) edge (u1)
        (o1) edge (u2) edge (u3)
        (o3) edge (u2) edge (u3) ;
      \colclass{u1,u2,u3,u0}
      \colclass{o1,o2,o3,o0}
  \end{scope}
\end{tikzpicture}
\qquad(b)\hspace{1mm}
\begin{tikzpicture}[every node/.style={inner sep=2pt},
lab/.style={draw=none,fill=none,inner sep=0pt,rectangle},
thick,baseline=6pt, xscale=.6,yscale=.65]
      \path (-1.7,0)   node (ylab)[lab] {};
  \begin{scope}[xshift=-90mm]
    \path
        (9.5,-1)   node         (i)   {$i$}
        (9.0,1.0)  node[Ngrn]   (ig)  {$\cN$}
       (10.0,1.0)  node[Lred]   (ir)  {$\cL$}
        (9.5,0.0)  node         (ib)  {$M$}
       (18.5,-1)   node         (p3)  {$i+3$}
       (18.0,1.0)  node[Ngrn]   (p3g) {$\cN$}
       (19.0,1.0)  node[Lred]   (p3r) {$\cL$}
       (18.5,0.0)  node         (p3b) {$M$}
       (15.5,-1)   node         (p2)  {$i+2$}
       (15.0,1.0)  node[Ngrn]   (p2g) {$\cN$}
       (16.0,1.0)  node[Lred]   (p2r) {$\cL$}
       (15.5,0.0)  node         (p2b) {$M$}
       (12.5,-1)   node         (p1)  {$i+1$}
       (12.0,1.0)  node[Ngrn]   (p1g) {$\cN$}
       (13.0,1.0)  node[Lred]   (p1r) {$\cL$}
       (12.5,0.0)  node         (p1b) {$M$}
       (ir)  edge[bend left]   (p2b) edge[bend left]  (p3g) 
       (p2b)  edge              (p3g) ;
      \colclass{ib,ig,ir}
      \colclass{p3b,p3g,p3r}
      \colclass{p2b,p2g,p2r}
      \colclass{p1b,p1g,p1r}
  \end{scope}
\end{tikzpicture}\\[5mm]
\strut\hspace{45mm}
\begin{tikzpicture}[every node/.style={inner sep=2pt},
lab/.style={draw=none,fill=none,inner sep=0pt,rectangle},
thick, xscale=.6,yscale=.65]
      \path (7,0)   node (ylab)[lab] {(c)};
  \begin{scope}[xshift=5mm]
    \path
        (9.5,-1)   node         (i)   {$i$}
        (9.0,1.0)  node[Ngrn]   (ig)  {$\cN$}
       (10.0,1.0)  node[Lred]   (ir)  {$\cL$}
        (9.5,0.0)  node         (ib)  {$M$}
       (15.5,-1)   node         (p2)  {$i+2$}
       (15.0,1.0)  node[Ngrn]   (p2g) {$\cN$}
       (16.0,1.0)  node[Lred]   (p2r) {$\cL$}
       (15.5,0.0)  node         (p2b) {$M$}
       (12.5,-1)   node         (p1)  {$i+1$}
       (12.0,1.0)  node[Ngrn]   (p1g) {$\cN$}
       (13.0,1.0)  node[Lred]   (p1r) {$\cL$}
       (12.5,0.0)  node         (p1b) {$M$}
       (ir)  edge[bend left]   (p2b) 
       (ib)  edge              (p1g) ;
      \colclass{ib,ig,ir}
      \colclass{p2b,p2g,p2r}
      \colclass{p1b,p1g,p1r}
  \end{scope}
\end{tikzpicture}
\caption{
(a) The subgraph $G_n[X_{i+2},X_{i+3}]$ corresponds to the $2K_{2,2}$-interspace
$\ccc_n[X_{i+2},X_{i+3}]$, determining the matchings $M$ in $\ccc_n[X_{i+2}]$
and $\cN$ in $\ccc_n[X_{i+3}]$.
(b) The constraints on an automorphism $\phi$ of $G_n$ imposed by the triangle
$\{i,i+2,i+3\}$ in $D_n$.
(c) The two constraints on $\phi_i$, $\phi_{i+1}$, and $\phi_{i+2}$ involving~$\phi_i$.
}
\label{fig:automcolor}
\end{figure}

Specifically, we set
$V(\ccc_n)=\setdef{(i,j)}{i\in\integers_n,1\leq j\leq 4}$.
The fibers of $\ccc_n$ are the sets
$X_i=\Set{(i,1),(i,2),(i,3),(i,4)}$
for each $i\in\integers_n$ and, correspondingly, each
vertex from $X_i$
has color $i$ in $G_n$. 
From now on, we will just refer to $(i,j)$
as the vertex $j$ in color class
$X_i$, i.e., drop the $(i,)$ in most cases.
Recall that the construction ensures that $\ccc[X_i]\simeq F_4$.
Let $M,\cN,\cL$ be the three matching relations in~$\ccc[X_i]$.

For each $i\in\integers_n$, the triangle $\{i,i+2,i+3\}$ in $D_n$ contributes a triple of 
directly connected non-uniform interspaces in $\ccc_n$.
We construct $\ccc_n$ so that the non-uniform interspaces $\ccc_n[X_i,X_{i+2}]$,
$\ccc_n[X_i,X_{i+3}]$, and $\ccc_n[X_{i+2},X_{i+3}]$
determine the matchings $\cL$ in $\ccc_n[X_i]$, $M$ in $\ccc_n[X_{i+2}]$, 
and $\cN$ in $\ccc_n[X_{i+3}]$; see Figure \ref{fig:automcolor}(a).
This defines the coherent configuration $\ccc_n$ unambiguously; cf.\ Lemma \ref{lem:equi-chain}. 
Figure \ref{fig:automcolor}(b) shows the connection scheme of the above three interspaces,
which is the same for each $i\in\integers_n$.
In the graph $G_n=G(\ccc_n)$, there remain two possibilities for each of the subgraphs 
$G_n[X_i,X_{i+2}]$, $G_n[X_i,X_{i+3}]$, and $G_n[X_{i+2},X_{i+3}]$, but for the following
argument it does not matter which $2K_{2,2}$-fragment is in the graph 
and which is in its complement. We make an arbitrary choice in each case, and 
$G_n$ is therewith fixed.

Let $\phi$ be an automorphism of $G_n$, and let $\phi_i$ denote the 
restriction of $\phi$ to $X_i$. Since $\phi_i$ must map each matching of $X_i$ onto itself, 
this permutation belongs to the Klein four-group $K(X_i)$ consisting of the permutations
$\psi_0=\mathrm{id}_{X_i}$, $\psi_M=(12)(34)$, $\psi_N=(13)(24)$ and 
$\psi_L=(14)(23)$. Differently from the terminology after Lemma 
\ref{lem:sca}, we now name the non-identity group elements by the matching relation 
they fix, e.g., $\psi_M=\phi_{NL}$. 

The three connections between matchings shown in Figure \ref{fig:automcolor}(b)
can be seen as constraints on a sequence $(\phi_j)_{j\in\integers_n}$ corresponding to
an automorphism $\phi$. Indeed, if $\phi$ fixes (resp.\ flips) one of the three
connected matchings, then it must also fix (resp.\ flip) each of the other two.
For example, if $\phi_i=\psi_L$, then $\phi_{i+2}$ must fix the matching $M$ on $X_{i+2}$,
which implies that either $\phi_{i+2}=\psi_M$ or $\phi_{i+2}=\psi_0$, where $\psi_0=\mathrm{id}_{X_{i+2}}$.
Figure \ref{fig:automcolor}(c) shows the two constraints on $\phi_i$ and the two subsequent
local automorphisms $\phi_{i+1}$ and $\phi_{i+2}$. Alternatively these constraints can be
described by a table in Table \ref{tab:}(a). Consider, for example, the first
row of this table, which corresponds to the equality $\phi_i=\psi_M$. The constraint 
on the pair $(\phi_i,\phi_{i+1})$ forces $\phi_{i+1}$ to fix the matching $N$ and, therefore,
$\phi_{i+1}\in\{\psi_N,\psi_0\}$. Furthermore, since $\phi_i=\psi_M$ flips $L$,
the constraint on the pair $(\phi_i,\phi_{i+2})$ forces $\phi_{i+2}$ to flip $M$ and, therefore,
$\phi_{i+2}\in\{\psi_N,\psi_L\}$. Table \ref{tab:}(a) gives eight possibilities for the pair $(\phi_i,\phi_{i+1})$.
It turns out that, in each of these eight cases, Table \ref{tab:}(a) determines the next local automorphism
$\phi_{i+2}$ unambiguously. This can be seen from Table \ref{tab:}(b). For example, if $\phi_i=\psi_M$, 
then column $i+2$ of Table \ref{tab:}(a) shows that $\phi_{i+2}\in\{\psi_N,\psi_L\}$. If, moreover, $\phi_{i+1}=\psi_N$,
then the intersection of column $i+1$ and row $N$ of Table \ref{tab:}(a) shows that $\phi_{i+2}\in\{\psi_M,\psi_L\}$.
Therefore, the equalities $\phi_i=\psi_M$ and $\phi_{i+1}=\psi_N$ imply that $\phi_{i+2}=\psi_L$.

\begin{table}[h]
  \centering
(a)\hspace{1mm}
\begin{tabular}{|c|cc|}
\hline
 $i$ & $i+1$ & $i+2$\\\hline
 $M$ & $0/N$ & $N/L$\\
 $N$ & $M/L$ & $N/L$\\
 $L$ & $M/L$ & $0/M$\\
 $0$ & $0/N$ & $0/M$\\\hline
\end{tabular}
\qquad(b)\hspace{1mm}
\begin{tabular}{|c|c|r|l|c|}
\hline
 $i$ & $i+1$ & $i+2$ && $i+3$\\\hline
 $M$ & $N$   & $L$&$=\{N,L\}\cap\{M,L\}$ & $L$\\
 $M$ & $0$   & $N$&$=\{N,L\}\cap\{0,N\}$ & $M$\\
 $N$ & $M$   & $N$&$=\{N,L\}\cap\{0,N\}$ & $L$\\
 $N$ & $L$   & $L$&$=\{N,L\}\cap\{M,L\}$ & $M$\\
 $L$ & $M$   & $0$&$=\{0,M\}\cap\{0,N\}$ & $N$\\
 $L$ & $L$   & $M$&$=\{0,M\}\cap\{M,L\}$ & $0$\\
 $0$ & $N$   & $M$&$=\{0,M\}\cap\{M,L\}$ & $N$\\
 $0$ & $0$   & $0$&$=\{0,M\}\cap\{0,N\}$ & $0$\\
\hline
\end{tabular}
  \caption{(a) A table representation of the constraints on $\phi_i$ and $\phi_{i+1}$
and on $\phi_i$ and $\phi_{i+2}$ depicted in Figure \ref{fig:automcolor}(c).
(b) Extrapolation of the sequence $(\phi_i)_i$ on the basis of
$\phi_i$, $\phi_{i+1}$ and the recurrence relation implied by Table~(a).}
\label{tab:}
\end{table}

The rest of our analysis is based on Table \ref{tab:}(b).
Observe that the eight pairs in columns $i$ and $i+1$ are exactly the same as the
eight pairs in columns $i+1$ and $i+2$. It follows that
the constraints of Table (a) completely determine the entire sequence $(\phi_j)_{j}$
for each of the eight consistent pairs $(\phi_i,\phi_{i+1})$, for an arbitrarily fixed $i$.
We can imagine that the index $j$ ranges through the set of all integers,
remembering that it has to be considered modulo $n$.
Moreover, observe that the pair $(0,0)$ stays in the same row, while the other seven pairs
$(M,N),\ldots,(0,N)$ change their rows according to the cyclic permutation $(1372564)$.
This has the following consequences.
First, the pair $(M,N)$ eventually appears in every non-zero row and, hence, the
infinite expansions of the seven non-zero rows are identical up to a shift. 
Second, the sequence that appears in this way is periodic with period 7.
It follows that, if $G_n$  has a nontrivial automorphism $\phi$, then
the corresponding infinity sequence of local automorphisms $(\phi_j)_{j}$,
where each index $j$ is considered modulo $n$, must be periodic with period 7, namely
\begin{equation}
  \label{eq:phiseq}
  \begin{array}{ccccccccc}
\ldots&\phi_i&\phi_{i+1}&\phi_{i+2}&\phi_{i+3}&\phi_{i+4}&\phi_{i+5}&\phi_{i+6}& \ldots \\
\ldots&\psi_M&\psi_N&\psi_L&\psi_L&\psi_M&\psi_0&\psi_N&\ldots 
  \end{array}
\end{equation}
for some choice of $i$. We immediately conclude from here that, if $n$ is
not divisible by 7, then $G_n$  has no nontrivial automorphism.

Table \ref{tab:}(b) includes also column $i+3$.
Looking at columns $i$ and $i+3$, we see that, in each of the eight
possible cases, $\phi_{i+3}$ fixes $N$ whenever $\phi_i$ fixes $L$, and
$\phi_{i+3}$ flips $N$ whenever $\phi_i$ flips $L$.
This leads us to the following conclusion: If a sequence $(\phi_j)_{j}$
satisfies the constraints on $\phi_i$ and $\phi_{i+2}$ and $\phi_{i+2}$ and $\phi_{i+3}$,
shown in Figure \ref{fig:automcolor}(b), for every $i$, then it also satisfies 
the constraint on $\phi_i$ and $\phi_{i+3}$ for every $i$. It readily follows that,
if $n$ is divisible by 7, then any assignment of local automorphisms as in \refeq{phiseq}
determines a nontrivial automorphism $\phi$ of~$G_n$.
\end{proof}

\begin{remark}
  There are exactly three $(9_3)$-configurations \cite{Gruenbaum,PisanskiS}. 
The most famous of them is the \emph{Pappus configuration} shown in Figure~\ref{fig:3conf}(c). 
Computer-assisted verification shows that the corresponding 36-point coherent
configuration is non-separable. Of the other two $(9_3)$-configurations,
one is the cyclic $(9_3)$-configuration defined above, and the other is obtained
similarly by rotating the triangle $\{0,3,4\}$ (instead of $\{0,2,3\}$) in $\integers_9$. 
These two produce separable coherent configurations.
\end{remark}

\section{Irredundant configurations: The general case}\label{ss:gen-case}

Given $C\in\dcc\ccc$ and a non-empty $U\subsetneq C$, let $S(U,C)$ be the set of all edges $\{X,Y\}$ in 
$\fg\ccc$ such that $X\in U$ and $Y\in C\setminus U$. 
Using the notation $f_S$ introduced in Section \ref{ss:strict}, 
we now define $f_{X,C}=f_{S(\{X\},C)}$ for $X\in C$,.

\begin{lemma}\label{lem:gen}
Suppose that a coherent configuration \ccc is irredundant.
\begin{enumerate}[\bf 1.]
\item 
$f_S\in\saa\ccc$ if and only if, for every $C\in\dcc\ccc$, either the intersection
$S\cap{\binom C2}$ is empty or it forms a spanning bipartite subgraph of ${\binom C2}$,
where $\binom C2$ is considered the complete graph on the vertex set~$C$.
\item 
$\saa\ccc$ is generated by the set of $f_{X,C}$ for all $C\in\dcc\ccc$ and all $X\in C$.
\end{enumerate}
\end{lemma}

\begin{proof}
\textit{1.}  
For $C\in\dcc\ccc$, denote $S[C]=S\cap{\binom C2}$. By Lemma \ref{lem:dcc}, $\Set{S[C]}_{C\in\dcc\ccc}$
is a partition of $S$. Therefore, 
\begin{equation}
  \label{eq:SSC}
f_S=\prod_{C\in\dcc\ccc}f_{S[C]}.  
\end{equation}

  ($\Longleftarrow$)
It suffices to prove that each $f_{S[C]}$ is an algebraic automorphism of \ccc.
By Lemma \ref{lem:cl:b}, it is enough to check that, for every triple of fibers
$X,Y,Z$, the restriction of $f_{S[C]}$ to $\ccc[X\cup Y\cup Z]$ is an algebraic 
automorphism of $\ccc[X\cup Y\cup Z]$. If $|\{X,Y,Z\}\cap C|\le1$, then
$f_{S[C]}$ is the identity on $\ccc[X\cup Y\cup Z]$. 
If $|\{X,Y,Z\}\cap C|=2$, then Lemma \ref{lem:trans} implies that $\ccc[X\cup Y\cup Z]$
is either decomposable or skew-connected. The former case is obvious, and
in the latter case we are done by Part 1 of Lemma \ref{lem:skew+ddirect}. If
$\{X,Y,Z\}\subseteq C$, then the bipartiteness of $S[C]$ implies
that $f_{S[C]}$ switches either two (up to transposing) or no interspaces between $X,Y,Z$.
In this case we are done by Part 2 of Lemma \ref{lem:skew+ddirect}.

  ($\implies$)
Let $C\in\dcc\ccc$ and suppose that $S[C]$ is non-empty.
The claim is trivially true if $|C|=2$, so we assume that $|C|\ge3$.
Let $X$, $Y$, and $Z$ be three fibers in $C$. By assumption, 
the restriction of $f_S$ to $\ccc[X\cup Y\cup Z]$ is an algebraic 
automorphism of $\ccc[X\cup Y\cup Z]$. By Part 2 of Lemma \ref{lem:skew+ddirect},
$f_S$ makes either none or exactly two switches in $\ccc[X\cup Y\cup Z]$.
For $S[C]$, seen as a graph on the vertex set $C$, this implies that
$S[C]$ does not contain any induced subgraph isomorphic to $K_3$ or
to $K_2+K_1$, where the latter is the graph with 3 vertices and 1 edge.
A graph is $(K_2+K_1)$-free if and only if it is complete multipartite.
To see this, look at the complement and note that a graph is a vertex-disjoint
union of cliques if and only if it does not contain an induced copy of a path
on 3 vertices, the complement of $K_2+K_1$. Thus, $S[C]$ is a complete multipartite graph.
Since $S[C]$ is also triangle-free, it is bipartite.

\textit{2.}
By Equality \refeq{SSC}, Part 1 implies that $\saa\ccc$ is generated by the set of $f_{S(U,C)}$
for all $C\in\dcc\ccc$ and $\emptyset\ne U\subsetneq C$.
Note that, if $U$ is split into two non-empty parts $U_1$ and $U_2$,
then $f_{S(U,C)}=f_{S(U_1,C)}\circ f_{S(U_2,C)}$ (as each interspace between $U_1$ and $U_2$
is switched twice). It follows that
$$
f_{S(U,C)}=\prod_{X\in U}f_{X,C},
$$
which implies the lemma.
\end{proof}

Lemma \ref{lem:gen} suggests two approaches to deciding separability
of an irredundant configuration.

\paragraph{1st approach}
is based on Part 1 of Lemma \ref{lem:gen}. We infer from it that
\begin{equation}
  \label{eq:saaccc}
|\saa\ccc|=\prod_{C\in\dcc\ccc}2^{|C|-1}=2^{\of{\sum_{C\in\dcc\ccc}|C|}-|\dcc\ccc|}.
\end{equation}
It remains to compute the order of the group $\saai\ccc$ and check whether or not
$|\saai\ccc|=|\saa\ccc|$. By Lemma \ref{lem:quotient}, $|\saai\ccc|=|\sca\ccc|/|\scac\ccc|$,
where $|\sca\ccc|$ is easy to determine using Lemma \ref{lem:sca}.
Indeed, Lemma \ref{lem:sca} says that $\sca\ccc\cong\prod_{X\in F(\ccc)}\scac{\ccc[X]}$, and we have
$\scac{\ccc[X]}=K(X)$ (the Klein group of order 4) for $\ccc[X]\simeq F_4$,
$\scac{\ccc[X]}\cong\mathbb{D}_4$ (the dihedral group of order 8) for $\ccc[X]\simeq C_4$, and
$\scac{\ccc[X]}\cong\integers_4$ (the cyclic group of order 4) for $\ccc[X]\simeq\vec C_4$.
It remains to compute the order of the group of color-preserving
automorphisms $|\scac\ccc|$. To this end, we construct a  vertex-colored graph $G^*(\ccc)$
whose automorphism group $\mathrm{Aut}(G^*(\ccc))$ is precisely $\scac\ccc$, 
compute a set of generators of $\mathrm{Aut}(G^*(\ccc))$ as in \cite{ArvindK06},
and apply the Schreier–Sims algorithm to compute the order of $\mathrm{Aut}(G^*(\ccc))$
based on this set of generators.

We construct $G^*=G^*(\ccc)$ similarly to the graph $G(\ccc)$ in Section \ref{ss:3-reg} with the only 
difference that, for each $X\in F(\ccc)$, the subgraph $G^*[X]$ induced by $G^*$ on $X$ is defined more carefully:
\begin{itemize}
\item 
If there are interspaces $\ccc[Y,X]$ and $\ccc[Z,X]$ with askew connection at $X$,
then $G^*[X]$ is empty (in this case $\ccc[X]\simeq F_4$ by Part 2 of Lemma \ref{lem:X2Y4}, and
each matching relation on $X$ will be anyway preserved by any automorphism of~$G^*$);
\item 
Otherwise, $G^*[X]$ depends on $\ccc[X]$. We define $G^*[X]$ so that $\mathrm{Aut}(G^*[X])$
consists exactly of the color-preserving combinatorial automorphisms of $\ccc[X]$
(i.e., those mapping each basis relation of $\ccc[X]$ onto itself). Specifically,
\begin{itemize}
\item 
if $\ccc[X]\simeq F_4$, then we put a matching $2K_2$ in $G^*[X]$ different from the one
determined by some interspace $\ccc[Y,X]$ (at least one such an interspace
must exist because \ccc is indecomposable);
\item 
if $\ccc[X]\simeq C_4$, then we leave $G^*[X]$ empty (a matching on $X$ is implicitly
determined anyway);
\item 
if $\ccc[X]\simeq\vec C_4$, we have to put a directed 4-cycle in $G^*[X]$ coherently
with the matching implicitly determined on $X$. To avoid making $G^*$ 
a directed graph, we subdivide each edge of this cycle with two differently colored
vertices in the direction given by $\vec C_4$. This costs us two new colors
and four new vertices of each of these colors (which we put in $V(G^*)$ in addition
to the vertices of~\ccc).
\end{itemize}
\end{itemize}

\paragraph{2nd approach}
is based on Part 2 of Lemma \ref{lem:gen}.
For each pair $(X,C)$ where $X\in C\in\dcc\ccc$, we check whether the algebraic 
automorphism $f_{X,C}$ is induced by a combinatorial automorphism. 
A crucial fact is that the number of such pairs is polynomially bounded.
Fix $G^*=G^*(\ccc)$ as above and obtain a graph $G^*_{X,C}$ from $G^*$ by complementing
each bipartite subgraph $G^*[X,Y]$ spanned by the fiber $X$ and a fiber $Y$ in $C\setminus\{X\}$.
By construction, a combinatorial automorphism $\phi$ of \ccc induces $f_{X,C}$
exactly when $\phi$ is an isomorphism of the graphs $G^*$ and $G^*_{X,C}$.
Thus, $f_{X,C}$ is induced by a combinatorial automorphism if and only if $G^*\cong G^*_{X,C}$.
The last condition is efficiently verifiable \cite{ArvindK06} as
the graphs $G^*$ and $G^*_{X,C}$ are of color multiplicity~4.

\begin{remark}\label{rem:GG}
Following the second approach, instead of $G^*=G^*(\ccc)$ we can still use the simpler
construction $G=G(\ccc)$ exactly as described in Section \ref{ss:3-reg}
(where $G[Z]$ for each $Z\in F(\ccc)$ is an independent set).
Though the automorphism group $\mathrm{Aut}(G(\ccc))$ can be strictly larger than $\scac\ccc$,
it is easy to see that $G\cong G_{X,C}$ exactly when $G^*\cong G^*_{X,C}$.
Indeed, any isomorphism $\phi$ from $G^*$ to $G^*_{X,C}$ is obviously an isomorphism
also from $G$ to $G_{X,C}$. Conversely, let $\phi$ be an isomorphism from $G$ to $G_{X,C}$.
Suppose that $\ccc[Z]$ is a cell with a single determined matching $M$. Any modification $\phi^*$ of $\phi$ within $Z$
which maps $M$ onto itself and flips $M$ if and only if $\phi$ does so
stays an isomorphism from $G$ to $G_{X,C}$. It follows that $\phi$ admits
a modification $\phi^*$ such that $\phi^*$ is not only an isomorphism from $G$ to $G_{X,C}$
but also an automorphism of $G^*[Z]$. Making such a modification on each such fiber $Z$,
we obtain an isomorphism $\phi^*$ from $G^*$ to $G^*_{X,C}$.
Summarizing, we see that our decision procedure has the same outcome
regardless of whether the simpler construction $G(\ccc)$ or its augmented version $G^*(\ccc)$ is used.
\end{remark}

\begin{example}\label{ex:1}
We make use of the construction of a coherent configuration $\ccc=\ccc(D)$
based on a given partial linear space $D$ as described in the proof of
Part 3 of Lemma \ref{lem:equi-chain}. For the partial linear space $D$
depicted in Figure \ref{fig:PLS}(a), the coherent configuration $\ccc(D)$
is separable. By Part 2 of Lemma \ref{lem:gen}, it is enough to show that
each $f_{X,C}$ is induced by a combinatorial automorphism. Suppose first that
$|C|=3$, say, $C=\{X,Y,Z\}$. For appropriate combinatorial automorphisms
$\phi\in K(Y)$ and $\psi\in K(Z)$, we have $\phi\,\psi f_{X,C}=f_S$ where $S$
consists of two edges of $\fg\ccc$ emanating from $Y$ and $Z$ such that
each of them forms a 2-clique in $D=\dcc\ccc$. Note that the six edges
of this kind form a connected subgraph of $\fg\ccc$. Like the analysis
of the CFI case in Section \ref{ss:CFI}, we see that $f_S$ is induced by
a combinatorial automorphism. Since $f_{X,C}=\psi^{-1}\phi^{-1}f_S$, the same
holds true for~$f_{X,C}$.

Suppose now that $|C|=2$, say, $C=\{X,Y\}$. Let $C'$ be the hyperedge of $D$
such that $X\in C'$ and $|C'|=3$. For a suitable combinatorial automorphisms
$\phi\in K(X)$, we have $\phi\,f_{X,C}=f_{X,C'}$. We already know that $f_{X,C'}$
is induced by a combinatorial automorphism. Therefore, this is so also
for $f_{X,C}=\phi^{-1}f_{X,C'}$.

Consider the same example also from the perspective of Part 1 of Lemma \ref{lem:gen}.
By Equality \refeq{saaccc}, we have $|\saa\ccc|=2^{10}$. As it readily follows from 
Lemma \ref{lem:sca}, $|\sca\ccc|=2^{12}$. We already know that $|\saai\ccc|=|\saa\ccc|$,
and Lemma \ref{lem:quotient} implies that $|\scac\ccc|=4$. 
This can be seen also directly. Moreover, $\scac\ccc$ is isomorphic to the Klein four-group.
Indeed, let $X_1,\ldots,X_6$ be the fibers of \ccc.
If $\phi$ is a color-preserving automorphism of \ccc,
then $\phi=\prod_{i=1}^6\phi_i$, where each $\phi_i\in K(X_i)$ is extended by identity
outside $X_i$. For every choice of a non-identity permutation $\phi_1$,
a simple argument shows that each of the other factors $\phi_2,\ldots,\phi_6$
is uniquely determined.
\end{example}

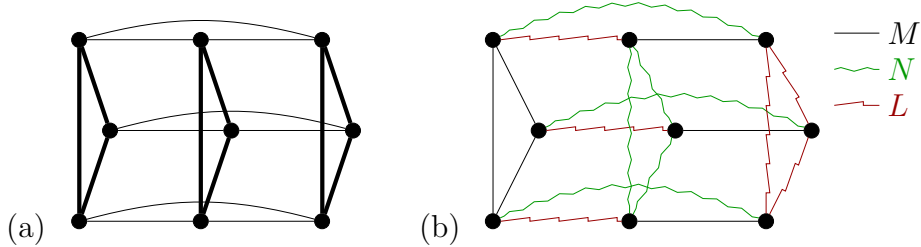
\begin{figure}
  \centering
\begin{tikzpicture}[every node/.style={circle,draw,inner sep=2pt},
lab/.style={draw=none,fill=none,inner sep=0pt,rectangle},
line width=0.25pt,
vt/.style={line width=1.6pt}]
\matrixgraph[name=m1,nolabel]{
  &[2mm]&[1cm] &[2mm]&[1cm] &[2mm]&\\
  1 &   & 4 &   & 7 &  \\[1cm]
    & 2 &   & 5 &   & 8\\[1cm]
  3 &   & 6 &   & 9 &  \\
  }{
  {1,2,3} --[matching] {4,5,6} --[matching] {7,8,9} 
  --[matching,bend right=15] {1,2,3};
  1--[vt] {2,3}; 2 --[vt] 3;
  4--[vt] {5,6}; 5 --[vt] 6;
  7--[vt] {8,9}; 8 --[vt] 9;
}
\matrixgraph[name=m2,nolabel,matrix anchor=west,
  at={($(m1.west)+(55mm,0mm)$)}]{
  &[4mm]&[1cm] &[4mm]&[1cm] &[4mm]&\\
  1 &   & 4 &   & 7 &  \\[1cm]
    & 2 &   & 5 &   & 8\\[1cm]
  3 &   & 6 &   & 9 &  \\
  }{
  {1,2,3} --[matching,Lred,sw] {4,5,6} --[matching] {7,8,9} 
  --[matching,bend right=25,Ngrn,zz] {1,2,3};
  1-- {2,3}; 2 -- 3;
  4--[Ngrn,zz] {5,6}; 5 --[Ngrn,zz] 6;
  7--[Lred,sw] {8,9}; 8 --[Lred,sw] 9;
}
\legend[at={($(m2.north east)+(8mm,-5mm)$)}]{
\legendrow{line width=0.25pt}{$M$}
\legendrow{line width=0.25pt,Ngrn,zz}{$\cN$}
\legendrow{line width=0.25pt,Lred,sw}{$\cL$}
}
\path 
  ($(m1.south west)-(6mm,0mm)$) node[lab] {(a)}
  ($(m2.south west)-(6mm,0mm)$) node[lab] {(b)}
;
\end{tikzpicture}
\caption{
(a) A pattern $D$ is represented as a clique partition of
a 9-vertex graph consisting of three 3-cliques and nine 2-cliques.
(b) We can assign the names $M$, $\cN$ and $\cL$ to the matching 
basis relations of each cell in $\ccc(D)$ such that every
interspace connects matchings with the same name. Each edge color
in the depicted graph represents this name.
}
\label{fig:mixed}
\end{figure}

\begin{example}\label{ex:2}
Consider next $D$ shown in Figure \ref{fig:mixed}(a). This pattern
yields a non-separable coherent configuration $\ccc=\ccc(D)$. To see this,
we make use of the names $M$, $\cN$ and $\cL$ for the matching
basis relations in each cell as introduced after Lemma 
\ref{lem:sca} and assign these names as shown in Figure \ref{fig:mixed}(b).
Note that the nine 2-cliques in $D$ form three disjoint 3-cycles, each
with exactly one edge of the type $M$, $\cN$ and $\cL$.
Combining the local combinatorial automorphisms
$\phi_{NL}$, $\phi_{ML}$ and $\phi_{MN}$ of the three cells in such a cycle
(e.g., for the cells in the top row of Figure \ref{fig:mixed}(b))
we get a nontrivial combinatorial automorphism that induces 
a trivial algebraic automorphism, i.e., a  nontrivial
color-preserving combinatorial automorphism.
By doing this on any one, two or three multicolored cycles, we obtain $2^3-1=7$
nontrivial color-preserving automorphisms.

Three further such automorphisms can be constructed 
from the monochromatic 3-cliques. Consider, for instance, 
the $M$- and the $N$-cliques. They are connected by a matching
consisting of three $L$-edges. If we pick
$\phi_{ML}$ for each cell of the $M$-clique
and $\phi_{NL}$ for each cell of the $N$-clique,
we obtain a nontrivial color-preserving automorphism.
The other two pairs of monochromatic 3-cliques give us two
more such automorphisms. 
It follows that $|\scac{\ccc}|\geq 10$.
By Lemma \ref{lem:sca}, $|\sca{\ccc}|=4^9=2^{18}$.
Lemma \ref{lem:quotient}, therefore, implies that
$|\saai{\ccc}|\le 2^{18}/10<2^{15}$.
On the other hand, by \refeq{saaccc}, we have
$|\saa{\ccc}|=2^{3\cdot3+9\cdot2-12}=2^{15}$.
This implies that not all strict algebraic automorphisms of $\ccc$
are induced by combinatorial automorphisms and, thus, $\ccc$ is non-separable.
\end{example}

\section{Putting it together}\label{s:putting}

We are now prepared to prove Theorem \ref{thm:sep4}. The cut-down lemmas 
(that is, Lemmas \ref{lem:excl-matching-both}, \ref{lem:excl2points}, and \ref{lem:excl-C8})
and our analysis of the irredundant case in Section \ref{s:2K22} yield the following algorithm
for recognizing whether or not a given coherent configuration \ccc
with fibers of size at most 4 is separable.

\begin{itemize}
\item 
Decompose \ccc in the direct sum of indecomposable subconfigurations $\ccc_1,\ldots,\ccc_m$
and handle each of them separately. By Lemma \ref{lem:direct}, \ccc is separable
if and only if every $\ccc_i$ is separable.
\item 
Assume, therefore, that the input configuration \ccc is indecomposable.
If all fibers of \ccc are of size at most 3, immediately decide that \ccc is separable
(see Corollary \ref{cor:size3}). Otherwise:
\begin{itemize}
\item 
Remove all fibers of size 2 from~\ccc.
\item 
Remove all pairs of fibers $X$ and $Y$ with $\ccc[X,Y]\simeq C_8$.
\item 
As long as \ccc contains an interspace $\ccc[X,Y]$ with a matching, 
remove the fiber $X$ from~\ccc.
\end{itemize}
\item 
If \ccc becomes decomposable, split it into indecomposable components
and handle each of them separately.
\item 
If \ccc 
becomes empty, decide that \ccc is separable.
\item 
Otherwise, we arrive at the case that \ccc is irredundant and proceed as 
described in Section~\ref{ss:gen-case}.
\item 
If all computational paths terminate with a positive decision, output `\emph{\ccc is separable}';
otherwise, output `\emph{\ccc is non-separable}'.
\end{itemize}

Due to \cite{ArvindK06}, each computational path for an irredundant coherent configuration 
is implementable in $\oplus\mathrm{L}$.
A list of all subconfigurations to which this step is applied can clearly be generated
in logarithmic space \cite{Reingold08}. 
Since $\mathrm{L}^{\oplus\mathrm{L}}=\oplus\mathrm{L}$ (see \cite{BuntrockDHM92}),
the whole algorithm can be implemented in~$\oplus\mathrm{L}$.
Theorem \ref{thm:sep4} is proved.

As a by-product of our analysis, we state the following fact.

\begin{theorem}\label{thm:cc16}\hfill
  \begin{enumerate}[\bf 1.]
  \item 
All coherent configurations with 15 or fewer points and maximum fiber size 4 are separable. 
\item 
There is a unique, up to combinatorial isomorphism, non-separable coherent configuration
on 16 points with maximum fiber size~4.  
  \end{enumerate}
\end{theorem}

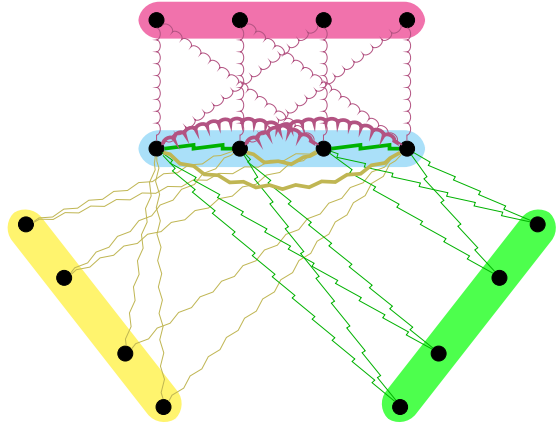
\begin{figure}
\centering
\begin{tikzpicture}[every node/.style={circle,draw,black,
  inner sep=2pt,fill=black},
  elabel/.style={black,draw=none,fill=none,rectangle},
  every edge/.append style={every node/.append style={elabel}},
  lab/.style={draw=none,fill=none,inner sep=0pt,rectangle},
  be/.style={label position=below,label distance=2mm},
  br/.style={label position=below right},
  ri/.style={label position=right},
  le/.style={label position=left},
  nl/.style={nolabel},
  line width=0.3pt,
  vt/.style={line width=1.4pt},
  zcol/.style={yellow!70!black,-,zz},
  ycol/.style={magenta!70!black,-,bps},
  ucol/.style={green!70!black,-,sw},
]
  \matrixgraph[name=m1,nolabel]
    {&[3mm]&[6mm]&[3mm]&[-3mm]&[9mm]&[9mm]&[9mm]
    &[-3mm]&[3mm]&[6mm]&[3mm]\\
    &&&&    y_1&y_2&y_3&y_4        \\[15mm]
    &&&&    x_1&x_2&x_3&x_4        \\[8mm]
     z_4[le] &&&&&&  &&&&& u_4[ri]\\[5mm]
    & z_3[le] &&&& & &&&& u_3[ri]\\[8mm]
    && z_2[le] &&&&  &&& u_2[ri]\\[5mm]
    &&& z_1[le] && & && u_1[ri]\\
  }{
    {z_1,z_2} --[zcol] {x_1,x_4};
    {z_3,z_4} --[zcol] {x_2,x_3};
    {u_1,u_2} --[ucol] {x_1,x_2};
    {u_3,u_4} --[ucol] {x_3,x_4};
    {y_1,y_3} --[ycol] {x_1,x_3};
    {y_2,y_4} --[ycol] {x_2,x_4};
    {x_1,x_3} --[matching,vt,ucol] {x_2,x_4};
    {x_1,x_2} --[matching,vt,ycol,bend left=30]
      {x_3,x_4};
    {x_1,x_2} --[matching,vt,zcol,bend right=30]
      {x_4,x_3};
    
  };
  \colclass[0,fill=magenta!70]{y_1,y_2,y_3,y_4}
  \colclass[0,fill=cyan!30]{x_1,x_2,x_3,x_4}
  \colclass[-52,fill=yellow!70]{z_1,z_2,z_3,z_4}
  \colclass[52,fill=green!70]{u_1,u_2,u_3,u_4}
\end{tikzpicture}
\caption{A fragment of the unique non-separable coherent configuration $\cct$ with 16 points:
Three pairwise skew-connected interspaces and the matching basis relations they induce.}
\label{fig:3Intsp3Match}
\end{figure}

\begin{proof}
\textit{1.}  
Suppose that a coherent configuration has at most 15 vertices. 
After cutting down 2-point cells, matching interspaces, and $C_8$-interspaces
and ignoring possible single-fiber components, we are faced with an
irredundant configuration \ccc having 2 or 3 fibers.
It follows from Lemma \ref{lem:trans} that $\ccc$ is either skew-connected
or has 3 fibers with all connections between non-uniform fibers
being direct. In the former case, since we obviously have $\delta(\fg\ccc)\le2$,
the coherent configuration \ccc is separable by Corollary \ref{cor:skew}. In the latter case,
the separability of \ccc follows from Part 2 of Lemma \ref{lem:gen}.
Indeed, $\dcc\ccc$ consists of a single 3-element hyperedge $C$,
and we only have to check that $f_{X,C}$ for any $X\in C$
is induced by a combinatorial automorphism $\phi$.
Let $M$ be the matching basis relation in $\ccc[X]$ determined by
the interspaces between $X$ and the other two cells in $C$.
As a desired $\phi$, we can take any color-preserving automorphism of $\ccc[X]$
flipping $M$ and extend it to $V(\ccc)$ by identity. 
Such an automorphism exists in all three cases $\ccc[X]\simeq F_4,C_4,\vec C_4$
(cf.\ the proof of Part 2 of Lemma~\ref{lem:CFI}).

\textit{2.}
Taking into account the proof of Part 1, we only have to consider the case that an
irredundant configuration \ccc has 4 fibers.
If $\dcc\ccc$ consists of a single 4-element hyperedge, that is, all interspaces are
non-uniform and all connections between them are direct, then the separability of \ccc
follows by Part 2 of Lemma \ref{lem:gen} as in Part 1.

Suppose now that $\dcc\ccc$ has a hyperedge $C$ of size 3.
To show that \ccc is separable, we again use Part 2 of Lemma \ref{lem:gen}.
Let $X\in C$. Note that the degree of $X$ in the hypergraph $\dcc\ccc$ is at most 2.
Therefore, the same argument as in Part 1 works, showing that $f_{X,C}$ is induced by a combinatorial 
automorphism of \ccc. Similarly, $f_{X,C'}$ is induced by a combinatorial 
automorphism if $X$ belongs also to a 2-element hyperedge $C'$ of $\dcc\ccc$.
More specifically, in this case we have $C'=\{X,Y\}$ where $Y$ is the fiber of \ccc not belonging to $C$.
Thus, the interspace $\ccc[X,Y]$ is non-uniform, and $\ccc[X]\simeq F_4$
with one matching $N$ determined by $\ccc[Y,X]$ and another matching $M$ determined by the interspaces
between $X$ and the other two fibers in $C$. Then  $f_{X,C'}$ is induced by
$\phi_{NL}\in K(X)$ where $L$ is the other matching in $\ccc[X]$ different from $N$ and~$M$.

If all hyperedges of $\dcc\ccc$ are of size 2, then \ccc is skew-connected.
By Corollary \ref{cor:skew}, \ccc is separable exactly when $\delta(\fg\ccc)\le2$.
It remains to note that $\delta(\fg\ccc)=3$ in the only case that $\fg\ccc$ is the complete
graph on 4 vertices; see Figure~\ref{fig:3Intsp3Match}. Any two skew-connected coherent configurations with such
fiber graph are combinatorially isomorphic, as easily follows from Part 2 of
Lemma~\ref{lem:matchmatch}. 
\end{proof}

\section{Back to graphs}\label{s:back}

\subsection{Proof of Theorem \ref{thm:amen4}}

Let $G$ be a colored graph as defined in Section \ref{ss:graphs}.
Suppose that the color multiplicity of $G$ is bounded by 4.
By Theorem \ref{thm:reduction}, $G$ is amenable to \wl if and only if its
coherent closure $\ccg G$ is separable. Given $G$ with $n$ vertices, the 
coherent closure $\ccg G$ is computable in time $O(n^3\log n)$ using
the algorithm in \cite{ImmermanL90}.
Since $G$ has color multiplicity at most 4, the coherent configuration $\ccg G$ 
has only fibers with at most 4 points. Therefore, we can decide separability of $\ccg G$
using the algorithm presented in Section~\ref{s:putting}.
This algorithm reduces deciding separability for $\ccg G$ to deciding separability
for a number of irredundant subconfigurations $\ccc_1,\ldots,\ccc_t$ such that
\begin{equation}
  \label{eq:nnn}
\bigcup_{i=1}^t F(\ccc_i)\subseteq F(\ccg G).   
\end{equation}
Producing the list of
coherent configurations $\ccc_1,\ldots,\ccc_t$ has low time complexity.
For each $i\le t$, we decide separability of $\ccc_i$ using the 2nd Approach presented
in Section~\ref{ss:gen-case}. Specifically, $\ccc_i$ is separable if and only if
the associated vertex-colored graph $G^i=G(\ccc_i)$ is isomorphic to its modified version $H^i=G^i_{X,C}$
for every $X\in F(\ccc_i)$, where $C$ is the hyperedge of $\dcc{\ccc_i}$ containing $X$.
Here, $G(\ccc)$ refers to the construction of a graph from a given irredundant configuration
described in Section \ref{ss:3-reg}; see Remark \ref{rem:GG}.
Denote the number of vertices in $G^i$ by $n_i$.
The isomorphism algorithm for graphs of color multiplicity 4 in \cite{ArvindK06}
performs a low-cost conversion of the pair $(G^i,H^i)$ into a system of $M_i<(n_i)^2$
linear equations with $N_i<n_i$ unknowns over the field $\integers_2$
such that $G^i\cong H^i$ if and only if the system is consistent.

Specifically, we here describe a simplified version of this general reduction
suitable for any pair $(G^i,H^i)$ arising from $\ccc_i$. Recall that 
$V(G^i)=V(H^i)=V(\ccc_i)$, and the vertex color classes of both
$G^i$ and $H^i$ are exactly the fibers $X_1,\ldots,X_s$ of $\ccc_i$,
where each $X_j$ has the same color both in $G^i$ and $H^i$.
For every vertex color class $X_j$, we have $G^i[X_j]=H^i[X_j]$.
Every non-empty bipartite subgraph $G^i[X_j,X_k]$ is isomorphic to $2K_{2,2}$.
In this case, $H^i[X_j,X_k]$ is equal either to $G^i[X_j,X_k]$ or to its bipartite complement.

Any isomorphism from $G^i$ and $H^i$ maps each vertex color class $X_j$ onto itself.
Moreover, if $G^i$ and $H^i$ are isomorphic, then there is an isomorphism $\phi$
preserving each of the three matchings on $X_j$ for every $j$
(note that $\phi$ is forced to preserve the matchings if $\ccc_i[X_j]$
has at least two determined matchings and $\phi$ can be modified to obey this condition if 
there is exactly one determined matching in $\ccc_i[X_j]$).
Denote the restriction of $\phi$ to $X_j$ by $\phi_j$. Thus, $\phi_j$
is one of the four elements of the Klein group $K(X_j)$.
Recall that a preserved matching can be either fixed or flipped.
Denote the matchings on $X_j$ by $A_j,B_j,C_j$. An element of $K(X_j)$
is uniquely determined by a triple $(a_j,b_j,c_j)$, where $a_j=1$ if $A_j$ is flipped
and $a_j=0$ if $A_j$ is fixed, and similarly for $b_j$ and $c_j$.
Since a non-identity element of $K(X)$ fixes one matching and flips
the other two, we have
$$
a_j\oplus b_j\oplus c_j=0.\leqno (E_j)
$$
Another constraint on $\phi_j$ is imposed by each pair $X_j,X_k$ such that $G^i[X_j,X_k]$ is
non-empty. To be specific, suppose that $\ccc_i[X_j,X_k]$ determines the matching $A_j$ in $X_j$
and the matching $B_k$ in $X_k$. Then
$$
a_j\oplus b_k=d_{j,k},\leqno (E_{j,k})
$$
where $d_{j,k}=0$ if $H^i[X_j,X_k]$ is equal to $G^i[X_j,X_k]$ and
$d_{j,k}=1$ if $H^i[X_j,X_k]$ is the bipartite complement of $G^i[X_j,X_k]$.
It remains to notice that a set of permutations $\Set{\phi_j}_{j=1}^s$
composing an isomorphism from $G^i$ to $H^i$ exists if and only if
the system of equations consisting of ($E_j$) for all $j\le s$ and
($E_{j,k}$) for all non-empty $G^i[X_j,X_k]$ has a solution.

The rank of an $M\times N$ matrix over a finite field is computable in time $O(MN^{\omega-1})$,
where $N\le M$ (see \cite{BunchH74,IbarraMH82}), or in randomized time $O(MN\log N+N^\omega)$ (see \cite{CheungKL13}).
Recall that we test isomorphism of $|F(\ccc_i)|$ pairs of graphs $G^i$ and $H^i$ and that,
for each pair, our task is reduced to checking solvability of a linear system
with $3|F(\ccc_i)|$ unknowns and at most $\binom{|F(\ccc_i)|}2$ equations.
Since $|F(\ccc_i)|=n_i/4$, in this way we can test separability of $\ccc_i$ in time $O((n_i)^{2+\omega})$
deterministically or in time $O((n_i)^{4}\log n_i)$ using randomization.
Taking into account the inequality $\sum_{i=1}^t n_i\le n$, which follows from \refeq{nnn},
and the general inequality 
$$
\sum_{i=1}^t (n_i)^\alpha \le \of{\sum_{i=1}^t n_i}^\alpha
$$ 
for any real $\alpha\ge1$, we conclude that separability of $\ccg G$ is decidable in 
deterministic time $O(n^{2+\omega})$ or in randomized time~$O(n^{4}\log^2 n)$,
where an extra logarithmic factor corresponds to the number of repetitions
needed to make the failure probability an arbitrarily small constant.

\subsection{Small graphs}

\begin{theorem}\label{thm:graphs16}\hfill
  \begin{enumerate}[\bf 1.]
  \item 
All graphs of color multiplicity 4 with at most 15 vertices are amenable.
\item 
Up to isomorphism and color renaming, 
there are 436 non-amenable graphs of color multiplicity 4 with 16 vertices.
More precisely, the number of non-trivial $\eqq$-equivalence classes is 218, 
each consisting of exactly two non-isomorphic graphs.
  \end{enumerate}
\end{theorem}

  \begin{proof}
    \textit{1.}
Let $G$ be a graph of color multiplicity 4 with at most 15 vertices.
By Theorem \ref{thm:reduction}, $G$ is amenable if and only if its
coherent closure $\ccg G$ is separable. Note that $\ccg G$ has at most
15 points, and every fiber of $\ccg G$ has size at most 4.
By Part 1 of Theorem \ref{thm:cc16}, $\ccg G$ is separable.

    \textit{2.}
Recall the notation introduced in Subsection \ref{ss:ccs} and the statement of
Lemma \ref{lem:aiso-eqq}. Given a colored graph $G$, let $\calR_G$ denote its underlying
rainbow, that is, the partition of $V(G)^2$ determined by the color classes of $G$.
In particular, if $G$ is a vertex-colored graph, then $\calR_G$ consists of the sets
of loops $vv$ of equally colored vertices, the set of pairs $uv$ with adjacent $u$ and $v$,
and the set of pairs $uv$ with non-adjacent $u$ and $v$. The coherent closure
$\ccc(G)=\ccc(\calR_G)$ is a refinement of the partition $\calR_G$.
Given a bijection $f\function{\ccc(G)}\ccd$ from $\ccc(G)$ onto a rainbow $\ccd$,
we extend $f$ to a bijection from $\ccc(G)^\cup$ onto $\ccd^\cup$ by the rule
$(X_1\cup\ldots\cup X_s)^f=X_1^f\cup\ldots\cup X_s^f$. For each $X\in\calR_G$,
this defines its image $X^f$ and, as usually, we have $(\calR_G)^f=\setdef{X^f}{X\in\calR_G}$,
which is a partition coarser than \ccd. Finally, we define $G^f$ as the colored version of $(\calR_G)^f$
where each $X^f$ inherits the color of the color class $X$ of $G$. If $G$ is a vertex-colored graph,
then $G^f$ is a vertex-colored graph as well.

We begin with stating a general fact that follows from Part 2 of Lemma \ref{lem:aiso-eqq}
and the discussion preceding this lemma.

\begin{claim}\label{cl:1}
$G\eqq H$ if and only if there is an algebraic isomorphism $f\function{\ccc(G)}{\ccc(H)}$
such that $H=G^f$.
\end{claim}

According to Part 2 of Lemma \ref{thm:cc16}, among all 16-point coherent configurations 
with maximum fiber size 4 there is a unique non-separable configuration \cct.
Moreover, the proof of this lemma shows that \cct is the unique skew-connected
configuration whose fiber graph $\fg\cct$ is isomorphic to the complete graph $K_4$.
Let $G$ be a vertex-colored graph of color multiplicity 4 on 16 vertices.
By Theorem \ref{thm:reduction}, $G$ is non-amenable if and only if $\ccg G\ciso\cct$.
Looking for non-amenable graphs, we can therefore assume that $\ccg G=\cct$.

Combining Claim \ref{cl:1} with Lemma \ref{lem:alg-comb}, we conclude that
$G\eqq H$ if and only if there is a strict algebraic automorphism $f$ of \cct
such that $H\cong G^f$. Looking for graphs $H$ such that $G\eqq H$, we can therefore 
assume that $H=G^f$ for $f\in\saa{\cct}$. If  $f\in\saai{\cct}$, i.e., $f$ is induced by 
a combinatorial automorphism $\phi$ of \cct, then obviously $\phi$ is an isomorphism from $G$
to $H$. Conversely, if $G\cong H$, then $f\in\saai{\cct}$ by Part 3 of Lemma \ref{lem:aiso-eqq}.
Thus,
\begin{equation}
  \label{eq:eqq-cong}
G\eqq G^f\text{ and }G\not\cong G^f\text{ iff }f\in\saa{\cct}\setminus\saai{\cct}.  
\end{equation}
Note that the difference $\saa{\cct}\setminus\saai{\cct}$ is non-empty by Part 3 
of Lemma \ref{lem:CFI} because \cct is skew-connected and $\delta(\fg\cct)=3$.
We now generalize the equivalence \refeq{eqq-cong} as follows.

\begin{claim}\label{cl:2}
Suppose that $\ccg G=\cct$. Let $h,f\in\saa{\cct}$. Then $G^h\cong G^f$ if
and only if $h^{-1}f\in\saai{\cct}$.
\end{claim}

\begin{subproof}
Denote $A=G^h$ and note that $G^f=A^{h^{-1}f}$. We obtain the claim by applying
the equivalence \refeq{eqq-cong} to~$A$.
\end{subproof}

Consider $h,f\in\saa{\cct}\setminus\saai{\cct}$. By Part 3 of Lemma \ref{lem:CFI},
$\saai{\cct}$ is a subgroup of $\saa{\cct}$ of index 2. This implies that
$h^{-1}f\in\saai{\cct}$ and, hence, $G^h\cong G^f$ by Claim \ref{cl:2}.
It follows that, if each of two graphs $H_1$ and $H_2$ is $\eqq$-equivalent
but not isomorphic to $G$, then $H_1$ and $H_2$ are isomorphic to each other.
This proves that every non-trivial $\eqq$-equivalence class of 16-vertex graphs
of color multiplicity 4 contains exactly two graphs.
It remains to count the number of such classes.

We first describe the structure of graphs $G$ of color multiplicity 4 with $\ccg G=\cct$.
Clearly, any such $G$ has 4 vertex color classes, each consisting of 4 vertices.
More precisely, the vertex color classes of $G$ are exactly the fibers of \cct.
For any two fibers $X$ and $Y$, the interspace $\cct[X,Y]$ is a refinement of the
partition of $X\times Y$ accordingly to the adjacency relation of $G[X,Y]$.
It follows that $G[X,Y]$ is isomorphic either to $2K_{2,2}$ or to $K_{4,4}$.
The latter option is actually impossible. Indeed, assume that $G[X,Y]\cong K_{4,4}$
and let $\cct_{-}$ be obtained from \cct by making the interspace $\cct[X,Y]$
uniform. Note that
$$
\cct\prec\cct_{-}\preccurlyeq\calR_G,
$$
and that $\cct_{-}$ is a coherent configuration (because every subconfiguration
$\cct_{-}[A\cup B\cup C]$ for $A,B,C\in F(\cct_{-})$ is obviously coherent).
This contradicts Proposition~\ref{prop:closure}.

Furthermore, the induced subgraph $G[X]$ for each vertex color class $X$ must
be regular, that is, $G[X]$ is either empty or isomorphic to $K_4$, $C_4$, or $2K_2$.
We conclude that every 16-vertex $G$ of color multiplicity 4 with $\ccg G=\cct$
is one of the graphs obtainable in the following way:
\begin{enumerate}[Step 1.]
\item
Take four disjoint 4-vertex sets $X_1,X_2,X_3,X_4$ and color them in four colors so that each $X_i$
is monochromatic.
\item
For each two indices $i$ and $j$, connect $X_i$ and $X_j$ with 8 edges so that
$G[X_i,X_j]\cong 2K_{2,2}$. Note that $G[X_i,X_j]$ determines a matching on $X_i$
and a matching on $X_j$. It is required that, for every $i$, the three subgraphs
$G[X_i,X_j]$ for $j\in\{1,2,3,4\}\setminus\{i\}$ determine three pairwise distinct
matchings on~$X_i$.
\item 
For each $i$, either leave $G[X_i]$ empty or plant a $K_4$-, or $C_4$-, or $2K_2$-subgraph
on~$X_i$. 
\end{enumerate}

We claim that $\ccg G\cong\cct$ for every graph $G$ obtainable as described above.
Due to Step 2, the edges of $G$ between the classes $X_1,X_2,X_3,X_4$ represent 
a coherent configuration \cct with fibers $X_1,X_2,X_3,X_4$. Since any $K_4$-, $C_4$-, 
or $2K_2$-subgraph planted in Step 3 can be split into matching, we have
\begin{equation}\label{eq:trg}
\cct\preccurlyeq\calR_G. 
\end{equation}

On the other hand, note that the vertex coloring of $G$ determines the partition
of the diagonal $\setdef{vv}{v\in V(G)}$ into four reflexive relations, each of size 4,
corresponding to the fibers of \cct.
Each $G[X_i,X_j]$ determines the partition of $X_i\times X_j$ of $2K_{2,2}$-type
corresponding to the interspace $\cct[X_i,X_j]$.
Moreover, each $G[X_i,X_j]$ determines a matching on $X_i$
in the sense of Part 2 of Lemma \ref{lem:X2Y4}. Due to Step~2, this determines
the $F_4$-factorization of each $(X_i)^2$. In terms of partitions, these observations
can be stated as
\begin{equation}\label{eq:cgt}
\ccg G\preccurlyeq\cct. 
\end{equation}
Relations \refeq{trg} and \refeq{cgt} imply by Proposition~\ref{prop:closure} that $\ccg G=\cct$.

Thus, we have described the family of all $G$ of color multiplicity 4 with $\ccg G\cong\cct$.
However, following Steps 1--3, we can generate the same, up to isomorphism, graph in many
different ways. Now, we want to count the number of such graphs up to isomorphism
and color renaming. As we already know that every non-trivial $\eqq$-equivalence class consists
of two non-isomorphic graphs, we can count the number of these classes and multiply it by~2.

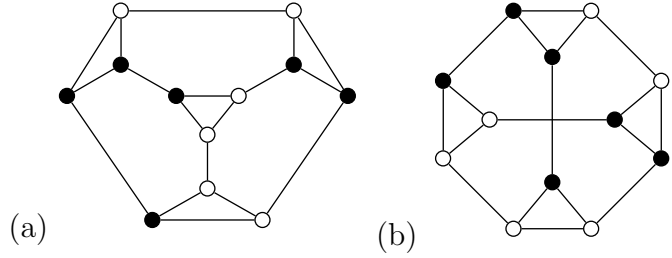
\begin{figure}
\centering
\begin{tikzpicture}[every node/.style={circle,draw,inner sep=2pt},
lab/.style={draw=none,fill=none,inner sep=0pt,rectangle},
line width=0.5pt,
wt/.style={fill=white},
vt/.style={line width=1.6pt}]
\matrixgraph[name=m1,nolabel]{
  &[5mm]&[2mm]&[1mm]&[2mm]&[2mm]&[1mm]&[2mm]&[5mm]\\
    &1[wt]&    &    &     &     &     &4[wt]   \\[5mm]
    & 2   &    &    &     &     &     & 5      \\[2mm]
  3 &     &    & 7  &     &8[wt]&     &     & 6\\[3mm]
    &     &    &    &9[wt]                     \\[5mm]
    &     &    &    &A[wt]                     \\[2mm]
    &     & B  &    &     &     &C[wt]         \\
  }{
  1 -- 2 -- 3 -- 1;
  4 -- 5 -- 6 -- 4;
  7 -- 8 -- 9 -- 7;
  A -- B -- C -- A;
  {1,2,3,8,9,C} --[matching] {4,7,B,5,A,6};
}
\matrixgraph[name=m2,nolabel,matrix anchor=north west,
  at={($(m1.north west)+(50mm,0mm)$)}]{
  &[4mm]&[1mm]&[3mm]&[3mm]&[1mm]&[4mm]\\
       &     &  1  &    &2[wt]&           \\[4mm]
       &     &     & 3  &     &           \\[1mm]
  4    &     &     &    &     &     &7[wt]\\[3mm]
       &5[wt]&     &    &     & 8         \\[3mm]
  6[wt]&     &     &    &     &     & 9   \\[1mm]
       &     &     & A  &     &           \\[4mm]
       &     &B[wt]&    &C[wt]&           \\
  }{
  1 -- 2 -- 3 -- 1;
  4 -- 5 -- 6 -- 4;
  7 -- 8 -- 9 -- 7;
  A -- B -- C -- A;
  {1,2,3,8,9,B} --[matching] {4,7,A,5,C,6};
}
\path 
  ($(m1.south west)-(4mm,0mm)$) node[lab] {(a)}
  ($(m2.south west)-(5mm,0mm)$) node[lab] {(b)}
;
\end{tikzpicture}
\caption{A colored truncated tetrahedral graph (two looks).}
\label{fig:tetr}
\end{figure}

Suppose that $G$ is obtained according to Steps 1--3. With $G$ we associate
a truncated tetrahedral graph $T_G$ whose vertices are colored black or white; 
see Figure \ref{fig:tetr}(a).
Denote the three matchings on the vertex color class $X_i$ of $G$ by $L_i$, $M_i$, and $N_i$.
The vertex set of $T_G$ is $\Set{L_i,M_i,N_i}_{i=1}^4$. Each triple $\{L_i,M_i,N_i\}$
forms a 3-clique. Moreover, a vertex $M\in\{L_i,M_i,N_i\}$ is adjacent to a vertex
$M'\in\{L_j,M_j,N_j\}$ whenever $G[X_i,X_j]$ determines $M$ on $X_i$ and $M'$ on $X_j$.
The edge set of $T_G$ is therewith defined. A vertex $M\in\{L_i,M_i,N_i\}$ is colored
black if $M$ is covered by the adjacency relation of $G[X_i]$ and it is colored white
otherwise. Thus, the clique $\{L_i,M_i,N_i\}$ contains 3, 2, 1 black vertices
exactly when $G[X_i]$ is isomorphic to $K_4$, $C_4$, $2K_2$ respectively.
All vertices in $\{L_i,M_i,N_i\}$ are white exactly when $X_i$ is an independent
set in~$G$.

\begin{claim}\label{cl:3}
$G$ and $H$ are $\eqq$-equivalent up to color renaming if and only if $T_G\cong T_H$.
\end{claim}

\begin{subproof}
Denote $\ccc=\ccg G$ and $\ccd=\ccg H$. Note that $T_G$ does not change after renaming
the colors in $G$. To prove the claim in the forward direction, we can therefore
assume that $G\eqq H$. By Claim \ref{cl:1}, there is an algebraic isomorphism $f\function{\ccc}{\ccd}$
such that $H=G^f$. Each of the matching relations $L_i$, $M_i$, and $N_i$ is in \ccc
and, therefore, the restriction of $f$ to the matching relations can be seen as
a bijection from $V(T_G)$ onto $V(T_H)$. Since the algebraic isomorphism $f$ takes the fibers
of \ccc to the fibers of \ccd, $f$ takes the 3-cliques of $T_G$ to the 3-cliques of $T_H$.
Since $f$ preserves the relation ``an interspace $I$ determines a matching $M$'', $f$
takes the remaining edges of $T_G$ to the remaining edges of $T_H$. Finally,
$f$ preserves the black-white coloring because $H=G^f$. Therefore, $f$ provides
an isomorphism from $T_G$ to~$T_H$.

Assume now that $T_G\cong T_H$. Given an isomorphism $\alpha$ from $T_G$ to $T_H$,
we construct an algebraic isomorphism $f\function{\ccc}{\ccd}$ such that $H=G^f$,
possibly after appropriately renaming the colors of vertices in $H$.
By Claim \ref{cl:1}, this will imply that $G\eqq H$ up to color renaming.

We first define $f$ on the set of the matching relations in \ccc
just by setting $f(M)=\alpha(M)$ for each of the 12 matchings.
The isomorphism $\alpha$ takes the 3-cliques of $T_G$ to the 3-cliques of $T_H$.
Therefore, if $M$ and $M'$ are matchings on the same fiber of \ccc, then 
$f(M)$ and $f(M')$ are matchings on the same fiber of \ccd.
Using this, we can consistently extend $f$ to a bijection from the set of the reflexive
relations of \ccc to the reflexive relations of \ccd. For a fiber $X_i$ of \ccc,
we will denote the corresponding fiber of \ccd by $f(X_i)$.
Renaming the vertex colors in $H$ if necessary, we can ensure that the vertex color classes
$X_i$ in $G$ and $f(X_i)$ in $H$ are equally colored.
Finally, for each pair of fibers $X_i,X_j\in F(\ccc)$, we define $f$ locally as a bijection from
the interspace $\ccc[X_i,X_j]$ onto the interspace $\ccd[f(X_i),f(X_j)]$ in the following way:
The element of $\ccc[X_i,X_j]$ corresponding to the adjacency relation of $G[X_i,X_j]$
is taken by $f$ to the element of $\ccd[f(X_i),f(X_j)]$ corresponding to the adjacency relation 
of $H[f(X_i),f(X_j)]$. Thus, we have defined a bijection $f$ from \ccc onto \ccd.
Since $\alpha$ preserves adjacency between vertices in different 3-cliques,
$f$ preserves the relation ``an interspace $I$ determines a matching $M$''.
This implies that $f$ is an algebraic isomorphism from \ccc onto~\ccd.

It remains to argue that $H=G^f$. Suppose that a relation $R\in\ccc$ is
covered by the adjacency relation of $G$. We have to verify that the relation
$R^f\in\ccd$ is covered by the adjacency relation of $H$. If $R$ belongs to an
interspace of \ccc, this follows directly from the definition of $f$.
If $R$ belongs to a cell of \ccc, that is, is a matching relation, this
follows from the fact that $\alpha$ preserves the black-white vertex coloring.
It remains to note that the correspondence between the reflexive color relations
in $G$ and $H$ was secured by color renaming.
\end{subproof}

Claim \ref{cl:3} reduces our task to counting the non-isomorphic black-white
colorings of the truncated tetrahedral graph $T$. We use the Pólya enumeration theorem.
Let $\mathrm{Aut}(T)$ be the subgroup of the symmetric group $S_{12}$ consisting
of the automorphisms of $T$. Every automorphism of $T$ takes a 3-clique to a 3-clique,
determining a permutation of the 4-element set of all 3-cliques.
This actually yields a one-to-one correspondence between $\mathrm{Aut}(T)$ and $S_4$
in accordance with the fact that the regular tetrahedron has the same isometries
as its truncated version. As a consequence, the permutations in $\mathrm{Aut}(T)$,
like in $S_4$, are split into 5 conjugacy classes. One class consists of the identity
permutation, which in $\mathrm{Aut}(T)$ has 12 cycles. The six transpositions in $S_4$
correspond to reflections in a plane, which for the truncated tetrahedron give six
permutations with 7 cycles each, seen in Figure \ref{fig:tetr}(a) as reflections in a line.
The eight 3-cycles in $S_4$ correspond to axial rotations, which for the truncated tetrahedron 
give eight permutations with 4 cycles each, seen as rotations in Figure \ref{fig:tetr}(a).
The six 4-cycles in $S_4$ correspond to six permutations in $\mathrm{Aut}(T)$
with 3 cycles each, which can be seen as rotations in Figure \ref{fig:tetr}(b). 
Finally, three products of two transpositions in $S_4$ correspond to three permutations in $\mathrm{Aut}(T)$
with 6 cycles each, which appear in Figure \ref{fig:tetr}(b) as the inversion in the central point.
By the Pólya enumeration theorem, the number of non-isomorphic ways to color
the vertices of $T$ in $n$ colors is equal to
$$
p(n)=\frac1{24}\of{n^{12}+6\,n^7+3\,n^6+8\,n^4+6\,n^3}.
$$
The number of different colorings in two colors is, therefore, equal to
$p(2)=218$.
\end{proof}

\begin{figure}[t]
\begin{center}
\begin{tikzpicture}[every node/.style={
    circle,draw,fill,inner sep=2pt},  
elabel/.style={black,draw=none,fill=none,rectangle},
  xcol/.style={black!10!cyan!80},
  zcol/.style={black!10!yellow!80},
  ucol/.style={black!10!green!80},
  ycol/.style={magenta!90}]
  \graph[empty nodes,branch down=12mm,grow right=12mm] {
    s00[ycol] -- s01[xcol] -- s02[ycol] -- s03[xcol];
    s10[zcol] -- s11[ucol] -- s12[zcol] -- s13[ucol];
    s20[ycol] -- s21[xcol] -- s22[ycol] -- s23[xcol];
    s30[zcol] -- s31[ucol] -- s32[zcol] -- s33[ucol];
    s00 --[bend left=40] { s03};
    s10 --[bend left=40] { s13};
    s20 --[bend right=40] { s23};
    s30 --[bend right=40] { s33};
    s00 -- s10 -- s20 -- s30;
    s01 -- s11 -- s21 -- s31;
    s02 -- s12 -- s22 -- s32;
    s03 -- s13 -- s23 -- s33;
    s00 --[dashed] s11 --[dashed] s22 --[dashed] s33 --[bend 
      right=20,looseness=0.75,dashed] s00;
    s01 --[dashed] s12 --[dashed] s23[dashed] --[dashed] s30 --
      [dashed] s01;
    s02 --[dashed] s13 --[dashed] s20 --[dashed] s31 --[dashed] s02;
    s03 --[dashed] s10 --[dashed] s21 --[dashed] s32 --[dashed] s03;
    s00 --[bend right=40] { s30};
    s01 --[bend right=40] { s31};
    s02 --[bend left=40] { s32};
    s03 --[bend left=40] { s33};
  };
  \node[name=a,elabel,at={($(s30.south west)+(-8mm,0mm)$)}] {(a)};
\begin{scope}[xshift=7cm]
  \graph[empty nodes,branch down=12mm,grow right=12mm] {
    r00[ycol] -- r01[xcol] -- r02[zcol] -- r03[ucol];
    r10[xcol] -- r11[ycol] -- r12[ucol] -- r13[zcol];
    r20[zcol] -- r21[ucol] -- r22[ycol] -- r23[xcol];
    r30[ucol] -- r31[zcol] -- r32[xcol] -- r33[ycol];
    r00 --[bend left=40] { r03};
    r10 --[bend left=40] { r13};
    r20 --[bend right=40] { r23};
    r30 --[bend right=40] { r33};
    r00 --[bend left=30,dashed] { r02};
    r01 --[bend left=30,dashed] r03;
    r10 --[bend left=30,dashed] { r12};
    r11 --[bend left=30,dashed] r13;
    r20 --[bend right=30,dashed] { r22};
    r21 --[bend right=30,dashed] r23;
    r30 --[bend right=30,dashed] { r32};
    r31 --[bend right=30,dashed] r33;
    r00 -- r10 -- r20 -- r30;
    r01 -- r11 -- r21 -- r31;
    r02 -- r12 -- r22 -- r32;
    r03 -- r13 -- r23 -- r33;
    r00 --[bend right=40] { r30};
    r01 --[bend right=40] { r31};
    r02 --[bend left=40] { r32};
    r03 --[bend left=40] { r33};
    r00 --[bend right=30,dashed] { r20};
    r10 --[bend right=30,dashed] r30;
    r01 --[bend right=30,dashed] { r21};
    r11 --[bend right=30,dashed] r31;
    r02 --[bend left=30,dashed] { r22};
    r12 --[bend left=30,dashed] r32;
    r03 --[bend left=30,dashed] { r23};
    r13 --[bend left=30,dashed] r33;
  };
  \node[name=b,elabel,at={($(r30.south west)+(-8mm,0mm)$)}] {(b)};
\end{scope}
\end{tikzpicture}
\end{center}
\caption{An example of \WL2-equivalent non-isomorphic colored graphs:
(a) The Shrikhande graph; (b) the $4\times4$ rook's graph.}
\label{fig:gallery}
\end{figure}
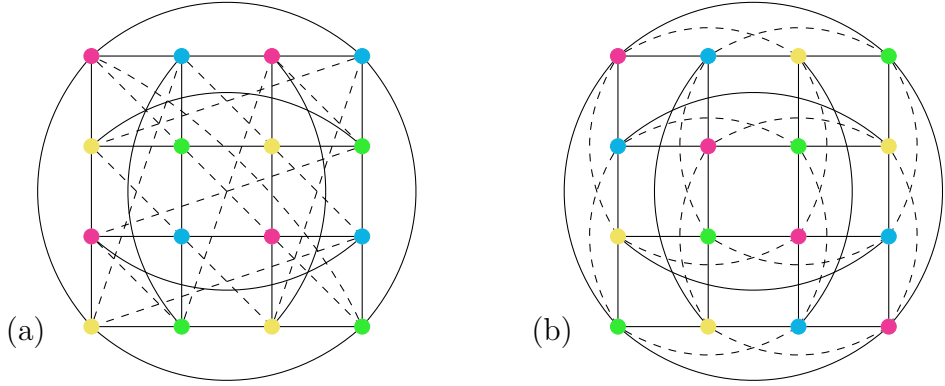

To illustrate Part 2 of Theorem \ref{thm:graphs16}, we show two $\eqq$-equivalent
and still non-isomorphic graphs produced as in the proof, with all color classes
left empty in Step 3. Remarkably, these are colored versions of two smallest
non-isomorphic strongly regular graphs with the same parameters, namely 
the Shrikhande graph and the $4\times4$ rook's graph both having parameters (16,6,2,2);
see Figure \ref{fig:gallery}(a)--(b).

\section{Further questions}\label{s:concl}

Our results raise questions about the parameterized complexity of recognizing
amenability of a given graph with the color multiplicity $m$
taken as the parameter. The problem is trivial for $m=3$ due to \cite{ImmermanL90}.
We show that it is solvable in polynomial time for $m=4$. Our analysis surely generalizes
to a few subsequent values of $m$. For any fixed $m$, the problem is in $\mathrm{coNP}$,
and it is open whether it is in $\mathrm{P}$ if $m$ is large.

Another open question, that naturally arises in light of Theorem \ref{thm:amen4},
concerns the next dimension of the Weisfeiler-Leman algorithm: Can the
amenability to \WL3 be decided in polynomial time on input graphs
of color multiplicity~4?

The \emph{WL dimension} of a graph $G$ is defined as the minimum $k$ such that $G$
is amenable to \WL k. The graphs with large WL dimension are of significant interest
in the study of the graph isomorphism problem. When we seek such graphs among 
graphs with color multiplicity 4, note that they must be at least non-amenable to \wl.
Cai, Fürer, and Immerman \cite{CaiFI92} give conditions ensuring linear WL dimension
for graphs whose coherent closure is, in our terminology, skew-connected.
Further such conditions are identified by the line of 
research \cite{DawarK19,GurevichS96,NeuenS17,NeuenS18}.
Can we achieve high WL dimension in other cases, say, for graphs whose coherent closure
corresponds to a line-point $(n_3)$-configuration?

\subsection*{Acknowledgment}

We thank Ilia Ponomarenko and the anonymous referees of the STACS'20 proceedings for their numerous
detailed comments and Daniel Neuen and Pascal Schweitzer for a useful discussion
of multipede graphs. 
The third author is especially grateful to Ilia Ponomarenko for 
his patient and insightful guidance through the theory of coherent configurations.

\end{document}